%% file: arxiv.tex
\documentclass[11pt]{article}
\input{prelim}

\title{Product Monte Carlo: Expectation Under a Product of Densities}
\author{
Pietro Carlotti\footnote{University of Texas at Austin, Department of Statistics and Data Sciences}    \and
Sifan Liu\footnote{Duke University, Department of Statistical Science}  \and
Yuling Yao\footnotemark[1]
}
\date{\today}

\begin{document}\sloppy
	\maketitle
	\thispagestyle{empty}
	\begin{abstract}
Aggregating independently computed sub-sample inferences is a well-known challenge in large-scale Bayesian computation, federated learning, and opinion pooling. The core problem is to approximate a distribution whose density is proportional to a product of individual densities. We propose Product Monte Carlo (PMC), an algorithm for computing arbitrary expectations under a product of densities using only Monte Carlo draws from individual densities. We turn the intractable product expectation into a regression problem by expressing it through a family of computable conditional expectations. Besides the point estimate, our fully Bayesian regression formulation provides coherent uncertainty quantification and principled diagnostics. We establish theoretical guarantees on consistency and characterize convergence rates for the PMC estimator. Across simulated and real-data experiments varying the number, dimensionality, and heterogeneity of the component distributions, PMC achieves higher estimation accuracy, better-calibrated uncertainty coverage, and improved scalability at a lower computational cost than existing methods. In addition, PMC requires no density evaluation, making it applicable to simulation-based inference, as we demonstrate in an application with protein images.

\vspace{.1in}
\end{abstract}

\begingroup
\renewcommand\thefootnote{}
\footnote{We thank Jinlin Lai for help with our SBI experiment. Code and reproducibility materials are available at \url{https://github.com/PietroCarlotti/product-monte-carlo}.}
\addtocounter{footnote}{-1}
\endgroup
\input{sections/01_introduction}
\input{sections/02_PMC}
\input{sections/03_asymptotics}
\input{sections/04_experiments}
\input{sections/05_discussion}

 \bibliographystyle{apalike}

\bibliography{references}

\clearpage
\appendix

\section*{Product Monte Carlo:  Supplementary Materials}
\input{sections/appendix}
\end{document}

%% file: prelim.tex
\usepackage[utf8]{inputenc}
\usepackage{amsfonts,epsfig}
\usepackage{stackengine}
\usepackage{relsize}
\usepackage{fancyvrb}
\usepackage{geometry}
\usepackage[font={it}]{caption}
\usepackage{subcaption}
\usepackage[normalem]{ulem}
\usepackage{titlesec}
\titleformat{\section}
  {\sffamily\bfseries\upshape\large\raggedright}
  {\thesection.}{0.6em}{}
\titleformat{\subsection}
  {\sffamily\bfseries\upshape\normalsize\raggedright}
  {\thesubsection.}{0.5em}{}
\titleformat{\subsubsection}
  {\sffamily\bfseries\upshape\normalsize}
  {\thesubsubsection.}{0.5em}{}
\titleformat{\paragraph}[runin]
  {\sffamily\mdseries\itshape\normalsize}
  {\theparagraph.}{0.4em}{}[\ignorespaces]

\usepackage[T1]{fontenc}
\usepackage[defaultsans]{lato}
\usepackage{inconsolata}
\usepackage{hyperref}
\hypersetup{
	colorlinks=true,
	linkcolor=black,
	citecolor=black,
	filecolor=black,
	urlcolor=black,
}
\makeatletter
\def\@maketitle{%
	\begin{center}%
		\let \footnote \thanks
		{\large \sffamily\fontseries{eb}\selectfont \@title \vspace{.125in}\par}
		{\normalsize
			\begin{tabular}[t]{c}%
				\@author
			\end{tabular}\par}%
             \vspace{0.1in}%
		{\small \@date}%
         \vspace{-0.125in}%
	\end{center}%
}
\makeatother

 \usepackage[english]{babel}
\usepackage[
    autostyle,
    english=american
]{csquotes}
\usepackage{booktabs}
\usepackage{amsmath}
\usepackage{amssymb}
\usepackage{amsthm}
\usepackage{mathtools}
\usepackage{tikz}
\usetikzlibrary{positioning, fit, arrows.meta, calc}
\usepackage{xcolor}
\usepackage{graphicx}
\usepackage{subcaption}
\usepackage{flafter}
\usepackage{placeins}
\usepackage{multirow}
\usepackage{algpseudocode}
\usepackage{algorithm}
\usepackage{algorithmicx}
\usepackage{enumitem}
\usepackage[round]{natbib}
\MakeOuterQuote{"}          
\allowdisplaybreaks         
\definecolor{inputColor}{RGB}{0,114,178}     
\definecolor{outputColor}{RGB}{213,94,0}     
\graphicspath{{figures/}} 

\newcommand{\R}{\mathbb{R}} 

\newcommand{\E}{\mathbb{E}}              

\newcommand{\pipro}{\pi_{\ast}}              

\newcommand{\Prob}{\mathbb{P}}           

\newcommand{\iid}{\overset{\textnormal{i.i.d.}}{\sim}} 
\newcommand{\indep}{\perp \!\!\! \perp} 
\newcommand{\weakto}{\overset{\mathrm{w}}{\to}} 

\newcommand{\convergencefiguretarget}[6]{
    \begin{figure*}[htbp]
        \def\convergencefunction{#3}
        \centering
        \includegraphics{#1/h_#2.pdf}
        \caption{#6}
        \label{fig:#5}
    \end{figure*}
}

\newcommand{\convergencefigure}[5][u^2]{
    \convergencefiguretarget{#2}{identity}{u}{#3}{#4}{#5}
    \convergencefiguretarget{#2}{square}{#1}{#3}{#4_square}{#5}
}

\newsavebox\singlebox

\newcommand{\singlefigure}[3]{
    \begin{figure*}[htbp]
        \centering
        \sbox\singlebox{\includegraphics{#1}}
        \ifdim\wd\singlebox>\linewidth
            \resizebox{\linewidth}{!}{\usebox\singlebox}
        \else
            \usebox\singlebox
        \fi
        \caption{#3}
        \label{fig:#2}
    \end{figure*}
}

\theoremstyle{plain}
\newtheorem{theorem}{Theorem}
\newtheorem{proposition}{Proposition}

\theoremstyle{definition}
\newtheorem{definition}{Definition}

\theoremstyle{remark}

\newtheorem{assumption}{Assumption}[section]

%% file: sections/01_introduction.tex
\section{Introduction}\label{sec:introduction}
Many statistical problems involve computing expectations under the product of densities. For illustration, we begin with an example from Bayesian subsampling: the data $\mathcal{D}$ are conditionally i.i.d. given a parameter $\theta \in \R^p$ and can be partitioned into $K \geq 2$ disjoint subsets $\mathcal{D}_{1}, \ldots, \mathcal{D}_{K}$. If we denote the prior distribution $\rho_{0}(\theta)$ and the likelihood function $ \pi(\mathcal{D} \mid \theta)$, the full posterior distribution of $\theta$ given the data $\mathcal{D}$ is
\[
    \pi(\theta \mid \mathcal{D}) \propto \rho_{0}(\theta) \prod_{k=1}^K  \pi(\mathcal{D}_k \mid \theta).
\]
The subposterior is the inference based on subset data:
\[
    \rho_k(\theta) \propto \rho_{0}(\theta)^{\frac{1}{K}} \,  \pi(\mathcal{D}_k \mid \theta) \quad k = 1, \ldots, K,
\]
where the prior is scaled to reflect the partition. We can then obtain $S$ samples from each subposterior, which often is easier than the full sampling:
\[
    \theta_{k}^{(s)} \sim \rho_k \quad s=1,\ldots,S, \quad k=1,\ldots,K.
\]
However, even if the subsamples are readily available, computing quantities under the full-data posterior remains nontrivial. The full posterior density is proportional to the \emph{product} of subposteriors:
\[
     \pi(\theta \mid \mathcal{D}) \propto \prod_{k=1}^K \rho_k(\theta).
\]
To solve this aggregation task, this paper develops \emph{Product Monte Carlo} (PMC), a framework for estimating expectations under a \emph{product density}, formally:
\begin{definition}[Product Distribution]
\label{def:product_distribution}
Let $\pi_1,\ldots,\pi_K$ be densities of distributions on $\mathbb{R}^p$, where $K \geq 2$.
The \emph{product density} $\pipro$ is defined as follows:
\[
     \pipro(x) = \frac{1}{Z}\prod_{k=1}^{K} \pi_{k}(x) \propto \prod_{k=1}^{K} \pi_{k}(x)  \qquad \forall \; x \in \mathbb{R}^p,
\]
with an unknown normalizing constant
\[
    Z=\int_{\mathbb{R}^p}\prod_{k=1}^K\pi_k(x)\,\mathrm{d}x,
    \qquad 0<Z<\infty.
\]
\end{definition}
Throughout this paper, ``\emph{product}'' refers to the pointwise product of component densities on $\R^p$, and should not be confused with a product measure on \((\mathbb{R}^p)^K\).

As in other Monte Carlo methods, our goal is to compute the expectation of any integrable function $h:\mathbb{R}^p \to \mathbb{R}^r$, with $r\geq1$, under the product density:
\[
\mu_h =
\E_{\pipro}\left[ h \right] =
\int_{\mathbb{R}^p} h(x) \, \pipro(x) \, \mathrm{d}x .
\]
In the subsampling example, the components $\pi_k$ are the subposterior densities, and their product, after normalization, is the full-data posterior. Taking $h(\theta)=\theta$ or $h(\theta)=\theta\theta^\top$ yields the first or second moment under the full posterior.

\paragraph{Heuristic of PMC.}
First draw $S$ samples from each component density, independently across components:
\[
    x_k^{(s)} \sim \pi_k
    \quad s=1,\ldots,S,\quad k=1,\ldots,K.
\]
PMC aggregates individual samples via a regression to estimate any expectation $\E_{\pipro}[h]$ under the product distribution. See Figure~\ref{fig:pmc_diagram} for the workflow.
\begin{figure}[t]
\centering
\resizebox{\linewidth}{!}{%
\input{sections/pmc_diagram}
}
\caption{PMC aggregates samples from individual components to estimate any expectation under the product distribution. No density evaluation of $\pi_k$ is needed.
}
\label{fig:pmc_diagram}
\end{figure}
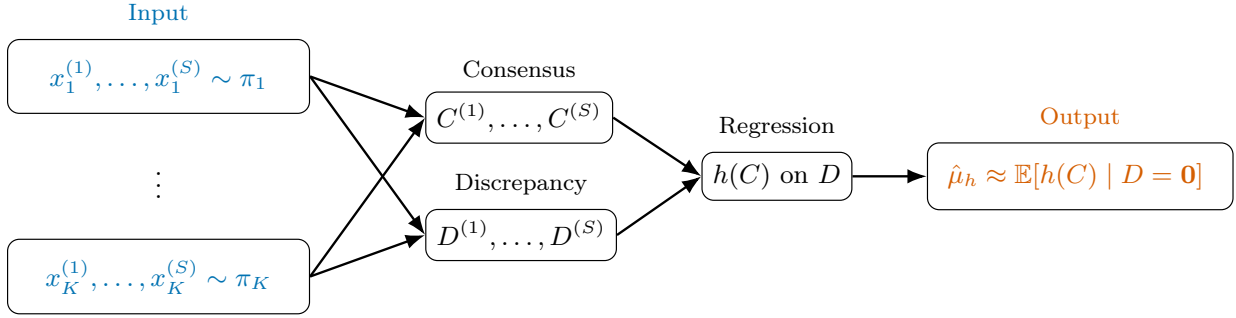

To build intuition for how PMC works, consider two independent random variables in $\R^p$:
$x_1 \sim \pi_1$, and $x_2 \sim \pi_2.
$
Define their scaled sum and difference by
\[
    C = \frac{x_1 + x_2}{2}, \quad D = \frac{x_1 - x_2}{2}.
\]
The conditional distribution of $C$ given $D = \mathbf{0}_{p}$ is exactly the target product distribution:
\[
    C \mid D = \mathbf{0}_{p} \sim \pipro \propto \pi_1 \pi_2.
\]
Intuitively, the product distribution is the slice of the joint distribution along the diagonal of the product space. Figure~\ref{fig:product_distribution} illustrates a Gaussian example.
\begin{figure}[t]
    \centering
         \vspace{-0.8em}
    \includegraphics[width=0.7\linewidth]{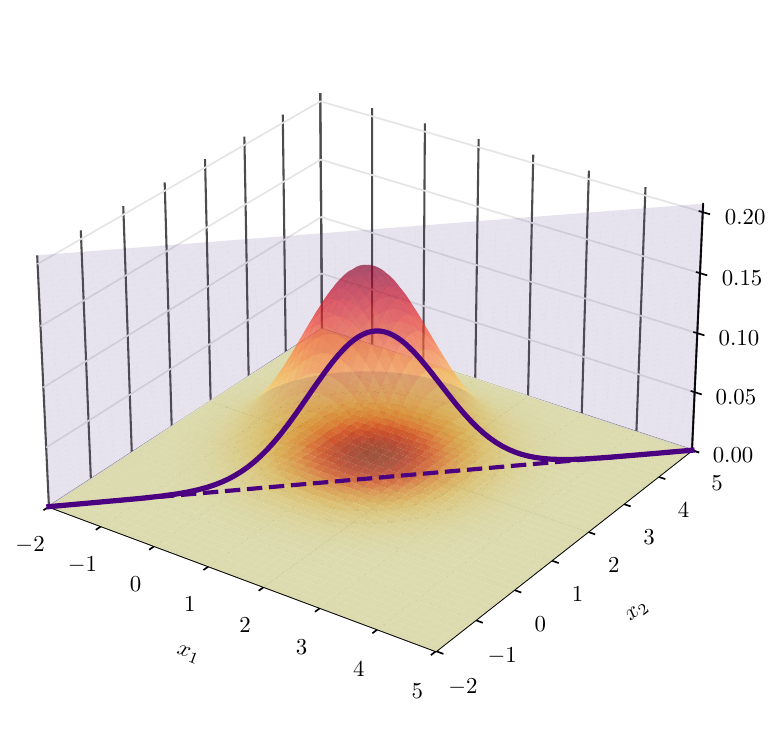}
     \vspace{-0.8em}
    \caption{For $\pi_1=\mathcal{N}(1,1)$ and $\pi_2=\mathcal{N}(2,1)$, the product distribution $\pipro=\mathcal{N}(1.5,0.5) \propto \pi_1\pi_2$ matches the conditional distribution on the diagonal slice.}
    \label{fig:product_distribution}
    \vspace{-0.5em}
\end{figure}

The above identity belongs to a broader family of conditional representations, which we develop in the rest of the paper, where $C$ captures \emph{consensus} among the component variables and $D$ measures their \emph{discrepancy}. This lets us express expectations under the product distribution as conditional expectations:
\[
    \E_{\pipro}[h] = \E[h(C)\mid D = \mathbf{0}].
\]
Although the event $\{D = \mathbf{0}\}$ has zero probability, we can estimate the conditional expectation through regression. PMC develops a general recipe for these representations and their fully Bayesian estimation.

As we will show through theoretical analysis and empirical evaluation in the subsequent sections, the salient features of PMC include the following:
\begin{itemize}[leftmargin=*,topsep=2pt,itemsep=1pt,parsep=0pt]
\item \emph{Distribution-free:} Without parametric restrictions, PMC allows arbitrary component densities $\pi_k$.
\item \emph{Evaluation-free:} PMC requires only samples from the component distributions, not evaluations of their densities or product.  This makes it suitable for simulation-based and likelihood-free inference.
\item \emph{Scalable:} PMC allows parallel sampling across components without coordination and reduces intractable product sampling to regression, which scales better with dimension and sample size.
\item \emph{Calibrated:} Our fully Bayesian formulation in PMC provides coherent uncertainty quantification and credible intervals, accounting for both the Monte Carlo and epistemic errors.
\item \emph{Diagnosable:} PMC turns a computing task into a modeling problem, which incorporates the established Bayesian workflow and supports model diagnostics, checks, and iterative model improvement.
\end{itemize}

\paragraph{Applications.}

Beyond its interest as a general problem, computing expectations under products of densities is useful in several practical settings.

As in our motivating example of \emph{large-scale inference}, partitioning a large dataset and sampling from subposteriors substantially reduces the cost: each subposterior requires a fraction of the likelihood and gradient evaluation cost; subposteriors can be sampled in parallel \citep{neiswanger2014asymptotically,scott2016bayes}; and amortized methods such as neural posterior estimation \citep[e.g.,][]{papamakarios2016fast} can make repeated inference inexpensive after training. The remaining task is then to aggregate them into the product. This product form arises naturally under conditional independence, as multiplication combines information sources \citep{bernardo1994bayesian}.

In \emph{federated learning} \citep[e.g.,][]{al-shedivat2021federated,liu2024bayesian,jonker2024bayesian} or privacy-sensitive applications, data cannot be centralized for storage and processing. Inference must use data held locally across different sites, and then multiply individual densities.

Products also arise in \emph{ensemble prediction}, where \emph{product-of-experts} and \emph{logarithmic pooling} combine predictive densities multiplicatively \citep{genest1986combining, hinton2002training}. An open challenge is then to evaluate the combined predictive distribution.

In many of these applications, individual component densities $\pi_k$ cannot be evaluated: local densities in privacy settings can leak data information; posterior predictive densities in expert pooling require additional integration; and generative-model-based samplers and simulation-based inference (SBI) methods provide samples without tractable density evaluations.

\subsection{Related methods}

Given its broad applicability, several strategies address inference under product distributions. We review three main approaches and compare them in Table~\ref{tab:comparison}.

\begin{table}[t]
\footnotesize
\centering
\setlength{\tabcolsep}{2pt}
\renewcommand{\arraystretch}{1.0}
\begin{tabular}{@{}cccc@{}}
\toprule
Method & \begin{tabular}{@{}c@{}}Distribution-free\end{tabular} & \begin{tabular}{@{}c@{}}Evaluation-free\end{tabular} & \begin{tabular}{@{}c@{}}UQ\end{tabular} \\
\midrule
\begin{tabular}{@{}c@{}}Parametric approx.\end{tabular} & & $\checkmark$ & \\ \addlinespace[4pt]
\begin{tabular}{@{}c@{}}Weighting\end{tabular} & $\checkmark$ & & \\ \addlinespace[4pt]
\begin{tabular}{@{}c@{}}Fusion \end{tabular} & $\checkmark$ & & \\ \addlinespace[4pt]
\textbf{PMC} & $\checkmark$ & $\checkmark$ & $\checkmark$ \\
\bottomrule
\end{tabular}
\caption{Comparison of available methods and PMC.}
\label{tab:comparison}
\end{table}

\paragraph{Parametric approximation.} It is natural to approximate each component with a parametric family that keeps the product distribution analytical. Exploiting the fact that the product of Gaussians is itself Gaussian,
\citet{huang2005sampling} used Gaussian approximations for large-scale Bayesian computation, Consensus Monte Carlo \citep{scott2016bayes} aggregated subposterior draws through precision-weighting, and \citet{carvalho2023bayesian} applied similar approximations to logarithmic pooling.
Moving beyond a single Gaussian approximation, \citet{neiswanger2014asymptotically} combined kernel-based estimates of the component densities, and \citet{cai2025fisher} derived an analytic representation for products of Student-$t$ densities using the Feynman parameterization. However, parametric approximations rely on strong assumptions and thus introduce additional bias that persists even with large Monte Carlo samples. In comparison, PMC imposes no parametric assumptions, yet our theory shows that PMC recovers parametric efficiency in the Gaussian case.

\paragraph{Importance weighting.}
Subposterior draws can be reweighted to approximate the product distribution \citep{mesquita2020embarrassingly}. Although simple to implement and free of parametric assumptions, importance weighting requires density evaluations and can yield unstable estimates with high or even infinite variance when importance ratios have heavy tails \citep{vehtari2024pareto}. This issue may be especially pronounced here because the product distribution is typically more concentrated than its component distributions.

\paragraph{Fusion methods.}

An alternative approach is to construct a sampler targeting the product distribution rather than estimating its expectations directly. To this end, fusion methods, including Monte Carlo Fusion \citep[MCF;][]{dai_monte_2019}, Bayesian Fusion \citep[BF;][]{dai_bayesian_2023}, and Generalized Bayesian Fusion \citep[GBF;][]{chan_divide-and-conquer_2023}, initialize $K$ stochastic processes at subposterior draws, constructed to coalesce at a common endpoint distributed according to the product target. Like PMC, fusion methods do not restrict component densities to specific parametric families. However, simulating the resulting $Kp$-dimensional stochastic processes $S$ times is itself a nontrivial sampling problem, which incurs heavy computational overhead for large $K$, $S$, or $p$. Moreover, fusion methods require evaluating component densities and their derivatives, thus limiting their applicability if they are unavailable or computationally expensive.

\paragraph{Uncertainty quantification.}
While many existing approaches can return a Monte Carlo error, they typically cannot assess epistemic uncertainty. 
For example, when component distributions are well-separated, estimating their product expectation is intrinsically hard.
PMC is designed to reflect this difficulty through larger posterior uncertainty, whereas other methods yield vanishingly small uncertainty as the Monte Carlo size increases, even for inaccurate estimates.

The remainder of this paper is organized as follows. Section~\ref{sec:PMC} introduces PMC for estimating expectations under products of densities. Section~\ref{sec:asymptotics} studies its asymptotic consistency and convergence rate. Section~\ref{sec:experiments} evaluates our method empirically on simulated and real scientific data. Section~\ref{sec:discussion} concludes with a discussion of our findings and directions for future work.

%% file: sections/pmc_diagram.tex
\begin{tikzpicture}[
  font=\footnotesize,
  node distance=1mm and 14mm,
  arrow/.style={->, thick},
  >=Latex
]

\node (n1) {\textcolor{inputColor}{$x_1^{(1)}, \ldots, x_1^{(S)} \sim \pi_1$}};
\node[draw, rounded corners, fit=(n1), inner sep=3pt, minimum width=3.6cm] (box1) {};
\node[below=16mm of box1] (nK) {\textcolor{inputColor}{$x_K^{(1)}, \ldots, x_K^{(S)} \sim \pi_K$}};
\node[draw, rounded corners, fit=(nK), inner sep=3pt, minimum width=3.6cm] (boxK) {};
\node at ($(box1.south)!0.5!(boxK.north)$) (vdots) {$\vdots$};

\node[fit=(box1)(boxK), inner sep=0pt] (leftgroup) {};
\node[above=1pt of leftgroup.north] {\scriptsize\textcolor{inputColor}{Input}};

\node[draw, rounded corners, right=30mm of vdots, yshift=7mm, minimum width=22mm] (consensus) {$C^{(1)}, \ldots, C^{(S)}$};
\node[above=1pt of consensus.north] {\scriptsize Consensus};
\node[draw, rounded corners, right=30mm of vdots, yshift=-7mm, minimum width=22mm] (discrepancy) {$D^{(1)}, \ldots, D^{(S)}$};
\node[above=1pt of discrepancy.north] {\scriptsize Discrepancy};
\node[fit=(consensus)(discrepancy), inner sep=0pt] (transform) {};

\node[draw, rounded corners, right=10mm of transform] (regression) {$h(C)$ on $D$};
\node[above=1pt of regression.north] {\scriptsize Regression};

\node[align=center, right=10mm of regression] (outtext) {%
  \textcolor{outputColor}{$\hat{\mu}_h \approx \E[h(C)\mid D=\mathbf{0}]$}
};

\node[draw, rounded corners, fit=(outtext), inner sep=3pt] (rightbox) {};
\node[above=1pt of rightbox.north] {\scriptsize\textcolor{outputColor}{Output}};

\draw[arrow] (box1.east) -- (consensus.west);
\draw[arrow] (box1.east) -- (discrepancy.west);
\draw[arrow] (boxK.east) -- (consensus.west);
\draw[arrow] (boxK.east) -- (discrepancy.west);
\draw[arrow] (consensus.east) -- (regression.west);
\draw[arrow] (discrepancy.east) -- (regression.west);
\draw[arrow] (regression.east) -- (rightbox.west);

\end{tikzpicture}

%% file: sections/02_PMC.tex
\section{Proposed method}\label{sec:PMC}


\paragraph{Setting and notation.}
Let $\pi_1,\ldots,\pi_K$ be $K\geq 2$ arbitrary distributions on $\mathcal{X}=\mathbb{R}^p$.
We assume we can draw mutually independent random samples
\[
    x_k \sim \pi_k \quad k = 1, \ldots, K
\]
from individual distributions, but we do not assume we can evaluate the density function. Collectively, we write $\boldsymbol{x}=(x_1,\ldots,x_K) \in \mathcal{X}^K$. Our goal is to compute $\mu_h$, the expectation of an arbitrary function $h$ under the product distribution.





\subsection{Consensus and discrepancy}\label{sec:consensus_discrepancy}

We generalize our heuristic: the product distribution $\pipro$ can be represented as the distribution of a \emph{consensus mapping} conditional on a \emph{discrepancy mapping} being zero. We first define these mappings.


\begin{definition}[Consensus]
\label{def:consensus}
A measurable and continuous mapping
$c:\mathcal{X}^K \to \mathcal{X}$ is a \emph{consensus} if
\begin{enumerate}[nosep,leftmargin=*]
    \item For identical inputs, $c$ returns their common value:
    \[
        c((x,\ldots,x))=x \quad \forall\,x\in\mathcal{X}.
    \]
    \item $c$ is translation-equivariant:
    \[
        c(\boldsymbol{x}+a)=c(\boldsymbol{x})+a
        \quad \forall\,\boldsymbol{x}\in\mathcal{X}^K,\ a\in\mathcal{X},
    \]
    where $\boldsymbol{x}+a=(x_1+a,\ldots,x_K+a)$.
\end{enumerate}
\end{definition}
Examples of consensus mappings $c$ include the \emph{arithmetic average}, the \emph{weighted average}, and the \emph{median}. Translation equivariance implies that shifting all components shifts the consensus by the same amount. We impose this natural condition to simplify the analysis.
\begin{definition}[Discrepancy]
\label{def:discrepancy}
Consider the full discrepancy mapping $\delta:\mathcal{X}^K\to\mathbb{R}^{\left(K-1\right)p}$ defined by
\[
    \delta\left(\boldsymbol{x}\right)
    = \left(x_2-x_1,\ldots,x_K-x_1\right) \quad \forall \; \boldsymbol{x} \in \mathcal{X}^K.
\]
Let $\psi:\mathbb{R}^{(K-1)p}\to\mathbb{R}^m$ be continuous and satisfy
\begin{enumerate}
    \item $\psi^{-1}(\{\mathbf{0}_m\})=\{\mathbf{0}_{(K-1)p}\}$.

    \item $\inf\limits_{\lVert v \rVert_{\infty} \geq a} \lVert \psi(v) \rVert_{\infty} > 0 \quad \forall \; a > 0$.
\end{enumerate}
Then, $d:\mathcal{X}^K \to \mathbb{R}^{m}$ is a \emph{discrepancy} if
\[
    d = \psi\circ\delta.
\]
\end{definition}
Examples of the discrepancy $d$ include (i) full discrepancy mapping $d= \delta(\boldsymbol{x}) \in \R^{(K-1)p}$, and (ii) its norm $d= ||\delta(\boldsymbol{x})||_\infty \in \R$.
By design, $d(\boldsymbol{x}) = \mathbf{0}_m$ implies $x_1= x_2=\ldots = x_K$.

Given $\boldsymbol{x}=(x_1,\ldots,x_K)$,
we henceforth write
\[
    C=c\left(\boldsymbol{x}\right),\qquad D=d\left(\boldsymbol{x}\right)
\]
for the induced consensus and discrepancy variables.

The following theorem formally connects the conditional distribution of the consensus given zero discrepancy to our target product distribution.

\begin{theorem}
    \label{thm:consensus_discrepancy}
    Let $c$ be a consensus and $d$ be a discrepancy. Suppose all component densities $\pi_k$ are $L^K$-integrable.
    Let $x_k\sim \pi_k$ be mutually independent for $k = 1, \ldots, K$.
    Let $\pi_{\varepsilon}$ denote the law of $C$ conditional on $D$ in the $\varepsilon$ neighborhood of zero:
    \[
        C \mid  \lVert D \rVert_\infty < \varepsilon \sim \pi_{\varepsilon} \quad \forall \; \varepsilon > 0.
    \]
    Then, $\pi_{\varepsilon}$ converges to $\pipro$ in total variation:
    \[
        \mathrm{TV} \left(\pi_{\varepsilon}, \pipro\right) \overset{}{\to} 0  \quad \mathrm{as} ~\varepsilon \to 0.
    \]
\end{theorem}
Convergence in total variation further implies
\[
  \pi_{\varepsilon} \weakto \pipro \quad \mathrm{as} ~\varepsilon \to 0.
\]
Moreover, for any bounded measurable $h$, we have
\[
\E[h(C) \mid \lVert D \rVert_\infty < \varepsilon] \to \mu_h \quad \mathrm{as} ~\varepsilon \to 0.
\]

\subsection{Conditional mean via regression}\label{sec:regression_estimation}
By Theorem~\ref{thm:consensus_discrepancy}, we can now recast product-expectation estimation as a regression problem. Given $S$ samples from each of the $K$ component distributions
\[
x_{k}^{(s)} \sim  \pi_k,
\quad s = 1, \ldots, S, \quad k = 1, \ldots, K,
\]
we obtain the training pairs
\[
C^{(s)}=c (\boldsymbol{x}^{(s)}),
\qquad
D^{(s)}=d (\boldsymbol{x}^{(s)}),
\quad s=1,\ldots,S.
\]
We then regress $h(C)$ on $D$ to estimate
$f_h(D)=\E[h(C)\mid D]$ and evaluate the fitted function at zero:
\[
\hat f_{h,S}(D) \approx f_h(D), \qquad
\hat\mu_h = \hat f_{h,S}(\mathbf 0_m).
\]

PMC provides a general regression framework, see Algorithm~\ref{alg:PMC}. For a point estimate, a natural choice is to parameterize \(f\) as a flexible neural network and minimize the empirical squared error $\min_{f}
\sum\!
\left\|h(C^{(s)})-f(D^{(s)})\right\|_2^2$, which ensures that the population optimum is the conditional mean: $\arg\min_f
\E\!\left\|h(C)-f(D)\right\|_2^2 = \E[h(C)| D]$.
More generally, our PMC framework turns a computing task into a modeling problem. Beyond the simple point estimate, we next develop a fully Bayesian formulation, and thus integrate the Bayesian workflow \citep{gelman2020bayesian} and uncertainty quantification.


\begin{algorithm}
\small
\caption{Product Monte Carlo}
\label{alg:PMC}
\begin{algorithmic}[1]
\For{$s = 1$ to $S$}
     \For{$k = 1$ to $K$}
     \State Draw $x_k^{(s)}\sim \pi_k$ independently.
     \EndFor
    \State $C^{(s)} \gets c\big( x_1^{(s)}, \dots, x_K^{(s)} \big)$.
    \State $D^{(s)} \gets d\big( x_1^{(s)}, \dots, x_K^{(s)} \big)$.
\EndFor

\State $\hat{f}_{h,S}\gets$ regress $h(C^{(s)})$ on $D^{(s)}$ via $f_h$, $s=1,\ldots,S$.

\State \textbf{return} $\hat{\mu}_h=\hat{f}_{h,S}(\mathbf{0}_m)$; when the
regression is fully Bayesian, also return posterior uncertainty and draws
$\mu_h^{(b)} \sim p\!\left(f_h(\mathbf{0}_m)\mid C^{(1:S)},D^{(1:S)}\right), \quad b=1,\ldots,B.$
\end{algorithmic}
\end{algorithm}

\subsection{Fully Bayesian formulation, uncertainty quantification, and diagnostics}\label{sec:UQ}

For simplicity, we present the scalar case $h:\mathcal X\to\mathbb R$. The vector-valued extension is immediate. We place a Gaussian process (GP) prior with an RBF kernel on the conditional mean and use a Gaussian likelihood
\[
\begin{aligned}
   h(C^{(s)}) \mid f_h & \overset{\mathrm{ind}}{\sim}\mathcal{N}(f_h(D^{(s)}), \sigma^2) \quad s=1,\ldots,S, \\
   f_h \mid \beta_0, \ell, \tau & \sim\operatorname{GP}(\beta_0,~\kappa), \\
   \kappa(D^{(s)}, D^{(s')}) & = \tau^2 \exp\left(-\frac{\lVert D^{(s)}-D^{(s')} \rVert_2^2}{2\ell^2}\right),
\end{aligned}
\]
and a prior on the hyperparameters: global mean $\beta_0$, lengthscale $\ell$, amplitude $\tau$, and noise scale $\sigma$.

Since $\mu_h=f_h(\mathbf{0}_{m})$, we can integrate out hyperparameters to propagate uncertainty and draw Monte Carlo samples from its posterior given all samples:
\[
    \mu_h^{(b)} \sim p\!\left(f_h(\mathbf{0}_m)\mid C^{(1:S)},D^{(1:S)}\right), \quad b=1,\ldots,B.
\]
It is then straightforward to compute the posterior mean, variance and credible intervals for $\mu_h$.

Uncertainty quantification is especially important when near-consensus observations are sparse. As we show in Section~\ref{sec:asymptotics}, without additional structural assumptions, the samples may provide little information about the product expectation.
In typical Monte Carlo settings, the standard error scales as $O(S^{-1/2})$, whereas PMC’s posterior variance reflects both Monte Carlo and epistemic uncertainty. For example, in our GP setting, as $\ell\to\infty$, we have
$
    \operatorname{Var}\left(f_h(\mathbf{0}_{m})\mid C^{(1:S)},D^{(1:S)}\right)
    =
    \left(\tau^{-2}+S\sigma^{-2}\right)^{-1}
$
so the variance rate is $O(S^{-1})$ asymptotically. In the opposite limit $\ell\to0$, GP regression degenerates and the posterior variance remains $\tau^2$, regardless of $S$. Our fully Bayesian formulation integrates over \(\ell\), effectively allowing data to determine how strongly observations away from $D=\mathbf{0}_{m}$ inform the product expectation. We evaluate PMC calibration in Section~\ref{sec:experiments}.

Our Bayesian formulation provides coherent uncertainty quantification and integrates PMC into a Bayesian workflow. Although we begin with a simple default model with homoscedastic Gaussian likelihood and a simple kernel, our PMC framework is not tied to a
single fixed regression specification. When the model is adequate, inference is calibrated \citep{rubin1984bayesianly}. Otherwise, the workflow supports data-driven model building, diagnostics, and comparison \citep{gelman2020bayesian}. For example, we
may consider heteroscedastic noise, non-Gaussian likelihood, or flexible kernels, and we can select or average
models. In practice, rather than assuming the correct regression specification, we recommend using posterior uncertainty and posterior predictive checks (PPC) to assess the adequacy of the
fitted regression model. For example, Figure~\ref{fig:pmc_ppc_misspecification} illustrates an example: when $C\mid D$ is nonlinear, a misspecified linear regression is detected by small PPC $p$-values. See Figure~\ref{fig:gp_overlap_posterior_variance} for another example.


\begin{figure}[t]
    \centering
    \includegraphics[width=0.9\linewidth]{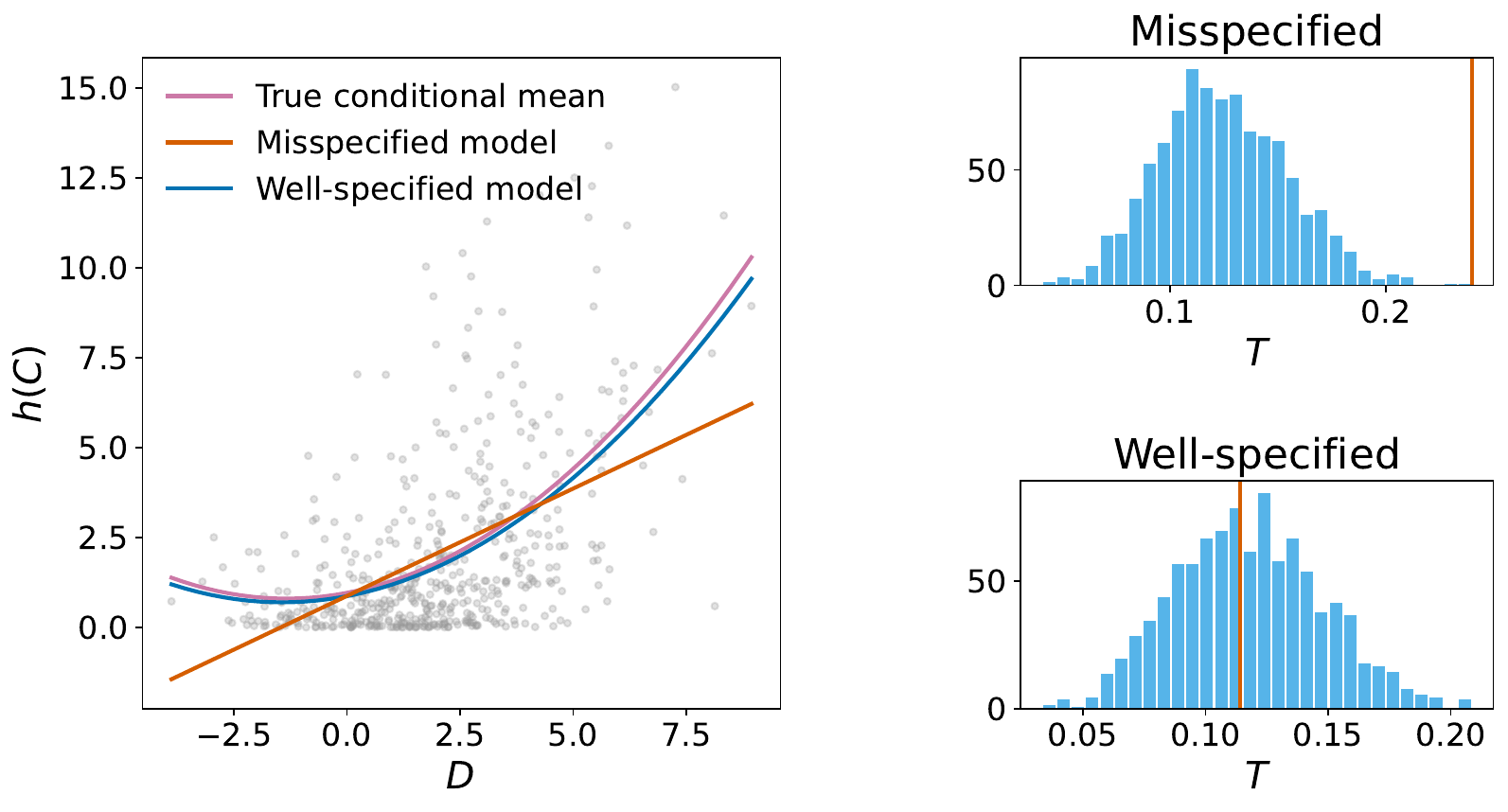}
    \caption{Using posterior predictive check in PMC. Left: true, well-specified, and misspecified regressions of $h(C)$ on $D$. Right: posterior predictive distributions of $T$, a statistic of simulated data, under the two fitted models. Orange lines mark the observed statistics.}
    \label{fig:pmc_ppc_misspecification}
\end{figure}


%% file: sections/03_asymptotics.tex
\section{Asymptotic theory}\label{sec:asymptotics}
Theorem~\ref{thm:consensus_discrepancy} guarantees that the target product expectation $\mu_h$ matches the conditional mean at zero:
\[
\mu_h
=
\E_{\pipro}[h]
=
\E[h(C)\mid D=\mathbf{0}_m].
\]
Therefore, any consistent estimator of \(\E[h(C)\mid D=\mathbf{0}_m]\) is consistent for \(\E_{\pipro}[h]\). This identification requires independence across components, not i.i.d. draws within samplers. PMC supports Markov chain Monte Carlo (MCMC) sampling, retaining consistency under standard ergodicity or mixing conditions.

To further study convergence rates, we specialize our estimator to a local linear regression. We assume i.i.d. component sampler $x_k^{(s)}\overset{\mathrm{iid}}{\sim}\pi_k$ to simplify the theory.


\begin{theorem}[Convergence Rate]
\label{thm:llr_rate}
Under the assumptions of Theorem~\ref{thm:consensus_discrepancy}, let $h$ be bounded and continuous, and let $\mu_h=\E_{\pipro}[h]$. Let $f_h$ and $\Sigma_h$ be versions of the conditional mean and covariance matrix satisfying
\[
    f_h(D)\overset{\mathrm{a.s.}}{=}\E[h(C)\mid D],
    \qquad
    \Sigma_h(D)\overset{\mathrm{a.s.}}{=}\operatorname{Cov}(h(C)\mid D).
\]
Let $\hat{f}_{h,S}$ be a local linear estimator of $f_h$ with fixed kernel $\mathcal{K}$ and bandwidth $b_S$. If Assumptions~\ref{ass:llr_discrepancy}--\ref{ass:llr_bandwidth} hold, the PMC estimator converges in probability
\[
    \hat{f}_{h,S}(\mathbf{0}_m)
    \overset{p}{\to}
    \mu_h.
\]
Moreover, its optimized leading-order asymptotic mean squared error (AMSE) is
\[
    \operatorname{AMSE}_S^\ast
    =
    \mathcal{C}_{\mathcal{K}}
    \left[
        \frac{
            \operatorname{tr}\!\left(\Sigma_h(\mathbf{0}_m)\right)
            \left\lVert\Delta f_h(\mathbf{0}_m)\right\rVert_2^{\frac{m}{2}}
        }{
            S \, p_D(\mathbf{0}_m)
        }
    \right]^{\frac{4}{m+4}},
\]
where $\operatorname{tr}$ denotes the trace, $\Delta$ denotes the Laplacian, and $\mathcal{C}_{\mathcal{K}}$ is a constant depending on the kernel.
\end{theorem}

Theorem~\ref{thm:llr_rate} characterizes the convergence rate of PMC through four quantities: the discrepancy density at zero \(p_D(\mathbf{0}_m)\), the curvature of the conditional mean $\left\lVert\Delta f_h(\mathbf{0}_m)\right\rVert_2$, the total conditional variance $\operatorname{tr}\!\left(\Sigma_h(\mathbf{0}_m)\right)$, and the discrepancy dimension $m$.

The discrepancy density at zero determines the local effective sample size. With bandwidth \(b_S\), the expected number of samples near zero is
\[
S \, \Prob(\|D\|_\infty<b_S)
\approx
S \, (2b_S)^m p_D(\mathbf{0}_m).
\]

Intuitively, when \(p_D(\mathbf{0}_m)\) is small, few samples lie near consensus, leaving PMC to rely more heavily on extrapolation. Proposition~\ref{prop:gcs_decomposition} shows that this discrepancy density at zero, and thereby the effective sample size, is smaller as (a) components move further away or (b) individual distributions are more diffuse.



\begin{proposition}[When PMC is harder]
\label{prop:gcs_decomposition}
Under Assumptions~\ref{ass:gcs_integrability}--\ref{ass:regular_discrepancy_transformation} and Assumption~\ref{ass:finite_component_entropies}, the density of the discrepancy at zero is
{\small
\begin{align*}
    p_D(\mathbf{0}_m)
    &=
    \mathcal{C}_{\psi}
    \exp\left(-\frac{1}{K}D_{\mathrm{GCS}}\left(\pi_{1:K}\right)\right)
    \prod_{k=1}^{K}\left\|\pi_k\right\|_{L^K}\\
 &\geq
    \mathcal{C}_{\psi}
    \exp\left(
        -\frac{1}{K}D_{\mathrm{GCS}}(\pi_{1:K})
        -\frac{K-1}{K}\sum_{k=1}^{K}H(\pi_k)
    \right),
\end{align*}}%
where $\mathcal{C}_{\psi}$ depends on $\psi$, $D_{\mathrm{GCS}}$ is the generalized Cauchy--Schwarz divergence, and $H$ is the entropy.
\end{proposition}




The remaining AMSE components in Theorem~\ref{thm:llr_rate} generally depend on the model and representation. For a concrete theoretical illustration, consider Gaussian components, whose product is also Gaussian:
{\small
\[
\pi_k=\mathcal N(\mu_k,\Sigma_k) \quad k=1,\ldots,K
\;\Rightarrow\;
\pipro=\mathcal N(\mu_\pi,\Sigma_\pi),
\]
\[
\mathrm{where}~~
\Sigma_\pi=\left(\sum_{k=1}^K\Sigma_k^{-1}\right)^{-1},
\qquad
\mu_\pi=\Sigma_\pi\sum_{k=1}^K\Sigma_k^{-1}\mu_k.
\]
}%
PMC can attain the same parametric efficiency without density evaluations or this closed-form update.

\begin{proposition}[Sharpness in Gaussian example]
\label{prop:gaussian_amse}
Consider either of the following PMC representations:
\begin{enumerate}[nosep,leftmargin=*]
    \item Weighted average consensus, full discrepancy:
    \[
        c_W\left(\boldsymbol{x}\right)=\sum_{k=1}^K W_kx_k \text{ with } \sum_{k=1}^K W_k=I_p,
        \text{ and } d=\delta.
    \]
    \item Precision-weighted consensus, scalar discrepancy:
    \[
        c_{W^\ast}\left(\boldsymbol{x}\right)
        =\sum_{k=1}^K\Sigma_\pi\Sigma_k^{-1}x_k,
        \text{ and }
        d=\left\lVert\delta\right\rVert_2^2.
    \]
\end{enumerate}
For $h(u)=u$, the leading-order AMSE of PMC for estimating $\E_{\pipro}[h]$ with local linear regression is
\[
\operatorname{AMSE}_S^\ast
=
\frac{\operatorname{tr}(\Sigma_\pi)}{S}
(1+\mathcal C_D),
\]
where $\mathcal{C}_D \geq 0$ is a constant depending on the chosen representation and the component parameters.
\end{proposition}

The first AMSE term matches the Cramér--Rao lower bound for estimating the product mean with $S$ draws from each Gaussian component. Thus, despite using only samples, PMC is asymptotically sharp, attaining the efficiency of the closed-form Gaussian estimator.


Proposition~\ref{prop:gaussian_amse} links the choice of $c$ and $d$ to the convergence rate of PMC, although consistency holds for any admissible representation. With \(d=\delta\), any linear consensus attains the same rate. Scalar \(d=\lVert\delta\rVert_2^2\) simplifies regression but attains the Cramér--Rao bound only with precision weighting. Our experiments use arithmetic-average consensus and full discrepancy.



%% file: sections/04_experiments.tex
\section{Experiments}\label{sec:experiments}


We demonstrate PMC on simulated data and three real applications: a privacy-restricted clinical trial, large-scale particle physics inference, and likelihood-free protein imaging. When \(\pi_k\) can be evaluated, we benchmark PMC against GBF, the current state-of-the-art fusion method, and long-run MCMC targeting the product density. This deliberately favors the baselines, as PMC requires no density evaluations.


As we mentioned, PMC accommodates arbitrary regression models. We consider three implementations:
\begin{itemize}[nosep,leftmargin=*]
    \item \textbf{PMC-OLS}: ordinary least squares regression.
    \item \textbf{PMC-MLP}: multilayer perceptron neural network.
    \item \textbf{PMC-GP}: Gaussian process regression.
\end{itemize}
\paragraph{Benchmark with simulated data.}\label{sec:simulated_data}
We construct benchmarks varying the component family $\pi_k$, distances between components, sample size $S$, dimension $p$, and number of components $K$. We choose parametric families with closed-form products, withholding this information from \textbf{PMC} and \textbf{GBF} during training. We also include a parametric \textbf{Oracle} that knows the product family, estimates component parameters, and applies the closed-form product update.

We assess performance on (a) accuracy, measured by mean squared error (MSE) for estimating $\E_{\pipro}[h]$, and (b) computational cost, measured by wall-clock training time. PMC is the only method considered that is both density-evaluation-free and distribution-free.

We first consider $K=2$ Gaussian components and a scalar target function $h$, varying dimension $p$ and sample size $S$. Figure~\ref{fig:normalized_sum_dimension_scaling} shows PMC-OLS closely matching the Oracle in MSE and improving steadily with $S$, as expected because the linear regression model is well-specified. Its sharp parametric rate agrees with Proposition~\ref{prop:gaussian_amse}, and it has the lowest training cost among PMC variants. PMC-MLP improves at large $S$, but its MSE increases in higher dimensions. At a similar runtime, PMC-GP is significantly more accurate, especially in higher dimensions. All PMC methods outperform GBF at moderately large sample sizes and are substantially faster, often by orders of magnitude. We show GBF results only where computationally feasible, as its training time grows steeply with $S$ and $p$.

\begin{figure}[t]
    \centering
    \includegraphics[width=\linewidth]{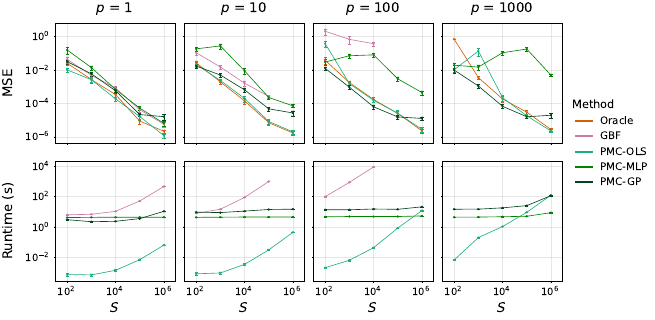}
    \caption{MSE and runtime versus $S$ for the dimension-scaling experiment for the standardized sum of coordinates.}
    \label{fig:normalized_sum_dimension_scaling}
\end{figure}

Further simulations in Appendix~\ref{app:additional_simulation_results} show qualitatively similar results: PMC outperforms benchmarks in accuracy and cost. Results are averaged over 20 replications, with error bars showing standard errors.

In our fully Bayesian formulation, PMC-GP quantifies uncertainty in the target expectation, accounting for both Monte Carlo and epistemic uncertainty. We demonstrate how to incorporate this uncertainty into a Bayesian workflow for model checking and diagnostics.  Figure~\ref{fig:univariate_gaussian_k2_calibration} assesses the calibration of the PMC-GP posterior uncertainty. Although PMC-GP is a Bayesian model and is not designed to guarantee frequentist calibration, its uncertainty estimates provide reasonable frequentist coverage. In contrast, GBF captures only Monte Carlo uncertainty, which is computed from the newly generated particles and therefore depends on particle count. As the number of particles increases, its intervals become increasingly undercovered. To examine how posterior uncertainty provides built-in PMC diagnostics, Figure~\ref{fig:gp_overlap_posterior_variance} considers a simple Gaussian example with two components, $\mathcal N(0,1)$ and $\mathcal N(\mu,1)$. As $\mu$ increases and the components become more separated, the PMC-GP posterior uncertainty increases, providing a warning that the component distributions have limited overlap, the available data provide little information about the target expectation, and the product-estimation problem becomes more difficult.


\begin{figure}[t]
    \centering
    \begin{minipage}[t]{0.49\linewidth}
        \vspace{0pt}
        \centering
        \captionsetup{font=scriptsize,skip=2pt}
        \begin{minipage}[t][0.72\linewidth][c]{\linewidth}
            \centering
            \includegraphics[width=\linewidth]{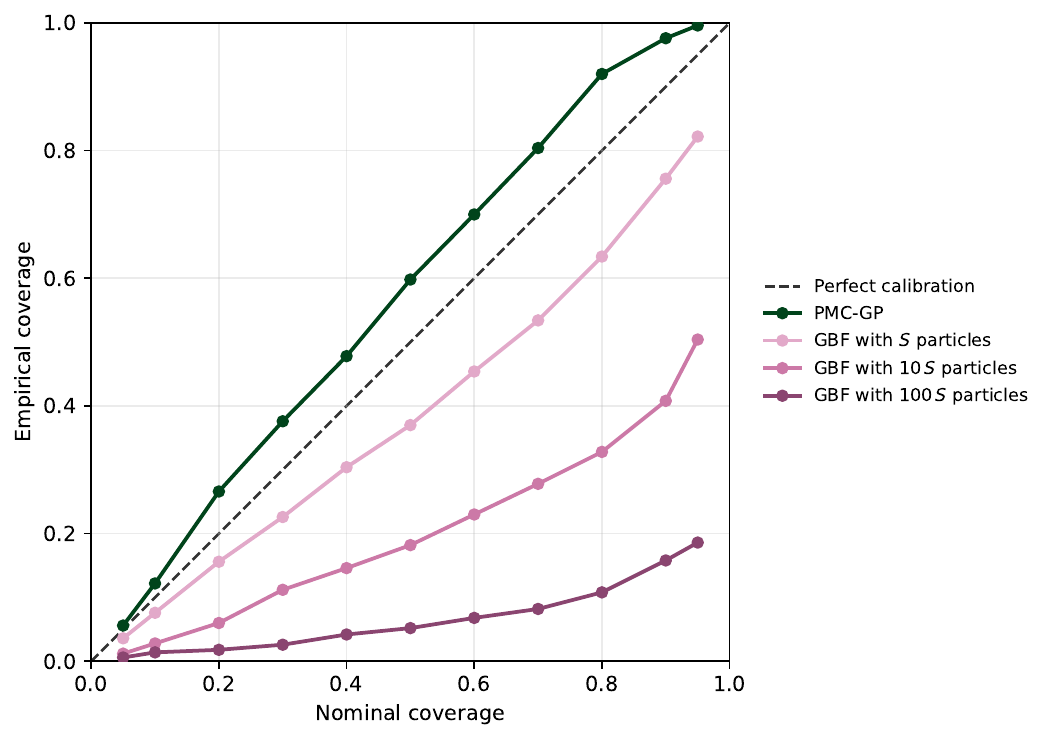}
        \end{minipage}
        \captionof{figure}{Calibration for two Gaussian components. PMC-GP is slightly conservative, while GBF's coverage worsens with particle count.}
        \label{fig:univariate_gaussian_k2_calibration}
    \end{minipage}
    \hfill
    \begin{minipage}[t]{0.49\linewidth}
        \vspace{0pt}
        \centering
        \captionsetup{font=scriptsize,skip=2pt}
        \begin{minipage}[t][0.72\linewidth][c]{\linewidth}
            \centering
            \includegraphics[width=\linewidth]{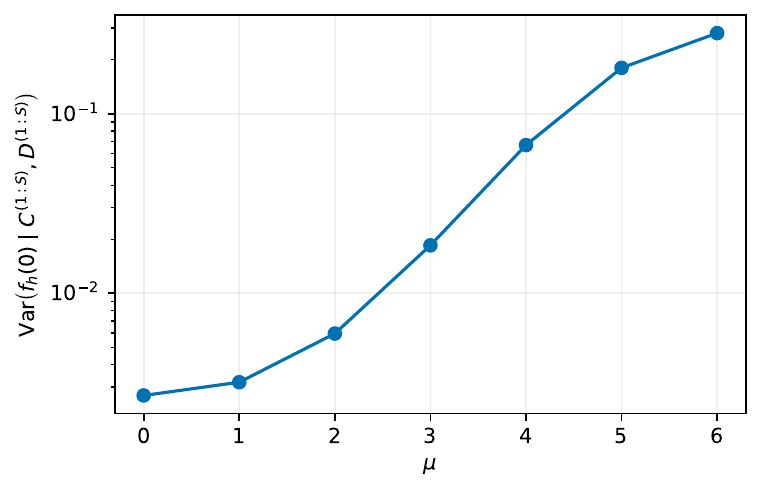}
        \end{minipage}
        \captionof{figure}{PMC-GP posterior variance of the target expectation $f_h(0)=\E_{\pipro}[h]$ versus mean separation between two Gaussian components.}
        \label{fig:gp_overlap_posterior_variance}
    \end{minipage}
\end{figure}

Next, we demonstrate PMC on real-data applications.
When feasible, we compute the true expectation via a full-data MCMC, and evaluate the root mean squared error and runtime of subsample approximations.

\paragraph{Privacy-restricted clinical trial}\label{sec:real_data_federated_learning}

The respiratory data~\citep{halekoh2006geepack} is a public clinical-trial dataset with patients from two clinics randomized to placebo or treatment and binary respiratory outcomes. Researchers fit a logistic regression predicting the outcome from treatment and other covariates. We assume that privacy constraints prevent the clinics from sharing data, but researchers still need to compute the full-data posterior mean. We implement PMC and GBF to aggregate subposteriors. For an oracle truth, we also run MCMC on the pooled data, which would violate the assumed privacy constraints. For runtime, we separate (a) \emph{Subsample}, the cost of obtaining local posterior samples by MCMC, from (b) \emph{Merge}, the additional cost of aggregating them. \emph{Total} is their sum. \emph{Full MCMC} reports the MCMC runtime on the pooled data. Table~\ref{tab:respiratory-identity-comparison} reports that all methods approximate the full-posterior mean well, while all PMC variants have lower aggregation costs than GBF.


\begin{table}[ht]
\centering
\small
\setlength{\tabcolsep}{2.5pt}
\input{tables/respiratory_identity_comparison}
\vspace{-0.5em}
\caption{Application in a privacy-restricted clinical trial.}\label{tab:respiratory-identity-comparison}
\vspace{-0.5em}
\end{table}

\paragraph{Large-scale inference }\label{sec:real_data_large_scale_computation}

The HIGGS data~\citep{baldi2014searching} contains 11 million simulated collision events labeled as Higgs-producing signal or background. We fit a logistic regression with 28 predictors, splitting the data into $K=10$ batches to reduce computational costs. Table~\ref{tab:higgs-identity-comparison} shows that PMC achieves small estimation errors with negligible computational cost relative to subsampling. Total runtime is roughly 10\% of full-data MCMC. GBF is computationally infeasible at this scale and omitted. Appendix~\ref{app:additional_real_data_results} reports additional large-scale inference examples, where GBF takes significantly longer even when feasible.

\begin{table}[ht]
\centering
\small
\setlength{\tabcolsep}{2.5pt}
\input{tables/higgs_identity_comparison}
\vspace{-0.5em}
\caption{Application in large-scale inference. PMC's minimal overhead keeps total cost a fraction of full-data MCMC.}\label{tab:higgs-identity-comparison}
\vspace{-0.5em}
\end{table}

\paragraph{Simulation-based inference}\label{sec:real_data_simulation_based_inference}

Recent advances in cryo-electron microscopy (cryoEM) enable inference on protein structures from 2D images using approaches such as SBI \citep{dingeldein2025amortized, lai2024predictive}. To mimic a realistic setting, we simulate \(K=10\) noisy cryoEM images of the HSP90 protein using cryoJAX~\citep{obrien2026cryojax}. We express the binding angle \(\theta\) as a transformation of a latent parameter \(z\) and perform posterior inference on the latent scale. Although the simulator is differentiable, the likelihood of the model is intractable because of additional nuisance parameters (e.g., rotations, translations, and measurement noise). For each image, we use neural posterior estimation to obtain individual posteriors, and then aggregate them with PMC into the full-data posterior. Figure~\ref{fig:hsp90} visualizes the subposteriors, the true combined posterior, and the PMC estimates. As shown in the table on the right of Figure~\ref{fig:hsp90}, PMC-MLP nearly exactly recovers the full-data posterior mean, while PMC-GP is very close. Sampling the image-level subposteriors with SBI alone takes 1.7 minutes each, whereas the additional cost of PMC is negligible. As a demonstration of an SBI-PMC workflow, amortized SBI enables fast, parallel inferences across subsamples, while PMC provides a reliable and scalable way to aggregate them.

\begin{figure}[t]
    \centering

    \begin{minipage}[c]{0.325\linewidth}
        \centering
        \includegraphics[width=\linewidth]{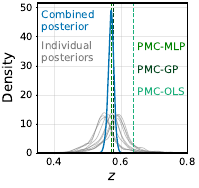}
    \end{minipage}
    \hfill
    \begin{minipage}[c]{0.65\linewidth}
        \centering
        \raisebox{6.1pt}[\height][\depth]{%
            \resizebox{\linewidth}{!}{
                \input{tables/hsp90_identity_comparison}
            }%
        }
    \end{minipage}

    \caption{
Application to SBI.
    \textit{Left:} Image-level subposteriors and the combined posterior of the latent parameter $z$. PMC-MLP and PMC-GP estimates recover the true posterior mean. \textit{Right Table:} PMC closely approximates the full-data posterior expectation with negligible merging cost relative to image-level subsampling SBI.
    }
    \label{fig:hsp90}
\end{figure}

%% file: tables/respiratory_identity_comparison.tex
\begin{tabular}{lccc}
\toprule
Method & Estimation error & Merge (min) & Total (min) \\
\midrule
Full MCMC & 0.0000 & -- & 0.8152 \\
Subsample & -- & -- & 0.9796 \\
GBF & 0.0184 & 0.2047 & 1.1843 \\
PMC-OLS & 0.0110 & 0.0002 & 0.9799 \\
PMC-MLP & 0.0144 & 0.1069 & 1.0866 \\
PMC-GP & 0.0258 & 0.1009 & 1.0806 \\
\bottomrule
\end{tabular}

%% file: tables/higgs_identity_comparison.tex
\begin{tabular}{lccc}
\toprule
Method & Estimation error & Merge (min) & Total (min) \\
\midrule
Full MCMC & 0.0000 & -- & 308.1124 \\
Subsample & -- & -- & 35.2890 \\
PMC-OLS & 0.0020 & 0.0122 & 35.3012 \\
PMC-MLP & 0.0001 & 0.1092 & 35.3982 \\
PMC-GP & 0.0001 & 6.7192 & 42.0082 \\
\bottomrule
\end{tabular}

%% file: tables/hsp90_identity_comparison.tex
\begin{tabular}{lccc}
\toprule
Method & Estimation error & Merge (min) & Total (min) \\
\midrule
Sub-SBI & -- & -- & 1.7089 \\
PMC-OLS & 0.0684 & 0.0000 & 1.7089 \\
PMC-MLP & 0.0003 & 0.0134 & 1.7223 \\
PMC-GP & 0.0089 & 0.0042 & 1.7130 \\
\bottomrule
\end{tabular}

%% file: sections/05_discussion.tex
\section{Discussion}\label{sec:discussion}


\paragraph{Limitations.}


We introduced PMC, a framework to compute expectations under products of densities using only samples from the individual distributions. As quantified by Theorem~\ref{thm:llr_rate} and Proposition~\ref{prop:gcs_decomposition}, limited component overlap reduces the local effective sample size, making aggregation more reliant on extrapolation. In PMC, however, large posterior uncertainty or failed posterior predictive checks warn that the samples may not support reliable estimation. In such cases, we recommend sampling from the full target when feasible or redesigning the components to improve overlap. This limitation highlights a tradeoff, especially in large-scale inference: smaller batches reduce local sampling cost but can yield more divergent subposteriors, whereas larger ones improve overlap at a greater sampling cost. An interesting direction is therefore to integrate PMC and its diagnostics into a principled workflow for subsampling and aggregation.




\paragraph{Future directions.}

PMC accommodates arbitrary regression models, and we recommend using the Bayesian workflow for model checking and tuning. Moreover, if additional knowledge about the target distribution is available, it could be incorporated into the regression model to further improve efficiency. Furthermore, our implementation uses a fixed choice of \(c\) and \(d\), but the framework permits learning or adapting them from data, which we view as a promising future direction. Finally, notice that the conditional distribution \(C\mid D=\mathbf{0}\) matches the full product distribution $\pipro$, not only its mean. This suggests that our consensus and discrepancy representation could be extended to \emph{generative sampling} by replacing the current conditional-mean regression with flexible conditional density estimation to sample directly from the estimated conditional distribution $\hat \pi (C|D=\bf 0)$, which we leave for future work.

%% file: sections/appendix.tex
\section{Proofs of Main Results}\label{app:proofs}


We use the following notation throughout the appendix. Let $q=(K-1)p$ denote the dimension of the difference space. We write $\operatorname{Leb}^s$ for $s$-dimensional Lebesgue measure, $C^s(U)$ for the class of functions on $U$ with continuous partial derivatives up to order $s$, $J_\varphi$ for the Jacobian of $\varphi$, and $\Phi_+$ for the set of non-negative measurable functions on $\mathcal{X}\times\mathbb{R}^q$. We use $\otimes$ and $\odot$ for the Kronecker and coordinatewise products, respectively. We write $o_p$ for terms of smaller asymptotic order in probability and $A\succ0$ when $A$ is a positive definite matrix.

\subsection{Proof of Theorem~\ref{thm:consensus_discrepancy}}
\label{app:proof_consensus_discrepancy}

The following definition introduces the norm on density functions used throughout the proof.

\begin{definition}[$L^K$ norm]
\label{def:lk_norm}
For a density $g$ on $\mathcal{X}$, the $L^K$ norm of $g$ is
\[
    \|g\|_{L^K} = \left(\int_{\mathcal{X}}g(x)^K\,\mathrm{d}x\right)^{\frac{1}{K}}.
\]
\end{definition}
This definition describes the coordinates induced by the consensus function.
\begin{definition}[Coordinate Map]
    The consensus and difference maps define the following coordinates:
    \[
        T_c(\boldsymbol{x})
        =\bigl(c(\boldsymbol{x}),\delta(\boldsymbol{x})\bigr)
        \quad \forall\,\boldsymbol{x}\in\mathcal X^K.
    \]
\end{definition}

The next two definitions introduce the sets used to analyze the conditioning event.

\begin{definition}[Small Difference Set]
\label{def:small_difference_set}
    The coordinate differences are small if they belong to
    \[
        B_\varepsilon
        =\left\{\boldsymbol v\in\mathbb R^q:
        \|\boldsymbol v\|_\infty<\varepsilon\right\}
        \quad \forall\,\varepsilon>0.
    \]
\end{definition}

\begin{definition}[Small Discrepancy Set]
\label{def:small_discrepancy_set}
    The discrepancy is small if the coordinate differences belong to
    \[
        S_\varepsilon
        =\left\{\boldsymbol v\in\mathbb R^q:
        \|\psi(\boldsymbol v)\|_\infty<\varepsilon\right\}
        \quad \forall\,\varepsilon>0.
    \]
\end{definition}

With these definitions in place, we can now prove Theorem~\ref{thm:consensus_discrepancy}.

\begin{proof}

Consider the following function that measures the consensus offset relative to the first component:
\[
    \gamma(\boldsymbol{v})
    =c(\mathbf{0}_p,v_1,\ldots,v_{K-1})
    \quad \forall \, \boldsymbol{v}\in\mathbb{R}^{q}.
\]
By the translation-equivariance assumption in Theorem~\ref{thm:consensus_discrepancy}, we write the consensus using the consensus offset:
\begin{equation}
\label{eq:consensus_offset}
    c(\boldsymbol{x})
    =x_1+\gamma(\delta(\boldsymbol{x}))
    \quad \forall \, \boldsymbol{x}\in\mathcal{X}^K.
\end{equation}
Now, consider the following functions that measure the residual shifts:
\[
    r_1(\boldsymbol{v})=-\gamma(\boldsymbol{v}),
    \qquad
    r_k(\boldsymbol{v})=v_{k-1}-\gamma(\boldsymbol{v}),
    \quad k=2,\ldots,K,
    \quad \forall \, \boldsymbol{v}\in\mathbb{R}^{q}.
\]
Then, by Equation~\eqref{eq:consensus_offset}, we write each coordinate in terms of the consensus and residual shifts:
\[
    x_k=c(\boldsymbol{x})+r_k(\delta(\boldsymbol{x})),
    \qquad k = 1, \ldots, K,
    \quad \forall \, \boldsymbol{x}\in\mathcal{X}^K.
\]
To establish surjectivity of the coordinate map, consider the following candidate preimage:
\[
    \boldsymbol{x}_{u, \boldsymbol{v}}
    =\bigl(u+r_1(\boldsymbol{v}),\ldots,u+r_K(\boldsymbol{v})\bigr)
    \quad \forall \, (u,\boldsymbol{v})\in\mathcal{X}\times\mathbb{R}^q.
\]
By the definitions of the residual shifts and the difference map, its difference coordinates are
\[
    \delta(\boldsymbol{x}_{u, \boldsymbol{v}})
    =\bigl(r_2(\boldsymbol{v})-r_1(\boldsymbol{v}),\ldots,
            r_K(\boldsymbol{v})-r_1(\boldsymbol{v})\bigr)
    =\boldsymbol{v} \quad \forall \, (u,\boldsymbol{v})\in\mathcal{X}\times\mathbb{R}^q.
\]
By Equation~\eqref{eq:consensus_offset}, its consensus coordinate is
\[
    c(\boldsymbol{x}_{u, \boldsymbol{v}})
    =u+r_1(\boldsymbol{v})+\gamma(\boldsymbol{v})=u \quad \forall \, (u,\boldsymbol{v})\in\mathcal{X}\times\mathbb{R}^q.
\]
Therefore, the
coordinate map is bijective with inverse
\begin{equation}
\label{eq:coordinate_map_inverse}
    T_c^{-1}(u,\boldsymbol{v})
    =\bigl(u+r_1(\boldsymbol{v}),\ldots,u+r_K(\boldsymbol{v})\bigr) \quad \forall \, (u,\boldsymbol{v})\in\mathcal{X}\times\mathbb{R}^q.
\end{equation}
Now, consider the linear change of variables to the first component and
the component differences:
\[
    L(\boldsymbol{x})=(x_1,\delta(\boldsymbol{x}))
    \quad \forall \, \boldsymbol{x}\in\mathcal{X}^K.
\]
It is easy to verify that the Jacobian of $L$ is block lower triangular with identity blocks on the diagonal:
\begin{equation}
\label{eq:linear_change_jacobian}
    J_L(\boldsymbol{x})=
    \begin{pmatrix}
        I_p  & 0   & \cdots & 0\\
        -I_p & I_p & \cdots & 0\\
        \vdots & \vdots & \ddots & \vdots\\
        -I_p & 0 & \cdots & I_p
    \end{pmatrix}
    \quad \forall \, \boldsymbol{x}\in\mathcal{X}^K.
\end{equation}
Then, by Equation~\eqref{eq:linear_change_jacobian}, the linear change of variables is volume-preserving:
\begin{equation}
\label{eq:linear_change_determinant}
    \det\!\left(J_L(\boldsymbol{x})\right)=1
    \quad \forall \, \boldsymbol{x}\in\mathcal{X}^K.
\end{equation}
The inverse of the linear change of variables is
\[
    L^{-1}(t,\boldsymbol{v})
    =(t,t+v_1,\ldots,t+v_{K-1})
    \quad \forall \, (t,\boldsymbol{v})\in\mathcal{X}\times\mathbb{R}^q.
\]
Then, by Equation~\eqref{eq:linear_change_determinant}, the inverse is also volume-preserving:
\[
    \det\!\left(J_{L^{-1}}(t,\boldsymbol{v})\right)=1
    \quad \forall \, (t,\boldsymbol{v})\in\mathcal{X}\times\mathbb{R}^q.
\]
By the change-of-variables formula and translation invariance of Lebesgue measure, $T_c$ is volume-preserving:
\begin{equation}
\label{eq:coordinate_map_volume_preserving}
\begin{aligned}
    \int_{\mathcal X^K} \varphi(T_c(\boldsymbol{x}))\,\mathrm{d}\boldsymbol{x}
    &=\int_{\mathbb R^q}\int_{\mathcal X}
        \varphi\bigl(T_c(L^{-1}(t,\boldsymbol{v}))\bigr)
        \,\mathrm{d}t\,\mathrm{d}\boldsymbol{v}\\
    &=\int_{\mathbb R^q}\int_{\mathcal X}
        \varphi(t+\gamma(\boldsymbol{v}),\boldsymbol{v})
        \,\mathrm{d}t\,\mathrm{d}\boldsymbol{v}\\
    &=\int_{\mathbb R^q}\int_{\mathcal X}
        \varphi(u,\boldsymbol{v})\,\mathrm{d}u\,\mathrm{d}\boldsymbol{v}
        \quad \forall \, \varphi\in\Phi_{+}.
\end{aligned}
\end{equation}
By Equation~\eqref{eq:coordinate_map_volume_preserving}, expectations of transformations of the coordinate map can be written as follows:
\begin{equation}
\label{eq:coordinate_map_expectation}
    \E[\varphi(T_c(\boldsymbol x))]
    =\int_{\mathbb R^q}\int_{\mathcal X}
        \varphi(u,\boldsymbol v)
        \prod_{k=1}^K\pi_k\bigl(u+r_k(\boldsymbol v)\bigr)
        \,\mathrm{d}u\,\mathrm{d}\boldsymbol v
    \quad \forall\,\varphi\in\Phi_+.
\end{equation}
Then, by Equation~\eqref{eq:coordinate_map_expectation},
a pushforward density under the coordinate map is
\begin{equation}
\label{eq:coordinate_map_density}
    p_{T_{c} (\boldsymbol{x})}(u,\boldsymbol{v})
    =\prod_{k=1}^{K}\pi_k\bigl(u+r_k(\boldsymbol{v})\bigr)
    \quad \forall \, (u,\boldsymbol{v})\in\mathcal{X}\times\mathbb{R}^{q}.
\end{equation}
By Definition~\ref{def:consensus}, the consensus offset is origin-preserving:
\begin{equation}
\label{eq:consensus_offset_zero}
    \gamma(\mathbf{0}_{q})
    =\mathbf{0}_p.
\end{equation}
Moreover, by the definitions of the residual shifts and
Equation~\eqref{eq:consensus_offset_zero}, the residual shifts are
origin-preserving:
\begin{equation}
\label{eq:residual_shift_zero}
    r_k(\mathbf{0}_{q})=\mathbf{0}_p,
    \quad k = 1, \ldots, K.
\end{equation}
Therefore, by Equations~\eqref{eq:coordinate_map_density} and~\eqref{eq:residual_shift_zero}, this density with null difference coordinates can be written as follows:
\begin{equation}
\label{eq:origin_density_slice}
    p_{T_{c} (\boldsymbol{x})}(u,\mathbf{0}_{q})
    =\prod_{k=1}^{K}\pi_k(u)
    \quad \forall \, u\in\mathcal{X}.
\end{equation}
Using a telescoping sum argument, we obtain the following:
\begin{equation}
\label{eq:tv_bound_telescopic}
    p_{T_{c} (\boldsymbol{x})}(u,\boldsymbol{v}) - p_{T_{c} (\boldsymbol{x})}(u,\mathbf{0}_{q}) = \sum_{k=1}^{K} \left(\pi_k(u+r_k(\boldsymbol{v})) - \pi_k(u)\right) \prod_{j<k} \pi_j(u+r_j(\boldsymbol{v})) \prod_{j>k} \pi_j(u) \quad \forall \, (u,\boldsymbol{v})\in\mathcal{X}\times\mathbb{R}^{q}.
\end{equation}
Now, let us denote slices of the density under the coordinate map as follows:
\[
    p_{T_{c}(\boldsymbol x)}^{\boldsymbol{v}} (u)
    =p_{T_c(\boldsymbol x)}(u,\boldsymbol v)
    \quad \forall \, (u,\boldsymbol{v})\in\mathcal{X}\times\mathbb{R}^{q}.
\]
By Equation~\eqref{eq:tv_bound_telescopic}, the triangle and Hölder inequalities, and translation invariance of Lebesgue measure, we have
\begin{equation}
\label{eq:tv_bound}
\begin{aligned}
    \left\lVert p_{T_{c}(\boldsymbol x)}^{\boldsymbol{v}} - p_{T_{c}(\boldsymbol x)}^{\mathbf{0}_{q}} \right\rVert_{L^1}
    &\leq \sum_{k=1}^K\int_{\mathcal X}
        |\pi_k(u+r_k(\boldsymbol v))-\pi_k(u)|
        \prod_{j<k}\pi_j(u+r_j(\boldsymbol v))
        \prod_{j>k}\pi_j(u)\,\mathrm{d}u\\
    &\leq \sum_{k=1}^K
        \left(\int_{\mathcal X}
        |\pi_k(u+r_k(\boldsymbol v))-\pi_k(u)|^K\,\mathrm{d}u
        \right)^{1/K}
        \prod_{j\ne k}\|\pi_j\|_{L^K} \quad \forall \, \boldsymbol{v}\in\mathbb{R}^{q}.
\end{aligned}
\end{equation}
By the continuity assumption in Theorem~\ref{thm:consensus_discrepancy} and by Equation~\eqref{eq:residual_shift_zero}, residual shifts are continuous at the origin:
\begin{equation}
\label{eq:residual_shift_limit}
    r_k(\boldsymbol{v}) \to \mathbf{0}_p \quad \text{as } \boldsymbol{v}\to\mathbf{0}_q, \quad k = 1, \ldots, K.
\end{equation}
By continuity of translations in \(L^K\) and Equations~\eqref{eq:tv_bound} and~\eqref{eq:residual_shift_limit}, the density slices converge in \(L^1\):
\begin{equation}
\label{eq:density_slice_convergence}
    \left\|p_{T_c(\boldsymbol x)}^{\boldsymbol v}
    -p_{T_c(\boldsymbol x)}^{\mathbf{0}_q}\right\|_{L^1}
    \to 0 \quad \text{as } \boldsymbol{v}\to\mathbf{0}_q.
\end{equation}
Since $\psi$ is bounded away from the origin by assumption in
Theorem~\ref{thm:consensus_discrepancy}, we obtain the following:
\begin{equation}
\label{eq:small_discrepancy_inclusion}
    \forall\,a>0\quad\exists\,\varepsilon_a>0:
    \quad S_\varepsilon\subseteq B_a
    \quad\forall\,\varepsilon\in(0,\varepsilon_a).
\end{equation}
Since $\psi$ is continuous and origin-preserving by assumption in Theorem~\ref{thm:consensus_discrepancy}, the small-discrepancy sets satisfy
\begin{equation}
\label{eq:small_discrepancy_openness}
    S_\varepsilon\text{ is open}
    \qquad \text{and} \qquad \mathbf{0}_q\in S_\varepsilon
    \quad \forall\,\varepsilon>0.
\end{equation}
By Equations~\eqref{eq:small_discrepancy_inclusion}
and~\eqref{eq:small_discrepancy_openness}, small-discrepancy sets
have positive, finite Lebesgue measure:
\[
    0<\operatorname{Leb}^q(S_\varepsilon)
    \leq\operatorname{Leb}^q(B_1)=2^q<\infty
    \quad\forall\,\varepsilon\in(0,\varepsilon_1).
\]
Now, consider the average density slice over the small-discrepancy set:
\begin{equation}
\label{eq:average_density_slice}
    p_{T_c(\boldsymbol x)}^{S_\varepsilon}(u)
    =\frac{1}{\operatorname{Leb}^q(S_\varepsilon)}
      \int_{S_\varepsilon}
      p_{T_c(\boldsymbol x)}^{\boldsymbol v}(u)\,\mathrm{d}\boldsymbol v \quad \forall \, u \in \mathcal{X}, \quad \varepsilon\in(0,\varepsilon_1)
\end{equation}
Then, its normalizing constant is defined as follows:
\[
    Z_\varepsilon
    =\int_{\mathcal X} p_{T_c(\boldsymbol x)}^{S_\varepsilon}(u) \,\mathrm{d}u
    =\frac{\Prob(\|D\|_\infty<\varepsilon)}
           {\operatorname{Leb}^q(S_\varepsilon)} \quad \forall \, \varepsilon\in(0,\varepsilon_1)
\]
By Equation~\eqref{eq:average_density_slice}, the triangle inequality,
and Tonelli's theorem, we obtain the following bound:
\begin{equation}
\label{eq:average_density_bound}
\begin{aligned}
    \left\|p_{T_c(\boldsymbol x)}^{S_\varepsilon}
    -p_{T_c(\boldsymbol x)}^{\mathbf{0}_q}\right\|_{L^1}
    &=
    \int_{\mathcal X}\left|
        \frac{1}{\operatorname{Leb}^q(S_\varepsilon)}
        \int_{S_\varepsilon}
        \left[p_{T_c(\boldsymbol x)}^{\boldsymbol v}(u)
        -p_{T_c(\boldsymbol x)}^{\mathbf{0}_q}(u)\right]
        \,\mathrm{d}\boldsymbol v\right|\,\mathrm{d}u\\
    &\leq
    \frac{1}{\operatorname{Leb}^q(S_\varepsilon)}
    \int_{S_\varepsilon}
        \left\|p_{T_c(\boldsymbol x)}^{\boldsymbol v}
        -p_{T_c(\boldsymbol x)}^{\mathbf{0}_q}\right\|_{L^1}
        \,\mathrm{d}\boldsymbol v\\
    &\leq
    \sup_{\boldsymbol v\in S_\varepsilon}
        \left\|p_{T_c(\boldsymbol x)}^{\boldsymbol v}
        -p_{T_c(\boldsymbol x)}^{\mathbf{0}_q}\right\|_{L^1}.
\end{aligned}
\end{equation}
By Equations~\eqref{eq:density_slice_convergence}, \eqref{eq:small_discrepancy_inclusion}, and~\eqref{eq:average_density_bound}, the average density slice converges in \(L^1\):
\begin{equation}
\label{eq:average_density_convergence}
    \left\|p_{T_c(\boldsymbol x)}^{S_\varepsilon}
    -p_{T_c(\boldsymbol x)}^{\mathbf{0}_q}\right\|_{L^1}
    \to 0 \quad \text{as } \varepsilon\to0.
\end{equation}
By Equation~\eqref{eq:origin_density_slice}, $Z$ can be written as the integral of the slice at the origin:
\begin{equation}
\label{eq:z_slice_integral}
    Z
    =\int_{\mathcal X}\prod_{k=1}^K\pi_k(u)\,\mathrm{d}u
    =\int_{\mathcal X}p_{T_c(\boldsymbol x)}^{\mathbf{0}_q}(u)\,\mathrm{d}u
\end{equation}
Then, by Equations~\eqref{eq:average_density_convergence} and~\eqref{eq:z_slice_integral} and the triangle inequality, $Z_{\varepsilon}$ converges to $Z$:
\begin{equation}
\label{eq:z_epsilon_limit}
\begin{aligned}
    |Z_\varepsilon-Z|
    &=\left|\int_{\mathcal X}
        \left[p_{T_c(\boldsymbol x)}^{S_\varepsilon}(u)
        -p_{T_c(\boldsymbol x)}^{\mathbf{0}_q}(u)\right] \,\mathrm{d}u\right|\\
    &\leq
        \left\|p_{T_c(\boldsymbol x)}^{S_\varepsilon}
        -p_{T_c(\boldsymbol x)}^{\mathbf{0}_q}\right\|_{L^1}
    \to 0 \quad \text{as } \varepsilon\to0.
\end{aligned}
\end{equation}
Since $Z$ is positive by assumption, Equation~\eqref{eq:z_epsilon_limit} implies that $Z_\varepsilon$ becomes positive as the discrepancy shrinks:
\[
    \exists\,\varepsilon_0\in(0,\varepsilon_1):
    \quad Z_\varepsilon\geq\frac{Z}{2}>0
    \quad\forall\,\varepsilon\in(0,\varepsilon_0).
\]
Consequently, the probability of the conditioning event is positive:
\begin{equation}
\label{eq:conditioning_event_probability}
    \Prob\!\left(\|D\|_\infty<\varepsilon\right)
    =\operatorname{Leb}^q(S_\varepsilon) \, Z_\varepsilon>0
    \quad\forall\,\varepsilon\in(0,\varepsilon_0).
\end{equation}
By Equations~\eqref{eq:coordinate_map_expectation},
\eqref{eq:average_density_slice}, and~\eqref{eq:conditioning_event_probability},
the conditional law is well-defined and has density
\begin{equation}
\label{eq:pi_epsilon_density}
    \pi_\varepsilon(u)
    =\frac{p_{T_c(\boldsymbol x)}^{S_\varepsilon}(u)}
           {Z_\varepsilon} \quad \forall \, u \in \mathcal{X}, \quad \forall \, \varepsilon \in (0, \varepsilon_0).
\end{equation}
By Definition~\ref{def:product_distribution} and Equations~\eqref{eq:coordinate_map_density} and~\eqref{eq:residual_shift_zero}, the product density satisfies
\begin{equation}
\label{eq:product_density_slice}
    \pipro(u)
    =\frac{p_{T_c(\boldsymbol x)}^{\mathbf{0}_q}(u)}{Z} \quad \forall \, u \in \mathcal{X}.
\end{equation}
By Equations~\eqref{eq:average_density_convergence}, \eqref{eq:z_epsilon_limit}, \eqref{eq:pi_epsilon_density}, and~\eqref{eq:product_density_slice}, and the triangle inequality, we then have that
\[
\begin{aligned}
    \mathrm{TV}(\pi_\varepsilon,\pipro)
    &=\frac12\left\|
        \frac{p_{T_c(\boldsymbol x)}^{S_\varepsilon}}{Z_\varepsilon}
        -\frac{p_{T_c(\boldsymbol x)}^{\mathbf{0}_q}}{Z}
        \right\|_{L^1}\\
    &\leq
        \frac{\left\|p_{T_c(\boldsymbol x)}^{S_\varepsilon}
        -p_{T_c(\boldsymbol x)}^{\mathbf{0}_q}\right\|_{L^1}}
             {2Z_\varepsilon}
        +\frac{|Z-Z_\varepsilon|}{2Z_\varepsilon}\\
    &\leq
        \frac{\left\|p_{T_c(\boldsymbol x)}^{S_\varepsilon}
        -p_{T_c(\boldsymbol x)}^{\mathbf{0}_q}\right\|_{L^1}}
             {Z_\varepsilon}
    \to 0 \quad \text{as } \varepsilon\to0.
\end{aligned}
\]
Hence, the conditional law of the consensus converges in total variation to the product law.
\end{proof}

\subsection{Weighted-Average Consensus}\label{app:linear_consensus_statements}

This section defines the weighted-average consensus and states a result used in Section~\ref{sec:asymptotics}. It uses the full discrepancy mapping $\delta$ of Definition~\ref{def:discrepancy}, which corresponds to choosing $\psi=\operatorname{id}$, and $q=(K-1)\,p$ is the dimension of its range.

\begin{definition}[Weighted-Average Consensus]
\label{def:weighted_average_consensus}
Let $W_1,\ldots,W_K\in\mathbb{R}^{p\times p}$ satisfy $\sum_{k=1}^K W_k=I_p$. The \emph{weighted-average consensus} with weights $W=(W_1,\ldots,W_K)$ is
\[
    c_W(\boldsymbol{x})=\sum_{k=1}^K W_kx_k
    \quad\forall\,\boldsymbol{x}\in\mathcal{X}^K.
\]
\end{definition}

It is a consensus in the sense of Definition~\ref{def:consensus}: it is continuous, $c_W(\boldsymbol{1}_K\otimes x)=\sum_{k=1}^K W_kx=x$, and $c_W(x_1+a,\ldots,x_K+a)=c_W(\boldsymbol{x})+a$ for every $a\in\mathcal{X}$.

\begin{proposition}
\label{prop:linear_consensus_discrepancy}
If $c=c_W$ and $d=\delta$, then we have
\[
    C \mid D = \mathbf{0}_{q} \sim \pipro.
\]
\end{proposition}

\subsection{Proof of Proposition~\ref{prop:linear_consensus_discrepancy}}\label{app:proof_linear_consensus_discrepancy}

The result follows from a linear change of coordinates.

\begin{proof}
Consider the following function that measures the consensus offset relative to the first component:
\[
    \gamma\left(\boldsymbol{v}\right)=\sum_{k=2}^{K} W_kv_{k-1} \quad \forall\,\boldsymbol{v}\in\mathbb{R}^{q}.
\]
By the definition of the linear consensus, we write the consensus using the consensus offset:
\[
    c(\boldsymbol{x})
    =\sum_{k=1}^{K}W_kx_k
    =x_1+\sum_{k=2}^{K}W_k(x_k-x_1)
    =x_1+\gamma\bigl(\delta(\boldsymbol{x})\bigr)
    \quad \forall\,\boldsymbol{x}\in\mathcal X^K.
\]
Now, consider the following functions that measure the residual shifts:
\[
    r_1(\boldsymbol{v})=-\gamma(\boldsymbol{v})=-\sum_{j=2}^K W_jv_{j-1},\qquad r_k(\boldsymbol{v})=v_{k-1}-\gamma(\boldsymbol{v})=v_{k-1}-\sum_{j=2}^K W_jv_{j-1},\quad k=2,\ldots,K,\quad \forall\,\boldsymbol{v}\in\mathbb{R}^{q}.
\]
Given our choices of $c$ and $d$, the Jacobian of the coordinate map can be written as follows:
\[
    J_{T_c}(\boldsymbol{x})
    =
    \begin{bmatrix}
        W_1 & W_2 & \cdots & W_K \\
        -\mathbf{1}_{K-1}\otimes I_p & \multicolumn{3}{c}{I_{K-1}\otimes I_p}
    \end{bmatrix}
    \quad \forall\,\boldsymbol{x}\in\mathcal X^K.
\]
Since adding columns $2,\ldots,K$ to column $1$ preserves the determinant, the first row becomes
\[
    W_1+W_2+\cdots+W_K=I_p
\]
by the assumption on the matrices $W_1,\ldots,W_K$. Doing the same on the other rows, we obtain
\begin{equation}
\label{eq:linear_consensus_transformed_jacobian}
    \begin{bmatrix}
        I_p & W_2 & \cdots & W_K \\
        0_{q\times p} & \multicolumn{3}{c}{I_q}
    \end{bmatrix},
\end{equation}
where we used the fact that $I_{K-1}\otimes I_p=I_q$. By Equation~\eqref{eq:linear_consensus_transformed_jacobian}, this matrix is block upper triangular, so
\[
    \det\!\left(J_{T_c}(\boldsymbol{x})\right)
    =\det\!\begin{pmatrix} I_p & W_2 & \cdots & W_K \\ 0_{q\times p} & \multicolumn{3}{c}{I_q} \end{pmatrix}
    =1
    \quad \forall\,\boldsymbol{x}\in\mathcal X^K.
\]
By Equation~\eqref{eq:coordinate_map_inverse} and the change-of-variables formula, the density of the coordinate map is
\[
    p_{T_c(\boldsymbol x)}(u,\boldsymbol v)
    =\prod_{k=1}^K\pi_k\bigl(u+r_k(\boldsymbol v)\bigr)
    \quad \forall\,(u,\boldsymbol v)\in\mathcal X\times\mathbb R^q.
\]
By Equation~\eqref{eq:coordinate_map_inverse}, when the difference coordinate vanishes, the inverse coordinate map recovers the diagonal:
\[
    T_c^{-1}(u,\mathbf{0}_q)
    =\boldsymbol{1}_K\otimes u
    \quad \forall\,u\in\mathcal X.
\]
By Equation~\eqref{eq:coordinate_map_density}, the marginal density of $D$ can be written as follows:
\[
    p_D(\boldsymbol{v})
    =\int_{\mathcal X}p_{T_c(\boldsymbol x)}(u,\boldsymbol{v})\,\mathrm{d}u
    \quad \forall\,\boldsymbol{v}\in\mathbb{R}^{q}.
\]
Consequently, by Equation~\eqref{eq:origin_density_slice} and Definition~\ref{def:product_distribution}, the marginal density of $D$ at the origin is
\begin{equation}
\label{eq:p_D_zero}
    p_D(\mathbf{0}_q)
    =\int_{\mathcal X}p_{T_c(\boldsymbol x)}(u,\mathbf{0}_q)\,\mathrm{d}u
    =\int_{\mathcal X}\prod_{k=1}^K\pi_k(u)\,\mathrm{d}u
    =Z>0.
\end{equation}
By Equation~\eqref{eq:p_D_zero}, we may evaluate the chosen conditional-density version at the origin:
\[
    p_{C\mid D}(u\mid\mathbf{0}_q)
    =\frac{p_{C,D}(u,\mathbf{0}_q)}{p_D(\mathbf{0}_q)}
    =\frac{p_{T_c(\boldsymbol x)}(u,\mathbf{0}_q)}{p_D(\mathbf{0}_q)}
    =\frac{\prod\limits_{k=1}^{K}\pi_k(u)}{Z}
    =\pipro(u) \quad\forall\,u\in\mathcal X.
\]
Hence, we conclude that the conditional law of the consensus given null discrepancy is the product law:
\[
    C\mid D=\mathbf{0}_q\sim\pipro.
\]
\end{proof}

\subsection{Gaussian Components}

Additional knowledge on the component distributions may yield stronger results. For example, with Gaussian components it is possible to pick a consensus that marginally recovers the product distribution and is independent of the discrepancy.

\begin{proposition}
\label{prop:gaussian_consensus_discrepancy}
Let the components be Gaussian:
\[
    \pi_k = \mathcal{N}\left(\mu_k,\Sigma_k\right) \quad k = 1, \ldots, K,
\]
where each $\mu_k\in\mathbb{R}^p$ and $\Sigma_k\in\mathbb{R}^{p\times p}$ is positive definite. Let $c=c_{W^\ast}$ and $d=\delta$, where
\[
    \Sigma_\pi^{-1} = \sum_{j=1}^{K}\Sigma_j^{-1},
    \qquad
    W_k^\ast = \Sigma_\pi\Sigma_k^{-1},
    \quad k = 1, \ldots, K.
\]
Then, we have that
\[
    C \indep D \text{ and } C \sim \pipro.
\]
\end{proposition}

\subsection{Proof of Proposition~\ref{prop:gaussian_consensus_discrepancy}}\label{app:proof_gaussian_consensus_discrepancy}

The result follows from standard properties of Gaussian distributions.

\begin{proof}
By Definition~\ref{def:product_distribution}, in this case the product distribution is also Gaussian:
\[
    \pipro=\mathcal{N}\left(\mu_\pi,\Sigma_\pi\right),
\]
where
\[
    \Sigma_\pi^{-1}=\sum_{k=1}^{K}\Sigma_k^{-1}
    \quad\text{and}\quad
    \mu_\pi=\Sigma_\pi\sum_{k=1}^{K}\Sigma_k^{-1}\mu_k.
\]
Since \(C\) is a linear combination of independent Gaussian random vectors, it is also Gaussian with mean
\[
    \begin{aligned}
        \E\left[C\right]
        &=\E\left[\sum_{k=1}^{K}W_k^\ast x_k\right] \\
        & =\sum_{k=1}^{K}W_k^\ast\mu_k \\
        & =\sum_{k=1}^{K}\Sigma_\pi\Sigma_k^{-1}\mu_k \\
        & =\mu_\pi,
    \end{aligned}
\]
and covariance
\[
    \begin{aligned}
        \operatorname{Cov}\left(C\right)
        &=\operatorname{Cov}\left(\sum_{k=1}^{K}W_k^\ast x_k\right) \\
        & =\sum_{k=1}^{K}W_k^\ast\Sigma_k(W_k^\ast)^\top \\
        & =\sum_{k=1}^{K}\Sigma_\pi\Sigma_k^{-1}\Sigma_k\Sigma_k^{-1}\Sigma_\pi \\
        & =\Sigma_\pi\left(\sum_{k=1}^{K}\Sigma_k^{-1}\right)\Sigma_\pi \\
        & =\Sigma_\pi\Sigma_\pi^{-1}\Sigma_\pi \\
        & =\Sigma_\pi.
    \end{aligned}
\]
Therefore, we conclude that the consensus is distributed according to the product distribution:
\[
    C \sim \pipro.
\]
Since $c$ and $d$ are linear maps, $C$ and $D$ are jointly Gaussian. Their cross-covariances satisfy
\[
    \operatorname{Cov}\left(C,x_k-x_1\right)
    =W_k^\ast\Sigma_k-W_1^\ast\Sigma_1
    =\Sigma_\pi-\Sigma_\pi
    =0 \quad k=2,\ldots,K.
\]
Thus, since $C$ and $D$ are jointly Gaussian and uncorrelated, they are also independent:
\[
    C \indep D.
\]
\end{proof}

\subsection{Proof of Theorem~\ref{thm:llr_rate}}\label{app:proof_llr_rate}

We first recall the general definition of the local linear regression estimator.

\begin{definition}[Local linear regression]
\label{def:llr}
Let the data be generated as follows:
\[
    \left(y^{(i)},x^{(i)}\right)\iid\left(y,x\right)\in\mathbb{R}^r\times\mathbb{R}^m,
    \quad i = 1, \ldots, S.
\]
A local linear regression estimates the conditional mean function at a given point:
\[
    f(x_0) = \E[y\mid x=x_0], \quad x_0\in\mathbb{R}^m.
\]
It does so by solving the following optimization problem:
\begin{equation}
\label{eq:llr_optimization}
    \min_{\alpha\in\mathbb{R}^r,\,\beta\in\mathbb{R}^{r\times m}}
    \sum_{i=1}^S
    \left\lVert y^{(i)}-\alpha-\beta\left(x^{(i)}-x_0\right)\right\rVert_2^2
    \mathcal{K}\left(\frac{x^{(i)}-x_0}{b_S}\right),
\end{equation}
where $\mathcal{K}:\mathbb{R}^m\to[0,\infty)$ is a kernel and $b_S>0$ is a bandwidth. Then, the resulting estimate is
\[
    \hat{f}_S(x_0) = \hat{\alpha}_S,
\]
where $\hat{\alpha}_S$ and $\hat{\beta}_S$ are the minimizers of Equation~\eqref{eq:llr_optimization}.
\end{definition}

We now impose some regularity assumptions for Theorem~\ref{thm:llr_rate}. The first concerns the density of $D$.

\begin{assumption}[Regular discrepancy density]
\label{ass:llr_discrepancy}
The density of $D$ is continuously differentiable near the origin:
\[
    \exists\, U \text{ open with } \mathbf{0}_m\in U: p_D\in C^1(U).
\]
Moreover, it is positive at the origin:
\[
    p_D(\mathbf{0}_m) > 0.
\]
\end{assumption}

The next two assumptions concern the smoothness of the conditional mean and covariance matrix.

\begin{assumption}[Smooth conditional mean]
\label{ass:llr_mu}
The function $f_h$ is twice continuously differentiable near the origin:
\[
    \exists\, U \text{ open with } \mathbf{0}_m\in U: f_h\in C^2(U).
\]
Moreover, its Laplacian is non-zero at the origin:
\[
    \Delta f_h(\mathbf{0}_m)\neq\mathbf{0}_r.
\]
\end{assumption}

\begin{assumption}[Smooth conditional covariance]
\label{ass:llr_covariance}
The function $\Sigma_h$ is continuous near the origin:
\[
    \exists\, U \text{ open with } \mathbf{0}_m\in U: \Sigma_h\in C^0(U).
\]
Moreover, its trace is positive at the origin:
\[
    \operatorname{tr}\!\left(\Sigma_h(\mathbf{0}_m)\right)>0.
\]
\end{assumption}

The following assumption specifies the regularity conditions on the kernel $\mathcal{K}$.

\begin{assumption}[Regular kernel]
\label{ass:llr_kernel}
The kernel $\mathcal{K}$ is bounded, compactly supported, and radially symmetric:
\[
    \exists \, k:[0,\infty) \to [0,\infty): \mathcal{K}(u) = k(\|u\|_2) \quad \forall\, u\in\mathbb{R}^m.
\]
Moreover, it is a probability density with finite second moment and finite roughness:
\[
    \int_{\mathbb{R}^m}\mathcal{K}(u)\,\mathrm{d}u=1,
    \qquad
    \int_{\mathbb{R}^m}u_1^2 \, \mathcal{K}(u)\,\mathrm{d}u<\infty,
    \qquad
    \int_{\mathbb{R}^m}\mathcal{K}(u)^2\,\mathrm{d}u<\infty.
\]
\end{assumption}

The final assumption specifies the asymptotic regime for the bandwidth.

\begin{assumption}[Bandwidth regime]
\label{ass:llr_bandwidth}
The bandwidth $b_S$ satisfies
\[
    b_S\overset{S}{\to}0,
    \qquad
    S \, b_S^m\overset{S}{\to}\infty.
\]
\end{assumption}

With these definitions and assumptions in place, we now prove Theorem~\ref{thm:llr_rate}.

\begin{proof}[Proof of Theorem~\ref{thm:llr_rate}]

By Assumptions~\ref{ass:llr_discrepancy}--\ref{ass:llr_bandwidth} and standard results for multivariate local linear regression \citep{ruppert_multivariate_1994}, we have an explicit expansion for the asymptotic MSE of $\hat{f}_{h,S}(\mathbf{0}_m)$.

Applying the local linear regression bias expansion componentwise \citep{ruppert_multivariate_1994} gives
\begin{equation}
\label{eq:llr_bias_expansion}
    \E\left[\hat{f}_{h,S}(\mathbf{0}_m)\mid D^{(1:S)}\right] - f_h(\mathbf{0}_m)
    =
    \frac{1}{2} \, b_S^2 \, \mu_2(\mathcal{K}) \, \Delta f_h(\mathbf{0}_m) + o_p\left(b_S^2\right),
\end{equation}
where $\mu_2(\mathcal{K})$ is the second moment of $\mathcal{K}$:
\[
    \mu_2(\mathcal{K}) = \int_{\mathbb{R}^m}u_1^2 \, \mathcal{K}(u)\,\mathrm{d}u.
\]
Recall that the variance of a linear projection is the corresponding quadratic form of the covariance matrix:
\begin{equation}
\label{eq:llr_variance_projection}
    \operatorname{Var}\!\left(
        a^\top\hat{f}_{h,S}(\mathbf{0}_m)\mid D^{(1:S)}
    \right)
    =
    a^\top
    \operatorname{Cov}\!\left(
        \hat{f}_{h,S}(\mathbf{0}_m)\mid D^{(1:S)}
    \right)
    a
    \qquad \forall\,a\in\mathbb{R}^r.
\end{equation}
Combining Equation~\eqref{eq:llr_variance_projection} with the scalar variance expansion \citep{ruppert_multivariate_1994} gives
\begin{equation}
\label{eq:llr_covariance_expansion}
    \operatorname{Cov}\left(\hat{f}_{h,S}(\mathbf{0}_m)\mid D^{(1:S)}\right)
    =
    \frac{R(\mathcal{K})}{S \, b_S^m \, p_D(\mathbf{0}_m)}\Sigma_h(\mathbf{0}_m)
    +o_p\!\left(\frac{1}{S b_S^m}\right),
\end{equation}
where $R(\mathcal{K})$ is the roughness of $\mathcal{K}$:
\[
    R(\mathcal{K}) = \int_{\mathbb{R}^m}\mathcal{K}(u)^2\,\mathrm{d}u.
\]
By Assumption~\ref{ass:llr_bandwidth} and Equation~\eqref{eq:llr_bias_expansion}, the bias vanishes in probability:
\begin{equation}
\label{eq:llr_conditional_bias_vanishes}
    \E\left[\hat{f}_{h,S}(\mathbf{0}_m)\mid D^{(1:S)}\right]
    -f_h(\mathbf{0}_m)
    \overset{p}{\to}\mathbf{0}_r.
\end{equation}
Likewise, by Assumption~\ref{ass:llr_bandwidth} and Equation~\eqref{eq:llr_covariance_expansion}, the total conditional variance vanishes in probability:
\begin{equation}
\label{eq:llr_total_conditional_variance_vanishes}
    \operatorname{tr}\!\left(
        \operatorname{Cov}\left(
            \hat{f}_{h,S}(\mathbf{0}_m)\mid D^{(1:S)}
        \right)
    \right)
    \overset{p}{\to}0.
\end{equation}
By Equation~\eqref{eq:llr_total_conditional_variance_vanishes} and Chebyshev's inequality conditionally on $D^{(1:S)}$, we have this convergence in probability:
\begin{equation}
\label{eq:llr_conditional_chebyshev}
    \Prob\!\left(
        \left\lVert
            \hat{f}_{h,S}(\mathbf{0}_m)
            -\E\left[\hat{f}_{h,S}(\mathbf{0}_m)\mid D^{(1:S)}\right]
        \right\rVert_2>\varepsilon
        \,\middle|\,
        D^{(1:S)}
    \right)
    \leq
    \frac{
        \operatorname{tr}\!\left(
            \operatorname{Cov}\left(
                \hat{f}_{h,S}(\mathbf{0}_m)\mid D^{(1:S)}
            \right)
        \right)
    }{
        \varepsilon^2
    }
    \overset{p}{\to}0
    \qquad \forall\,\varepsilon>0.
\end{equation}
Taking expectations in Equation~\eqref{eq:llr_conditional_chebyshev} yields
\begin{equation}
\label{eq:llr_centered_convergence}
    \hat{f}_{h,S}(\mathbf{0}_m)
    -\E\left[\hat{f}_{h,S}(\mathbf{0}_m)\mid D^{(1:S)}\right]
    \overset{p}{\to}\mathbf{0}_r.
\end{equation}
Combining Equations~\eqref{eq:llr_conditional_bias_vanishes} and~\eqref{eq:llr_centered_convergence}, Slutsky's theorem gives
\begin{equation}
\label{eq:llr_consistency_conditional_mean}
    \hat{f}_{h,S}(\mathbf{0}_m)
    \overset{p}{\to}
    f_h(\mathbf{0}_m).
\end{equation}
Since $h$ is bounded, Theorem~\ref{thm:consensus_discrepancy} implies convergence of the conditional expectation as the discrepancy shrinks:
\begin{equation}
\label{eq:llr_product_limit}
    \E_{\pi_\varepsilon}[h]
    \overset{\varepsilon}{\to}
    \E_{\pipro}[h].
\end{equation}
By the definitions in Theorems~\ref{thm:llr_rate} and~\ref{thm:consensus_discrepancy} and Assumption~\ref{ass:llr_mu},
\begin{equation}
\label{eq:llr_regression_limit}
    \E_{\pi_\varepsilon}[h]
    =
    \E\left[f_h(D)\mid\|D\|_\infty<\varepsilon\right]
    \overset{\varepsilon}{\to}
    f_h(\mathbf{0}_m).
\end{equation}
By the uniqueness of limits, Equations~\eqref{eq:llr_product_limit} and~\eqref{eq:llr_regression_limit} imply
\begin{equation}
\label{eq:llr_target_identity}
    f_h(\mathbf{0}_m)
    =
    \mu_h
    =
    \E_{\pipro}[h].
\end{equation}
Combining Equations~\eqref{eq:llr_consistency_conditional_mean} and~\eqref{eq:llr_target_identity} yields consistency of the estimator for the target:
\[
    \hat{f}_{h,S}(\mathbf{0}_m)
    \overset{p}{\to}
    \mu_h.
\]
After dropping the $o_p(b_S^4)$ and $o_p\!\left((S b_S^m)^{-1}\right)$ terms, Equations~\eqref{eq:llr_bias_expansion} and~\eqref{eq:llr_covariance_expansion} yield the leading-order AMSE:
\begin{equation}
\label{eq:llr_amse}
    \operatorname{AMSE}(b_S)
    =
    \frac{1}{4} \, b_S^4 \, \mu_2(\mathcal{K})^2 \, \left\lVert\Delta f_h(\mathbf{0}_m)\right\rVert_2^2
    +
    \frac{\operatorname{tr}\!\left(\Sigma_h(\mathbf{0}_m)\right) \, R(\mathcal{K})}{S \, b_S^m \, p_D(\mathbf{0}_m)}.
\end{equation}
The first-order condition for minimizing the AMSE in Equation~\eqref{eq:llr_amse} over $b_S$ gives the following equation:
\begin{equation}
\label{eq:llr_amse_foc}
    \frac{\mathrm{d}}{\mathrm{d}b_S}\operatorname{AMSE}(b_S)
    =
    b_S^3 \, \mu_2(\mathcal{K})^2 \, \left\lVert\Delta f_h(\mathbf{0}_m)\right\rVert_2^2
    -
    \frac{m \, \operatorname{tr}\!\left(\Sigma_h(\mathbf{0}_m)\right) \, R(\mathcal{K})}{S \, b_S^{m+1} \, p_D(\mathbf{0}_m)}
    =0.
\end{equation}
Multiplying Equation~\eqref{eq:llr_amse_foc} by $b_S^{m+1}$ and solving for $b_S^{m+4}$ gives
\[
    b_S^{m+4}
    =
    \frac{m \, \operatorname{tr}\!\left(\Sigma_h(\mathbf{0}_m)\right) \, R(\mathcal{K})}{\mu_2(\mathcal{K})^2 \, \left\lVert\Delta f_h(\mathbf{0}_m)\right\rVert_2^2 \, S \, p_D(\mathbf{0}_m)},
\]
so that the optimal bandwidth can be written as follows:
\begin{equation}
\label{eq:llr_optimal_bandwidth}
    b_S^\ast
    =
    \left(
        \frac{m \, \operatorname{tr}\!\left(\Sigma_h(\mathbf{0}_m)\right) \, R(\mathcal{K})}{\mu_2(\mathcal{K})^2 \, \left\lVert\Delta f_h(\mathbf{0}_m)\right\rVert_2^2 \, S \, p_D(\mathbf{0}_m)}
    \right)^{\frac{1}{m+4}}.
\end{equation}
By Equations~\eqref{eq:llr_amse_foc} and~\eqref{eq:llr_optimal_bandwidth}, at the optimal bandwidth we have
\begin{equation}
\label{eq:llr_optimal_balance}
    (b_S^\ast)^{m+4} \, \mu_2(\mathcal{K})^2 \, \left\lVert\Delta f_h(\mathbf{0}_m)\right\rVert_2^2 = \frac{m \, \operatorname{tr}\!\left(\Sigma_h(\mathbf{0}_m)\right) \, R(\mathcal{K})}{S \, p_D(\mathbf{0}_m)}.
\end{equation}
By Equations~\eqref{eq:llr_amse} and~\eqref{eq:llr_optimal_balance}, the variance term at $b_S^\ast$ can be rewritten as
\begin{equation}
\label{eq:llr_optimal_variance}
    \frac{\operatorname{tr}\!\left(\Sigma_h(\mathbf{0}_m)\right) \, R(\mathcal{K})}{S \, (b_S^\ast)^m \, p_D(\mathbf{0}_m)}
    =
    \frac{1}{m} \, (b_S^\ast)^4 \, \mu_2(\mathcal{K})^2 \, \left\lVert\Delta f_h(\mathbf{0}_m)\right\rVert_2^2.
\end{equation}
Substituting Equation~\eqref{eq:llr_optimal_variance} into Equation~\eqref{eq:llr_amse} gives
\begin{equation}
\label{eq:llr_optimal_amse_intermediate}
    \operatorname{AMSE}_S^\ast
    =
    \operatorname{AMSE}(b_S^\ast)
    =
    \frac{m+4}{4m}(b_S^\ast)^4 \, \mu_2(\mathcal{K})^2 \, \left\lVert\Delta f_h(\mathbf{0}_m)\right\rVert_2^2.
\end{equation}
Substituting Equation~\eqref{eq:llr_optimal_bandwidth} into Equation~\eqref{eq:llr_optimal_amse_intermediate} and collecting exponents gives
\[
    \operatorname{AMSE}_S^\ast
    =
    \mathcal{C}_{\mathcal{K}}
    \left[
        \frac{
            \operatorname{tr}\!\left(\Sigma_h(\mathbf{0}_m)\right)
            \left\lVert\Delta f_h(\mathbf{0}_m)\right\rVert_2^{\frac{m}{2}}
        }{
            S p_D(\mathbf{0}_m)
        }
    \right]^{\frac{4}{m+4}},
\]
where
\[
    \mathcal{C}_{\mathcal{K}}
    =
    \frac{m+4}{4\,m^{\frac{m}{m+4}}}
    \mu_2(\mathcal{K})^{\frac{2m}{m+4}}
    R(\mathcal{K})^{\frac{4}{m+4}}.
\]

\end{proof}

\subsection{Proof of Proposition~\ref{prop:gcs_decomposition}}\label{app:proof_gcs_decomposition}

We first introduce two measures of divergence between probability densities.

\begin{definition}[Cauchy--Schwarz divergence]
\label{def:cs_divergence}
The Cauchy--Schwarz (CS) divergence between two densities is
\[
    D_{\mathrm{CS}}(g_{1:2}) = -\log\left[\frac{\left(\displaystyle\int_{\mathcal{X}}g_1(x) \, g_2(x)\,\mathrm{d}x\right)^2}{\|g_1\|_{L^2}^2 \, \|g_2\|_{L^2}^2}\right],
\]
where $g_1$ and $g_2$ are $L^2$-integrable densities on $\mathcal{X}$.
\end{definition}

\begin{definition}[Generalized CS divergence]
\label{def:gcs_divergence}
The generalized CS divergence extends $D_{\mathrm{CS}}$ to $K$ densities:
\[
    D_{\mathrm{GCS}}(g_{1:K})
    =
    -\log\left[
        \frac{\left(\displaystyle\int_{\mathcal{X}}\prod_{k=1}^Kg_k(x)\,\mathrm{d}x\right)^K}
        {\displaystyle\prod_{k=1}^K\|g_k\|_{L^K}^K}
    \right],
\]
where $g_1,\ldots,g_K$ are $L^K$-integrable densities on $\mathcal{X}$.
\end{definition}

Now, consider these regularity assumptions on the component densities and the discrepancy transformation.

\begin{assumption}[$L^K$-integrable components]
\label{ass:gcs_integrability}
Each component density is $L^K$-integrable:
\[
    \|\pi_k\|_{L^K} < \infty \qquad k = 1, \ldots, K.
\]
\end{assumption}

\begin{assumption}[Regular discrepancy transformation]
\label{ass:regular_discrepancy_transformation}
The transformation $\psi$ preserves the dimension of $\delta$:
\[
    \psi:\mathbb{R}^q\to\mathbb{R}^q.
\]
It is continuously differentiable near the origin:
\[
    \exists\,U\text{ open with }\mathbf{0}_q\in U:\psi\in C^1(U).
\]
Moreover, its Jacobian is nonsingular at the origin:
\[
    \det\!\left(J_{\psi}(\mathbf{0}_q)\right)\neq0.
\]
Finally, $\psi$ is bounded away from the origin:
\[
    \inf_{\|\boldsymbol v\|_\infty\geq a}\|\psi(\boldsymbol v)\|_\infty>0
    \quad\forall\,a>0.
\]
\end{assumption}

With these definitions and assumptions in place, we can now prove Proposition~\ref{prop:gcs_decomposition}.

\begin{proof}[Proof of Proposition~\ref{prop:gcs_decomposition}]

By Equation~\eqref{eq:p_D_zero}, the density of the difference coordinates at the origin is
\begin{equation}
    \label{eq:difference_density_origin}
    p_{\delta(\boldsymbol x)}(\mathbf{0}_q)
    = Z
    = \int_{\mathcal{X}}\prod_{k=1}^K\pi_k(x)\,\mathrm{d}x.
\end{equation}
By Definition~\ref{def:discrepancy}, the transformation $\psi$ is origin-preserving:
\begin{equation}
    \label{eq:difference_discrepancy_origin}
    \psi(\mathbf{0}_q)=\mathbf{0}_q.
\end{equation}
By Assumption~\ref{ass:regular_discrepancy_transformation} and the inverse function theorem, $\psi$ admits a continuously differentiable local inverse:
\[
    \exists\,a>0 \text{ and } \psi^{-1}:\psi(B_a)\to B_a :
    \quad \psi^{-1}\in C^1\!\left(\psi(B_a)\right).
\]
Moreover, the Jacobian of the local inverse satisfies
\begin{equation}
    \label{eq:inverse_discrepancy_jacobian}
    J_{\psi^{-1}}(\mathbf{0}_q)
    =
    J_{\psi}(\mathbf{0}_q)^{-1}.
\end{equation}
By the bounded-away condition in Assumption~\ref{ass:regular_discrepancy_transformation}, we have
\[
    \forall\,a>0\quad\exists\,\varepsilon_a>0:
    \quad S_{\varepsilon_a}\subseteq B_a.
\]
By Equations~\eqref{eq:difference_density_origin}, \eqref{eq:difference_discrepancy_origin}, and~\eqref{eq:inverse_discrepancy_jacobian}, the change-of-variables formula gives
\begin{equation}
    \label{eq:transformed_discrepancy_density_origin}
    p_D(\mathbf{0}_q)
    =
    p_{\delta(\boldsymbol x)}(\mathbf{0}_q)
    \left|\det\!\left(J_{\psi^{-1}}(\mathbf{0}_q)\right)\right|
    =
    \frac{1}{\left|\det\!\left(J_{\psi}(\mathbf{0}_q)\right)\right|}
    \int_{\mathcal{X}}\prod_{k=1}^K\pi_k(x)\,\mathrm{d}x.
\end{equation}
By Definition~\ref{def:gcs_divergence}, we can express the exponentiated negative divergence as follows:
\begin{equation}
    \label{eq:gcs_exponentiated}
    \exp\left(-D_{\mathrm{GCS}}(\pi_{1:K})\right)
    =
    \frac{\left(\displaystyle\int_{\mathcal{X}}\prod_{k=1}^K\pi_k(x)\,\mathrm{d}x\right)^K}{\prod\limits_{k=1}^K\|\pi_k\|_{L^K}^K}.
\end{equation}
By Assumption~\ref{ass:gcs_integrability}, each $L^K$ norm is finite, so raising both sides of Equation~\eqref{eq:gcs_exponentiated} to the power $\frac{1}{K}$ gives
\begin{equation}
    \label{eq:gcs_overlap_identity}
    \exp\left(-\frac{1}{K}D_{\mathrm{GCS}}(\pi_{1:K})\right)
    =
    \frac{\displaystyle\int_{\mathcal{X}}\prod_{k=1}^K\pi_k(x)\,\mathrm{d}x}{\displaystyle\prod_{k=1}^K\|\pi_k\|_{L^K}}.
\end{equation}
Combining Equations~\eqref{eq:transformed_discrepancy_density_origin} and~\eqref{eq:gcs_overlap_identity} gives
\[
    p_D(\mathbf{0}_q)
    =
    \mathcal{C}_{\psi}
    \exp\left(-\frac{1}{K}D_{\mathrm{GCS}}(\pi_{1:K})\right)
    \prod_{k=1}^K\|\pi_k\|_{L^K},
\]
where
\[
    \mathcal{C}_{\psi}
    =
    \frac{1}{\left|\det\!\left(J_{\psi}(\mathbf{0}_q)\right)\right|}.
\]

\end{proof}

\subsection{Proof of the Entropy Lower Bound in Proposition~\ref{prop:gcs_decomposition}}
\label{app:proof_entropy_lower_bound}

We first introduce another measure of dispersion for probability densities.

\begin{definition}[Differential entropy]
\label{def:differential_entropy}
Let $g$ be a probability density on $\mathcal{X}$. Then, its differential
entropy is
\[
    H(g)=-\int_{\mathcal{X}}g(x)\log\left(g(x)\right)\,\mathrm{d}x.
\]
\end{definition}
Now consider the following regularity assumption on the component densities.
\begin{assumption}[Finite component entropies]
\label{ass:finite_component_entropies}
Each component density has finite differential entropy:
\[
    H(\pi_k)<\infty \qquad k = 1, \ldots, K.
\]
\end{assumption}
With these definitions and assumptions in place, we can now prove the entropy lower bound in Proposition~\ref{prop:gcs_decomposition}.
\begin{proof}[Proof of the entropy lower bound in Proposition~\ref{prop:gcs_decomposition}]
Since the logarithm is a concave function, by Jensen's inequality, we have that
\begin{equation}
\label{eq:entropy_jensen}
\begin{aligned}
    \log\left(\lVert\pi_k\rVert_{L^K}\right)
    &=\frac{1}{K}\,\log\left(\int_{\mathcal{X}}\pi_k(x)^K\,\mathrm{d}x\right) \\
    &=\frac{1}{K}\,\log\left(\int_{\mathcal{X}}
        \pi_k(x) \, \pi_k(x)^{K-1} \, \mathrm{d}x\right) \\
    &=\frac{1}{K}\,\log\left(
        \E_{x_k\sim\pi_k}\left[\pi_k(x_k)^{K-1}\right]
    \right) \\
    &\geq\frac{1}{K}\,\E_{x_k\sim\pi_k}
        \left[\log\left(\pi_k(x_k)^{K-1}\right)\right] \\
    &=\frac{1}{K}\,\int_{\mathcal{X}}\pi_k(x)
        \log\left(\pi_k(x)^{K-1}\right)\,\mathrm{d}x \\
    &=\frac{K-1}{K}\,\int_{\mathcal{X}}\pi_k(x)
        \log\left(\pi_k(x)\right)\,\mathrm{d}x \\
    &=-\frac{K-1}{K}\,H(\pi_k)
    \quad k = 1, \ldots, K.
\end{aligned}
\end{equation}
By Equation~\eqref{eq:entropy_jensen}, we can lower bound the $L^K$ norm of each component in terms of its differential entropy:
\begin{equation}
\label{eq:entropy_lp_bound}
    \lVert\pi_k\rVert_{L^K}
    \geq\exp\left(-\frac{K-1}{K}\,H(\pi_k)\right)
    \quad k = 1, \ldots, K.
\end{equation}
Multiplying Equation~\eqref{eq:entropy_lp_bound} over the components gives the following lower bound on the product of $L^K$ norms:
\begin{equation}
\label{eq:entropy_product_bound}
    \prod_{k=1}^{K}\lVert\pi_k\rVert_{L^K}
    \geq
    \exp\left(
        -\frac{K-1}{K}\,\sum_{k=1}^{K}H(\pi_k)
    \right).
\end{equation}
Combining Equation~\eqref{eq:entropy_product_bound} with
Proposition~\ref{prop:gcs_decomposition} gives
\[
    p_D(\mathbf{0}_q)
    \geq
    \mathcal{C}_{\psi}
    \exp\left(
        -\frac{1}{K}\,D_{\mathrm{GCS}}(\pi_{1:K})
        -\frac{K-1}{K}\,\sum_{k=1}^{K}H(\pi_k)
    \right).
\]
\end{proof}

\subsection{Proof of Proposition~\ref{prop:gaussian_amse}}
\label{app:proof_gaussian_amse}

We begin by introducing a regularity assumption on the kernel.

\begin{assumption}[Regular kernel]
\label{ass:gaussian_kernel}
The kernel $\mathcal{K}$ is bounded and continuous at the origin:
\[
    \sup_{u\in\mathbb{R}^m}\left|\mathcal{K}(u)\right|<\infty,
    \qquad
    \lim_{u\to\mathbf{0}_m}\mathcal{K}(u)
    =\mathcal{K}(\mathbf{0}_m)>0.
\]
\end{assumption}

With this assumption in place, we can now prove Proposition~\ref{prop:gaussian_amse}.

\begin{proof}

Since the components are Gaussian, we can derive the asymptotic MSE without relying on Theorem~\ref{thm:llr_rate}.

\textit{Full difference discrepancy.} We first consider the case $c=c_W$ and $d=\delta$.

Since $C$ and $D$ are linear transformations of the jointly Gaussian vector $\boldsymbol{x}$, they are also jointly Gaussian:
\begin{equation}
\label{eq:gaussian_joint_consensus_discrepancy}
    \begin{pmatrix}
        C \\
        D
    \end{pmatrix}
    \sim
    \mathcal{N}\left(
        \begin{pmatrix}
            \mu_C \\
            \mu_D
        \end{pmatrix},
        \begin{pmatrix}
            \Sigma_C & \Sigma_{CD} \\
            \Sigma_{CD}^\top & \Sigma_D
        \end{pmatrix}
    \right),
\end{equation}
where the means and covariances are defined as follows:
\[
    \mu_C = \E[C],
    \qquad
    \mu_D = \E[D],
    \qquad
    \Sigma_C = \operatorname{Cov}(C),
    \qquad
    \Sigma_D = \operatorname{Cov}(D),
    \qquad
    \Sigma_{CD} = \operatorname{Cov}(C,D).
\]
The Jacobian of the difference map is constant:
\begin{equation}
\label{eq:gaussian_difference_jacobian}
    J_\delta(\boldsymbol{x})
    =
    \begin{pmatrix}
        -\mathbf{1}_{K-1}\otimes I_p & I_q
    \end{pmatrix}
    \quad \forall\,\boldsymbol{x}\in\mathcal{X}^K.
\end{equation}
By Equation~\eqref{eq:gaussian_difference_jacobian}, the last $q$ columns of $J_\delta$ form the identity matrix $I_q$. Thus, the Jacobian has full row rank:
\[
    \operatorname{rank}\!\left(J_\delta(\boldsymbol{x})\right)
    =q
    \quad \forall\,\boldsymbol{x}\in\mathcal{X}^K.
\]
Since each component covariance matrix is positive definite, so is the covariance of the difference coordinates:
\begin{equation}
\label{eq:gaussian_discrepancy_covariance}
    \Sigma_D
    =
    J_\delta(\boldsymbol{x})
    \operatorname{diag}(\Sigma_1,\ldots,\Sigma_K)
    J_\delta(\boldsymbol{x})^\top
    \succ 0
    \quad \forall\,\boldsymbol{x}\in\mathcal{X}^K.
\end{equation}
By the properties of the multivariate Gaussian distribution, we have the following:
\begin{equation}
\label{eq:gaussian_conditional_distribution}
    C\mid D=d
    \sim
    \mathcal{N}\left(
        \mu_C+\Sigma_{CD}\Sigma_D^{-1}(d-\mu_D),
        \,
        \Sigma_C-\Sigma_{CD}\Sigma_D^{-1}\Sigma_{CD}^\top
    \right).
\end{equation}
Recall that when the discrepancy is null, Proposition~\ref{prop:linear_consensus_discrepancy} gives
\begin{equation}
\label{eq:gaussian_conditional_distribution_at_zero}
    C\mid D=\mathbf{0}_q
    \sim
    \pipro
    =
    \mathcal{N}(\mu_\pi,\Sigma_\pi),
\end{equation}
where $\mu_\pi$ and $\Sigma_\pi$ are as in Proposition~\ref{prop:gaussian_consensus_discrepancy}. Then, by
comparing Equations~\eqref{eq:gaussian_conditional_distribution} and~\eqref{eq:gaussian_conditional_distribution_at_zero}, we obtain the following:
\begin{equation}
\label{eq:gaussian_conditional_parameters}
    \begin{aligned}
        \mu_C-\Sigma_{CD}\Sigma_D^{-1}\mu_D = \mu_\pi,
        \qquad
        \Sigma_C-\Sigma_{CD}\Sigma_D^{-1}\Sigma_{CD}^\top = \Sigma_\pi.
    \end{aligned}
\end{equation}
By Equations~\eqref{eq:gaussian_conditional_distribution} and~\eqref{eq:gaussian_conditional_parameters}, we can specify the following linear model:
\begin{equation}
\label{eq:gaussian_linear_model}
    C
    =
    \mu_\pi+\Sigma_{CD}\Sigma_D^{-1}D+\eta,
    \qquad
    \eta\sim\mathcal{N}(\mathbf{0}_p,\Sigma_\pi),
    \qquad
    \eta\indep D.
\end{equation}
Since $h=\operatorname{id}$, the conditional mean function is affine in the discrepancy:
\begin{equation}
\label{eq:gaussian_conditional_mean}
    f_h(d)
    =
    \mu_\pi+\Sigma_{CD}\Sigma_D^{-1}d \quad \forall \, d \in \mathbb{R}^q.
\end{equation}
\textit{Scalar discrepancy.} We now consider the case $c=c_{W^\ast}$ and $d=\|\delta\|_2^2$.

By Proposition~\ref{prop:gaussian_consensus_discrepancy}, the consensus is independent of the difference discrepancy and has the product distribution:
\begin{equation}
\label{eq:gaussian_difference_independence}
    C\indep\delta(\boldsymbol{x}),
    \qquad
    C\sim\pipro =\mathcal{N}(\mu_\pi,\Sigma_\pi).
\end{equation}
By Equation~\eqref{eq:gaussian_difference_independence}, the same relationship holds after applying the squared Euclidean norm to the difference map:
\begin{equation}
\label{eq:gaussian_radial_independence}
    C\indep D,
    \qquad
    C\sim\pipro=\mathcal{N}(\mu_\pi,\Sigma_\pi).
\end{equation}
By Equation~\eqref{eq:gaussian_joint_consensus_discrepancy}, the difference discrepancy has a multivariate Gaussian distribution:
\begin{equation}
\label{eq:gaussian_difference_distribution}
    \delta(\boldsymbol{x})
    \sim
    \mathcal{N}(\mu_\delta,\Sigma_\delta),
\end{equation}
where
\[
    \mu_\delta
    =
    \delta\bigl((\mu_1,\ldots,\mu_K)\bigr),
    \qquad
    \Sigma_\delta
    =
    \operatorname{diag}(\Sigma_2,\ldots,\Sigma_K)
    +
    \left(
        \mathbf{1}_{K-1}\mathbf{1}_{K-1}^\top
    \right)\otimes\Sigma_1.
\]
By Equation~\eqref{eq:gaussian_difference_distribution} and the standard moments of a Gaussian quadratic form, we have the following:
\begin{equation}
\label{eq:gaussian_radial_discrepancy_raw_moments}
    \E[D]
    =
    \operatorname{tr}(\Sigma_\delta)
    +\|\mu_\delta\|_2^2,
    \qquad
    \E[D^2]
    =
    \left(
        \operatorname{tr}(\Sigma_\delta)
        +\|\mu_\delta\|_2^2
    \right)^2
    +2\operatorname{tr}(\Sigma_\delta^2)
    +4\mu_\delta^\top\Sigma_\delta\mu_\delta.
\end{equation}
Therefore, by Equation~\eqref{eq:gaussian_radial_discrepancy_raw_moments}, the mean and variance of the scalar discrepancy are
\begin{equation}
\label{eq:gaussian_radial_discrepancy_moments}
    \mu_D
    =
    \operatorname{tr}(\Sigma_\delta)
    +\|\mu_\delta\|_2^2,
    \qquad
    \Sigma_D
    =
    2\operatorname{tr}(\Sigma_\delta^2)
    +4\mu_\delta^\top\Sigma_\delta\mu_\delta
    >0.
\end{equation}
By Equation~\eqref{eq:gaussian_radial_independence}, we can specify the following linear model:
\begin{equation}
\label{eq:gaussian_radial_linear_model}
    C
    =
    \mu_\pi+\eta,
    \qquad
    \eta\sim\mathcal{N}(\mathbf{0}_p,\Sigma_\pi),
    \qquad
    \eta\indep D.
\end{equation}
Since $h=\operatorname{id}$, the conditional mean function is constant in the discrepancy:
\begin{equation}
\label{eq:gaussian_radial_conditional_mean}
    f_h(d)
    =
    \mu_\pi
    \quad \forall\,d\in[0,\infty).
\end{equation}
\textit{Common argument.} Equations~\eqref{eq:gaussian_linear_model}--\eqref{eq:gaussian_conditional_mean} and~\eqref{eq:gaussian_radial_linear_model}--\eqref{eq:gaussian_radial_conditional_mean} show that, in both cases, the local linear model is correctly specified. Hence, whenever the weighted design matrix is invertible, the estimator is unbiased:
\begin{equation}
\label{eq:gaussian_conditional_unbiasedness}
    \E\left[
        \hat{f}_{h,S}(\mathbf{0}_m)
        \mid D^{(1:S)}
    \right]
    =
    \mu_\pi.
\end{equation}
In other words, by Equation~\eqref{eq:gaussian_conditional_unbiasedness}, the bias is null:
\[
    \operatorname{Bias}\left(
        \hat{f}_{h,S}(\mathbf{0}_m)
        \mid D^{(1:S)}
    \right)
    =
    \mathbf{0}_p.
\]
By Equation~\eqref{eq:gaussian_conditional_unbiasedness}, the conditional MSE equals the trace of the conditional covariance:
\begin{equation}
\label{eq:gaussian_conditional_mse}
    \E\left[
        \left\lVert
            \hat{f}_{h,S}(\mathbf{0}_m)-\mu_\pi
        \right\rVert_2^2
        \mid D^{(1:S)}
    \right]
    =
    \operatorname{tr}\left(
        \operatorname{Cov}\left[
            \hat{f}_{h,S}(\mathbf{0}_m)
            \mid D^{(1:S)}
        \right]
    \right).
\end{equation}
By Equations~\eqref{eq:gaussian_linear_model} and~\eqref{eq:gaussian_radial_linear_model} and the Gauss--Markov theorem, the asymptotic MSE is minimized when
\[
    \mathcal{K}\left(\frac{D^{(1)}}{b_S}\right)
    =
    \cdots
    =
    \mathcal{K}\left(\frac{D^{(S)}}{b_S}\right).
\]
Thus, we may choose an asymptotically optimal bandwidth sequence diverging to infinity:
\begin{equation}
\label{eq:gaussian_optimal_bandwidth}
    b_S^\ast\overset{S}{\to}\infty.
\end{equation}
In fact, by Assumption~\ref{ass:gaussian_kernel} and Equation~\eqref{eq:gaussian_optimal_bandwidth}, the kernel weights converge in probability to a constant:
\begin{equation}
\label{eq:gaussian_kernel_convergence}
    \mathcal{K}\left(\frac{D}{b_S^\ast}\right)
    \overset{p}{\to}\mathcal{K}(\mathbf{0}_m).
\end{equation}
Equations~\eqref{eq:gaussian_discrepancy_covariance} and~\eqref{eq:gaussian_radial_discrepancy_moments} show that the discrepancy covariance is positive definite in both cases:
\begin{equation}
\label{eq:gaussian_discrepancy_covariance_positive}
    \Sigma_D\succ0.
\end{equation}
By Equations~\eqref{eq:gaussian_joint_consensus_discrepancy} and~\eqref{eq:gaussian_difference_distribution}, the discrepancy has a finite fourth moment in both cases:
\[
    \E\left[\lVert D\rVert_2^4\right]<\infty.
\]
Moreover, $\mathcal{K}$ is bounded by Assumption~\ref{ass:gaussian_kernel}. Thus, Equation~\eqref{eq:gaussian_kernel_convergence} and the weak law of large numbers give
\begin{equation}
\label{eq:gaussian_weighted_design_limits}
    \frac{1}{S}\sum_{i=1}^{S}
    \mathcal{K}\left(\frac{D^{(i)}}{b_S^\ast}\right)^j
    \begin{pmatrix}
        1 \\
        D^{(i)}
    \end{pmatrix}
    \begin{pmatrix}
        1 \\
        D^{(i)}
    \end{pmatrix}^\top
    \overset{p}{\to}
    \mathcal{K}(\mathbf{0}_m)^j
    \begin{pmatrix}
        1 & \mu_D^\top \\
        \mu_D & \Sigma_D+\mu_D\mu_D^\top
    \end{pmatrix} \quad j = 1,2.
\end{equation}
By Equations~\eqref{eq:gaussian_discrepancy_covariance_positive} and~\eqref{eq:gaussian_weighted_design_limits}, the weighted design matrix is asymptotically invertible. Thus, the weighted least squares covariance formula and the continuous mapping theorem yield:
\begin{equation}
\label{eq:gaussian_ols_covariance_equivalence}
    \operatorname{Cov}\left[
        \hat{f}_{h,S}(\mathbf{0}_m)
        \mid D^{(1:S)}
    \right]
    =
    \operatorname{Cov}\left[
        \hat{f}_{h,S}^{\mathrm{OLS}}(\mathbf{0}_m)
        \mid D^{(1:S)}
    \right]
    +o_p\left(\frac{1}{S}\right).
\end{equation}
The standard random-design OLS covariance formula gives
\begin{equation}
\label{eq:ols_intercept_covariance}
    \operatorname{Cov}\left[
        \hat{f}_{h,S}^{\mathrm{OLS}}(\mathbf{0}_m)
        \mid D^{(1:S)}
    \right]
    =
    \frac{\Sigma_\pi}{S}
    \left(
        1+\mu_D^\top\Sigma_D^{-1}\mu_D
    \right)
    +o_p\left(\frac{1}{S}\right).
\end{equation}
Then, by Equations~\eqref{eq:gaussian_ols_covariance_equivalence} and~\eqref{eq:ols_intercept_covariance}, the conditional covariance of our estimator is
\begin{equation}
\label{eq:gaussian_ols_covariance}
    \operatorname{Cov}\left[
        \hat{f}_{h,S}(\mathbf{0}_m)
        \mid D^{(1:S)}
    \right]
    =
    \frac{\Sigma_\pi}{S}
    \left(
        1+\mu_D^\top\Sigma_D^{-1}\mu_D
    \right)
    +o_p\left(\frac{1}{S}\right).
\end{equation}
Therefore, Equations~\eqref{eq:gaussian_conditional_mse} and~\eqref{eq:gaussian_ols_covariance} give the optimized leading-order asymptotic MSE:
\[
    \operatorname{AMSE}_S^\ast
    =
    \frac{\operatorname{tr}(\Sigma_\pi)}{S}
    \left(1+\mathcal{C}_D\right),
\]
where
\[
    \mathcal{C}_D
    =
    \mu_D^\top\Sigma_D^{-1}\mu_D
    \geq 0.
\]
Finally, Equations~\eqref{eq:gaussian_conditional_unbiasedness} and~\eqref{eq:gaussian_ols_covariance} and Chebyshev's inequality give the claimed consistency:
\[
    \hat{f}_{h,S}(\mathbf{0}_m)
    \overset{p}{\to}
    \mu_h
    =
    \E_{\pipro}[h].
\]

\end{proof}

\section{Additional Experimental Results}\label{app:additional_experimental_results}

This section presents supplementary empirical results and details of our computational setup. Section~\ref{app:additional_simulation_results} reports further experimental findings across simulation settings. Section~\ref{app:additional_real_data_results} reports additional real-data results. Section~\ref{app:fusion_baseline_comparison} compares the fusion methods in terms of estimation accuracy and computational cost. Section~\ref{app:computational_setup} details the hardware and software used.

\begingroup

\renewcommand{\convergencefiguretarget}[6]{
    \begin{figure}[H]
        \def\convergencefunction{#3}
        \centering
        \includegraphics{#1/h_#2.pdf}
        \caption{#6}
        \label{fig:#5}
    \end{figure}
}

\renewcommand{\singlefigure}[3]{
    \begin{figure}[H]
        \centering
        \sbox\singlebox{\includegraphics{#1}}
        \ifdim\wd\singlebox>\linewidth
            \resizebox{\linewidth}{!}{\usebox\singlebox}
        \else
            \usebox\singlebox
        \fi
        \caption{#3}
        \label{fig:#2}
    \end{figure}
}

\subsection{Additional Simulation Results}\label{app:additional_simulation_results}

This section presents the remaining simulation results across data-generating processes and experimental settings, complementing those in the main text. We begin with a complete overview of the configurations considered:

\begin{itemize}
    \item \textbf{Gaussian}:
    \begin{itemize}
        \item \textbf{Univariate}:
        \begin{itemize}
            \renewcommand{\labelitemiii}{$\circ$}
            \item Location and scale.
            \item Number of components.
        \end{itemize}
        \item \textbf{Bivariate}:
        \begin{itemize}
            \renewcommand{\labelitemiii}{$\circ$}
            \item Mean separation.
            \item Isotropic variance scale.
            \item Marginal variance imbalance.
            \item Correlation strength and sign.
        \end{itemize}
        \item \textbf{Multivariate}:
        \begin{itemize}
            \renewcommand{\labelitemiii}{$\circ$}
            \item Dimension.
        \end{itemize}
    \end{itemize}
    \item \textbf{Gaussian mixtures}:
    \begin{itemize}
        \item Symmetry.
        \item Mixture weights.
        \item Mixture scales.
        \item Number of components.
        \item Heterogeneity.
    \end{itemize}
    \item \textbf{Beta}:
    \begin{itemize}
        \item Concentration.
        \item Heterogeneity.
        \item Skewness.
        \item Boundary singularities.
    \end{itemize}
\end{itemize}

\FloatBarrier

\subsubsection{Univariate Gaussian: Location and Scale}\label{app:univariate_gaussian_location_scale}

We consider the product of $K=2$ univariate Gaussian components with fixed first component and varying second component in terms of location and scale. The components are defined as follows:
\[
    \pi_1=\mathcal{N}(0, 1),
    \qquad
    \pi_2=\mathcal{N}(\mu_2,\sigma_2^2),
    \qquad
    \mu_2,\sigma_2\in\left\{1,2,3,4\right\}.
\]
It is easy to verify that the product of these two components is again a Gaussian distribution:
\[
    \pipro
    =
    \mathcal{N}\left(\frac{\mu_2}{1+\sigma_2^2}, \frac{\sigma_2^2}{1+\sigma_2^2}\right).
\]

The results are shown in Figures~\ref{fig:univariate_convergence_full} and~\ref{fig:univariate_convergence_full_square} for $h(u)=u$ and $h(u)=u^2$, respectively.

\convergencefigure
    {gaussian/univariate/k_2/location_scale_grid}
    {univariate_full}
    {univariate_convergence_full}
    {MSE and runtime versus $S$ for $h(u)=\convergencefunction$ with $K=2$ univariate Gaussian components. Columns and rows vary $\mu_2,\sigma_2\in\left\{1,2,3,4\right\}$, respectively. The first component is fixed with $\mu_1=0$ and $\sigma_1=1$.}

Across the settings, PMC-OLS closely matches the Oracle in terms of MSE for $h(u)=u$, as expected given its correct specification, but fails for $h(u)=u^2$, for which the linear model is misspecified. PMC-MLP and PMC-GP initially lag behind GBF but overtake it as the sample size increases for $h(u)=u$; for $h(u)=u^2$, only PMC-GP overtakes GBF within the range of sample sizes shown, while PMC-MLP reaches parity at about $S=10^6$. In terms of training time, all PMC methods are considerably faster than GBF, with PMC-OLS being the fastest.

\FloatBarrier

\subsubsection{Univariate Gaussian: Number of Components}\label{app:univariate_gaussian_number_of_components}

We consider the product of $K$ univariate Gaussian components with common location and either homogeneous or heterogeneous scales, with varying $K$. In the homogeneous case, the components are defined as follows:
\[
    \pi_k=\mathcal{N}(0,K),
    \qquad
    k = 1, \ldots, K,
    \qquad
    K\in\left\{2,3,5,10,20\right\}.
\]
In the heterogeneous case, the components are defined as follows:
\[
    \pi_k=\mathcal{N}(0,\sigma_k^2),
    \qquad
    \frac{1}{\sigma_k^2}=\frac{1}{2K}+\frac{k-1}{K(K-1)},
    \qquad
    k = 1, \ldots, K,
    \qquad
    K\in\left\{2,3,5,10,20\right\}.
\]
It is easy to verify that, in both cases, the product of these components is again a Gaussian distribution:
\[
    \pipro
    =
    \mathcal{N}(0, 1).
\]

The results are shown in Figures~\ref{fig:k_scaling_homogeneous_identity} and~\ref{fig:k_scaling_homogeneous_square} for the homogeneous case and in Figures~\ref{fig:k_scaling_heterogeneous_identity} and~\ref{fig:k_scaling_heterogeneous_square} for the heterogeneous case, for $h(u)=u$ and $h(u)=u^2$, respectively.

\singlefigure
    {gaussian/univariate/k_scaling/homogeneous/k_scaling_h_identity.pdf}
    {k_scaling_homogeneous_identity}
    {MSE and runtime versus $S$ for $h(u)=u$ with $K$ homogeneous univariate Gaussian components. Columns vary $K\in\left\{2,3,5,10,20\right\}$, and rows show MSE and runtime, respectively.}

\singlefigure
    {gaussian/univariate/k_scaling/homogeneous/k_scaling_h_square.pdf}
    {k_scaling_homogeneous_square}
    {MSE and runtime versus $S$ for $h(u)=u^2$ with $K$ homogeneous univariate Gaussian components. Columns vary $K\in\left\{2,3,5,10,20\right\}$, and rows show MSE and runtime, respectively.}

\singlefigure
    {gaussian/univariate/k_scaling/heterogeneous/k_scaling_h_identity.pdf}
    {k_scaling_heterogeneous_identity}
    {MSE and runtime versus $S$ for $h(u)=u$ with $K$ heterogeneous univariate Gaussian components. Columns vary $K\in\left\{2,3,5,10,20\right\}$, and rows show MSE and runtime, respectively.}

\singlefigure
    {gaussian/univariate/k_scaling/heterogeneous/k_scaling_h_square.pdf}
    {k_scaling_heterogeneous_square}
    {MSE and runtime versus $S$ for $h(u)=u^2$ with $K$ heterogeneous univariate Gaussian components. Columns vary $K\in\left\{2,3,5,10,20\right\}$, and rows show MSE and runtime, respectively.}

Across the settings, PMC-OLS closely matches the Oracle in terms of MSE for $h(u)=u$, as expected given its correct specification, but fails for $h(u)=u^2$, for which the linear model is misspecified. For $K\leq 5$, PMC-MLP and PMC-GP overtake GBF for $S\geq 5\times10^5$ when $h(u)=u$; for larger $K$, and for $h(u)=u^2$ with $K\geq 5$, they do not overtake GBF within the range of sample sizes shown. In terms of training time, all PMC methods are considerably faster than GBF, with PMC-OLS being the fastest. In both the homogeneous and heterogeneous cases, the runtime of GBF increases with $K$, while the runtime of PMC-MLP and PMC-GP remains largely unaffected by $K$ and that of PMC-OLS grows mildly but stays orders of magnitude below GBF.

\FloatBarrier

\subsubsection{Bivariate Gaussian: Mean Separation}\label{app:bivariate_gaussian_mean_separation}

We consider the product of $K=2$ bivariate Gaussian components with fixed first component and varying second component in terms of location along the diagonal. The components are defined as follows:
\[
    \pi_1=\mathcal{N}(\mathbf{0}_2, I_2),
    \qquad
    \pi_2=\mathcal{N}(\mu\mathbf{1}_2, I_2),
    \qquad
    \mu\in\left\{1,2,3\right\}.
\]
It is easy to verify that the product of these two components is again a Gaussian distribution:
\[
    \pipro
    =
    \mathcal{N}\left(\frac{\mu}{2}\mathbf{1}_2, \frac{1}{2}I_2\right).
\]

The results are shown in Figures~\ref{fig:bivariate_mean_shift_convergence} and~\ref{fig:bivariate_mean_shift_convergence_square} for $h(u)=u$ and $h(u)=u\odot u$, respectively.

\convergencefigure[u\odot u]
    {gaussian/bivariate/k_2/mean_shift}
    {bivariate_mean_shift}
    {bivariate_mean_shift_convergence}
    {MSE and runtime versus $S$ for $h(u)=\convergencefunction$ with $K=2$ bivariate Gaussian components. Columns vary $\mu\in\left\{1,2,3\right\}$. The first component is fixed with $\mu=0$.}

Across the settings, PMC-OLS closely matches the Oracle in terms of MSE. PMC-MLP and PMC-GP lag behind GBF by a factor of about 2 to 5 in MSE up to $S=10^5$, with the gap growing with the location separation; they are comparable to or better than GBF at $S=10^6$ only for the smaller separations. In terms of training time, all PMC methods are considerably faster than GBF, with PMC-OLS being the fastest.

\FloatBarrier

\subsubsection{Bivariate Gaussian: Isotropic Variance Scale}\label{app:bivariate_gaussian_isotropic_variance_scale}

We consider the product of $K=2$ bivariate Gaussian components with fixed first component and varying second component in terms of isotropic scale. The components are defined as follows:
\[
    \pi_1=\mathcal{N}(\mathbf{0}_2, I_2),
    \qquad
    \pi_2=\mathcal{N}(\mathbf{1}_2, \sigma^2 I_2),
    \qquad
    \sigma\in\left\{1,2,3,4\right\}.
\]
It is easy to verify that the product of these two components is again a Gaussian distribution:
\[
    \pipro
    =
    \mathcal{N}\left(\frac{1}{1+\sigma^2}\mathbf{1}_2, \frac{\sigma^2}{1+\sigma^2}I_2\right).
\]

The results are shown in Figures~\ref{fig:bivariate_variance_scaling_convergence} and~\ref{fig:bivariate_variance_scaling_convergence_square} for $h(u)=u$ and $h(u)=u\odot u$, respectively.

\convergencefigure[u\odot u]
    {gaussian/bivariate/k_2/variance_scaling}
    {bivariate_variance_scaling}
    {bivariate_variance_scaling_convergence}
    {MSE and runtime versus $S$ for $h(u)=\convergencefunction$ with $K=2$ bivariate Gaussian components. Columns vary $\sigma\in\left\{1,2,3,4\right\}$. The first component is fixed with $\sigma=1$.}

Across the settings, PMC-OLS closely matches the Oracle in terms of MSE for $h(u)=u$, as expected given its correct specification, but fails for $h(u)=u\odot u$, for which the linear model is misspecified. PMC-GP overtakes GBF for $S\geq 5\times10^5$, while PMC-MLP only reaches parity at $S=10^6$. In terms of training time, all PMC methods are considerably faster than GBF, with PMC-OLS being the fastest. Overall, it looks like increasing the scale $\sigma$ of the second component makes the estimation problem harder for PMC-MLP, and for PMC-OLS when $h(u)=u\odot u$, at all sample sizes shown.

\FloatBarrier

\subsubsection{Bivariate Gaussian: Marginal Variance Imbalance}\label{app:bivariate_gaussian_marginal_variance_imbalance}

We consider the product of $K=2$ bivariate Gaussian components with fixed first component and varying second component in terms of the variance of its second coordinate. The components are defined as follows:
\[
    \pi_1=\mathcal{N}(\mathbf{0}_2, I_2),
    \qquad
    \pi_2=\mathcal{N}\left(\mathbf{1}_2, \operatorname{diag}(1,\sigma_2^2)\right),
    \qquad
    \sigma_2\in\left\{1,2,3,4\right\}.
\]
It is easy to verify that the product of these two components is again a Gaussian distribution:
\[
    \pipro
    =
    \mathcal{N}\left(\left(\frac{1}{2},\frac{1}{1+\sigma_2^2}\right)^\top, \operatorname{diag}\left(\frac{1}{2},\frac{\sigma_2^2}{1+\sigma_2^2}\right)\right).
\]

The results are shown in Figures~\ref{fig:bivariate_anisotropic_variance_convergence} and~\ref{fig:bivariate_anisotropic_variance_convergence_square} for $h(u)=u$ and $h(u)=u\odot u$, respectively.

\convergencefigure[u\odot u]
    {gaussian/bivariate/k_2/anisotropic_variance}
    {bivariate_anisotropic_variance}
    {bivariate_anisotropic_variance_convergence}
    {MSE and runtime versus $S$ for $h(u)=\convergencefunction$ with $K=2$ bivariate Gaussian components. Columns vary $\sigma_2\in\left\{1,2,3,4\right\}$. The first component is fixed with $\sigma_2=1$.}

Across the settings, PMC-OLS closely matches the Oracle in terms of MSE for $h(u)=u$, as expected given its correct specification, but fails for $h(u)=u\odot u$, for which the linear model is misspecified. PMC-GP overtakes GBF for $S\geq 5\times10^5$, while PMC-MLP only reaches parity for $h(u)=u$ and remains behind for $h(u)=u\odot u$. In terms of training time, all PMC methods are considerably faster than GBF, with PMC-OLS being the fastest.

\FloatBarrier

\subsubsection{Bivariate Gaussian: Correlation Strength and Sign}\label{app:bivariate_gaussian_correlation_strength_and_sign}

We consider the product of $K=2$ bivariate Gaussian components with fixed first component and varying second component in terms of the strength and sign of its correlation. The components are defined as follows:
\[
    \pi_1=\mathcal{N}(\mathbf{0}_2, I_2),
    \qquad
    \pi_2=\mathcal{N}\left(\mathbf{1}_2, \begin{pmatrix} 1 & \rho \\ \rho & 1 \end{pmatrix}\right),
    \qquad
    \rho\in\left\{-0.6,-0.3,0,0.3,0.6\right\}.
\]
It is easy to verify that the product of these two components is again a Gaussian distribution:
\[
    \pipro
    =
    \mathcal{N}\left(\frac{1}{2+\rho}\mathbf{1}_2, \frac{1}{4-\rho^2}\begin{pmatrix} 2-\rho^2 & \rho \\ \rho & 2-\rho^2 \end{pmatrix}\right).
\]

The results are shown in Figures~\ref{fig:bivariate_correlation_convergence} and~\ref{fig:bivariate_correlation_convergence_square} for $h(u)=u$ and $h(u)=u\odot u$, respectively.

\convergencefigure[u\odot u]
    {gaussian/bivariate/k_2/correlation}
    {bivariate_correlation}
    {bivariate_correlation_convergence}
    {MSE and runtime versus $S$ for $h(u)=\convergencefunction$ with $K=2$ bivariate Gaussian components. Columns vary $\rho\in\left\{-0.6,-0.3,0,0.3,0.6\right\}$. The first component is fixed with $\rho=0$.}

Across the settings, PMC-OLS closely matches the Oracle in terms of MSE for $h(u)=u$, as expected given its correct specification, but fails for $h(u)=u\odot u$, when the linear model is misspecified. PMC-MLP and PMC-GP initially lag behind GBF but overtake it in most settings as the sample size increases. In terms of training time, all PMC methods are considerably faster than GBF, with PMC-OLS being the fastest.

\FloatBarrier

\subsubsection{Multivariate Gaussian: Dimension}\label{app:multivariate_gaussian_dimension}

We consider the product of $K=2$ multivariate Gaussian components with fixed Euclidean distance between their means, varying the dimension. The components are defined as follows:
\[
    \pi_1=\mathcal{N}(\mathbf{0}_p, I_p),
    \qquad
    \pi_2=\mathcal{N}\left(\sqrt{\frac{8}{p}}\mathbf{1}_p, I_p\right),
    \qquad
    p\in\left\{2,3,5,10,20\right\}.
\]
It is easy to verify that the product of these two components is again a Gaussian distribution:
\[
    \pipro
    =
    \mathcal{N}\left(\sqrt{\frac{2}{p}}\mathbf{1}_p, \frac{1}{2}I_p\right).
\]

The results are shown in Figures~\ref{fig:multivariate_dimension_scaling_identity} and~\ref{fig:multivariate_dimension_scaling_square} for $h(u)=u$ and $h(u)=u\odot u$, respectively.

\singlefigure
    {gaussian/multivariate/k_2/dimension_scaling/dimension_scaling_h_identity.pdf}
    {multivariate_dimension_scaling_identity}
    {MSE and runtime versus $S$ for $h(u)=u$ with $K=2$ multivariate Gaussian components. Columns vary $p\in\left\{2,3,5,10,20\right\}$, and rows show MSE and runtime, respectively.}

\singlefigure
    {gaussian/multivariate/k_2/dimension_scaling/dimension_scaling_h_square.pdf}
    {multivariate_dimension_scaling_square}
    {MSE and runtime versus $S$ for $h(u)=u\odot u$ with $K=2$ multivariate Gaussian components. Columns vary $p\in\left\{2,3,5,10,20\right\}$, and rows show MSE and runtime, respectively.}

Across the settings, PMC-OLS closely matches the Oracle in terms of MSE for $h(u)=u$, as expected given its correct specification. In low dimension ($p\leq 5$), PMC-MLP and PMC-GP lag behind GBF, whereas at $p=20$ they outperform it, and at $p=10$ PMC-MLP matches it at large sample sizes. In terms of training time, all PMC methods are considerably faster than GBF, with PMC-OLS being the fastest. Overall, it looks like increasing the dimension makes the estimation problem more difficult, and the runtime of GBF seems to scale poorly with dimension, while the runtime of PMC-MLP is unaffected by the dimension and that of PMC-GP and PMC-OLS grows moderately, remaining far below GBF.

\FloatBarrier

\subsubsection{Gaussian Mixture: Symmetry}\label{app:gaussian_mixture_symmetry}

We consider the product of $K=2$ identical Gaussian-mixture components, each with two well-separated modes of equal weight and scale. The components are defined as follows:
\[
    \pi_1=\pi_2=\frac{1}{2}\mathcal{N}(-2,1)+\frac{1}{2}\mathcal{N}(2,1).
\]
It is easy to verify that the product of these two components is again a Gaussian mixture, now with three modes:
\[
    \pipro
    =
    \frac{1}{2\left(1+e^{-4}\right)}\mathcal{N}\left(-2,\frac{1}{2}\right)
    +\frac{e^{-4}}{1+e^{-4}}\mathcal{N}\left(0,\frac{1}{2}\right)
    +\frac{1}{2\left(1+e^{-4}\right)}\mathcal{N}\left(2,\frac{1}{2}\right).
\]

The results are shown in Figures~\ref{fig:gaussian_mixture_k2_symmetric_bimodal_convergence} and~\ref{fig:gaussian_mixture_k2_symmetric_bimodal_convergence_square} for $h(u)=u$ and $h(u)=u^2$, respectively.

\convergencefigure
    {gaussian/mixture/k_2/symmetric_bimodal}
    {gaussian_mixture_k2_symmetric_bimodal}
    {gaussian_mixture_k2_symmetric_bimodal_convergence}
    {MSE and runtime versus $S$ for $h(u)=\convergencefunction$ with $K=2$ symmetric bimodal Gaussian-mixture components.}

In this setting, PMC-OLS attains a lower MSE than the Oracle for $h(u)=u$, as expected given its correct specification, but fails for $h(u)=u^2$, for which the linear model is misspecified. PMC-MLP and PMC-GP outperform GBF for $h(u)=u$ at almost all sample sizes. For $h(u)=u^2$, PMC-MLP and PMC-GP seem to still underperform GBF. In terms of training time, all PMC methods are considerably faster than GBF, with PMC-OLS being the fastest, even faster than the Oracle. In particular, it looks like the runtime of GBF scales very poorly with the sample size.

\FloatBarrier

\FloatBarrier

\subsubsection{Gaussian Mixture: Mixture Weights}\label{app:gaussian_mixture_mixture_weights}

We consider the product of $K=2$ Gaussian-mixture components sharing the same two modes, but with opposing weights on them. The components are defined as follows:
\[
    \pi_1=0.85\,\mathcal{N}\left(-2,0.8^2\right)+0.15\,\mathcal{N}\left(2,0.8^2\right),
    \qquad
    \pi_2=0.15\,\mathcal{N}\left(-2,0.8^2\right)+0.85\,\mathcal{N}\left(2,0.8^2\right).
\]
It is easy to verify that the product of these two components is again a Gaussian mixture, now with three modes:
\[
    \pipro
    =
    \frac{51\,\mathcal{N}\left(-2,\frac{8}{25}\right)+298e^{-25/4}\mathcal{N}\left(0,\frac{8}{25}\right)+51\,\mathcal{N}\left(2,\frac{8}{25}\right)}{102+298e^{-25/4}}.
\]

The results are shown in Figures~\ref{fig:gaussian_mixture_k2_opposing_mode_weights_convergence} and~\ref{fig:gaussian_mixture_k2_opposing_mode_weights_convergence_square} for $h(u)=u$ and $h(u)=u^2$, respectively.

\convergencefigure
    {gaussian/mixture/k_2/opposing_mode_weights}
    {gaussian_mixture_k2_opposing_mode_weights}
    {gaussian_mixture_k2_opposing_mode_weights_convergence}
    {MSE and runtime versus $S$ for $h(u)=\convergencefunction$ with $K=2$ Gaussian-mixture components with opposing mode weights.}

In this setting, PMC-OLS attains a lower MSE than the Oracle for $h(u)=u$ at most sample sizes, as expected given its correct specification, but fails for $h(u)=u^2$, for which the linear model is misspecified. PMC-MLP and PMC-GP systematically outperform GBF for $h(u)=u$, and they seem to still lag behind GBF for $h(u)=u^2$. In terms of training time, all PMC methods are considerably faster than GBF, with PMC-OLS being the fastest, even faster than the Oracle. In particular, it looks like the runtime of GBF scales very poorly with the sample size.

\FloatBarrier

\subsubsection{Gaussian Mixture: Mixture Scales}\label{app:gaussian_mixture_mixture_scales}

We consider the product of $K=2$ unimodal Gaussian-mixture components with different locations and heterogeneous scales. The components are defined as follows:
\[
    \pi_1=\mathcal{N}\left(-1,0.4^2\right),
    \qquad
    \pi_2=\mathcal{N}\left(1,1.2^2\right).
\]
It is easy to verify that the product of these two components is again a Gaussian distribution:
\[
    \pipro
    =
    \mathcal{N}\left(-\frac{4}{5},\frac{18}{125}\right).
\]

The results are shown in Figures~\ref{fig:gaussian_mixture_k2_heterogeneous_scale_unimodal_convergence} and~\ref{fig:gaussian_mixture_k2_heterogeneous_scale_unimodal_convergence_square} for $h(u)=u$ and $h(u)=u^2$, respectively.

\convergencefigure
    {gaussian/mixture/k_2/heterogeneous_scale_unimodal}
    {gaussian_mixture_k2_heterogeneous_scale_unimodal}
    {gaussian_mixture_k2_heterogeneous_scale_unimodal_convergence}
    {MSE and runtime versus $S$ for $h(u)=\convergencefunction$ with $K=2$ unimodal Gaussian-mixture components with heterogeneous scales.}

In this setting, PMC-OLS closely matches the Oracle in terms of MSE for $h(u)=u$, as expected given its correct specification, but fails for $h(u)=u^2$, for which the linear model is misspecified. PMC-MLP and PMC-GP lag behind GBF at small sample sizes and are comparable to it at intermediate ones. In terms of training time, all PMC methods are considerably faster than GBF, with PMC-OLS being the fastest, even faster than the Oracle. In particular, it looks like the runtime of GBF scales very poorly with the sample size.

\FloatBarrier

\subsubsection{Gaussian Mixture: Number of Components}\label{app:gaussian_mixture_number_of_components}

We consider the product of $K=3$ Gaussian-mixture components with different numbers of modes, weights, and scales. The components are defined as follows:
\begin{align*}
    \pi_1&=0.6\,\mathcal{N}\left(-2,1^2\right)+0.4\,\mathcal{N}\left(1,0.7^2\right),\\
    \pi_2&=0.2\,\mathcal{N}\left(-1.5,0.8^2\right)+0.5\,\mathcal{N}\left(0,1.1^2\right)+0.3\,\mathcal{N}\left(2,0.6^2\right),\\
    \pi_3&=0.4\,\mathcal{N}\left(-0.5,1.2^2\right)+0.6\,\mathcal{N}\left(1.5,0.8^2\right).
\end{align*}
It is easy to verify that the product of these components is again a Gaussian mixture, now with $12$ terms. If we write each component as follows:
\[
    \pi_k=\sum_{i=1}^{m_k} \omega_{ki} \, \mathcal{N}\left(\mu_{ki}, \sigma_{ki}^2\right),
    \qquad
    k=1,2,3,
\]
where $m_k$ is the number of Gaussian terms in the mixture $\pi_k$, then the product of these components can be expressed as follows:
\[
    \pipro
    =
    \sum_{i_1=1}^{m_1}\sum_{i_2=1}^{m_2}\sum_{i_3=1}^{m_3}\omega_{i_1i_2i_3}\mathcal{N}\left(\mu_{i_1i_2i_3}, \sigma_{i_1i_2i_3}^2\right),
\]
where the means $\mu_{i_1i_2i_3}$, the variances $\sigma_{i_1i_2i_3}^2$, and the weights $\omega_{i_1i_2i_3}$ satisfy the following relationships:
\begin{align*}
    \mu_{i_1i_2i_3} &
    = \sigma_{i_1i_2i_3}^2\sum_{k=1}^{3}\frac{\mu_{ki_k}}{\sigma_{ki_k}^2}, \\
    \sigma_{i_1i_2i_3}^2 &
    =\left(\sum_{k=1}^{3}\frac{1}{\sigma_{ki_k}^2}\right)^{-1}, \\
    \omega_{i_1i_2i_3} &
    \propto \left(\prod_{k=1}^{3}\frac{\omega_{ki_k}}{\sigma_{ki_k}}\right)\sigma_{i_1i_2i_3}\exp\left(-\frac{1}{2}\left(\sum_{k=1}^{3}\frac{\mu_{ki_k}^2}{\sigma_{ki_k}^2}-\frac{\mu_{i_1i_2i_3}^2}{\sigma_{i_1i_2i_3}^2}\right)\right).
\end{align*}

The results are shown in Figures~\ref{fig:gaussian_mixture_k3_mixed_components_convergence} and~\ref{fig:gaussian_mixture_k3_mixed_components_convergence_square} for $h(u)=u$ and $h(u)=u^2$, respectively.

\convergencefigure
    {gaussian/mixture/k_3/mixed_components}
    {gaussian_mixture_k3_mixed_components}
    {gaussian_mixture_k3_mixed_components_convergence}
    {MSE and runtime versus $S$ for $h(u)=\convergencefunction$ with $K=3$ Gaussian-mixture components.}

In this setting, PMC-OLS fails to estimate both the first and second moments, as expected given its misspecification. PMC-MLP underperforms GBF for $S\leq 5\times10^4$ and outperforms it for $h(u)=u$ at $S=10^5$, while PMC-GP underperforms GBF at all sample sizes. In terms of training time, all PMC methods are considerably faster than GBF, with PMC-OLS being the fastest. In particular, it looks like the runtime of GBF scales very poorly with the sample size.

\FloatBarrier

\subsubsection{Gaussian Mixture: Heterogeneity}\label{app:gaussian_mixture_heterogeneity}

We consider the product of $K=2$ Gaussian-mixture components that are heterogeneous in their numbers of modes, weights, locations, and scales. The components are defined as follows:
\[
    \pi_1=0.75\,\mathcal{N}\left(-1.5,0.7^2\right)+0.25\,\mathcal{N}\left(2,1.2^2\right),
    \qquad
    \pi_2=0.2\,\mathcal{N}\left(-2,1.1^2\right)+0.5\,\mathcal{N}\left(0.5,0.6^2\right)+0.3\,\mathcal{N}\left(2.5,0.9^2\right).
\]
It is easy to verify that the product of these two components is again a Gaussian mixture, now with $6$ terms. If we write each component as follows:
\[
    \pi_k=\sum_{i=1}^{m_k} \omega_{ki} \, \mathcal{N}\left(\mu_{ki}, \sigma_{ki}^2\right),
    \qquad
    k=1,2,
\]
where $m_k$ is the number of Gaussian terms in the mixture $\pi_k$, then the product of these components can be expressed as follows:
\[
    \pipro
    =
    \sum_{i_1=1}^{m_1}\sum_{i_2=1}^{m_2}\omega_{i_1i_2}\mathcal{N}\left(\mu_{i_1i_2}, \sigma_{i_1i_2}^2\right),
\]
where the means $\mu_{i_1i_2}$, the variances $\sigma_{i_1i_2}^2$, and the weights $\omega_{i_1i_2}$ satisfy the following relationships:
\begin{align*}
    \mu_{i_1i_2} &
    = \sigma_{i_1i_2}^2\sum_{k=1}^{2}\frac{\mu_{ki_k}}{\sigma_{ki_k}^2}, \\
    \sigma_{i_1i_2}^2 &
    =\left(\sum_{k=1}^{2}\frac{1}{\sigma_{ki_k}^2}\right)^{-1}, \\
    \omega_{i_1i_2} &
    \propto \left(\prod_{k=1}^{2}\frac{\omega_{ki_k}}{\sigma_{ki_k}}\right)\sigma_{i_1i_2}\exp\left(-\frac{1}{2}\left(\sum_{k=1}^{2}\frac{\mu_{ki_k}^2}{\sigma_{ki_k}^2}-\frac{\mu_{i_1i_2}^2}{\sigma_{i_1i_2}^2}\right)\right).
\end{align*}

The results are shown in Figures~\ref{fig:gaussian_mixture_k2_heterogeneous_unequal_convergence} and~\ref{fig:gaussian_mixture_k2_heterogeneous_unequal_convergence_square} for $h(u)=u$ and $h(u)=u^2$, respectively.

\convergencefigure
    {gaussian/mixture/k_2/heterogeneous_unequal}
    {gaussian_mixture_k2_heterogeneous_unequal}
    {gaussian_mixture_k2_heterogeneous_unequal_convergence}
    {MSE and runtime versus $S$ for $h(u)=\convergencefunction$ with $K=2$ heterogeneous Gaussian-mixture components.}

In this setting, PMC-OLS fails to estimate both the first and second moments, as expected given its misspecification. For $h(u)=u$, both PMC-MLP and PMC-GP outperform GBF at almost all sample sizes, while for $h(u)=u^2$ PMC-MLP is comparable to GBF and PMC-GP underperforms it. In terms of training time, all PMC methods are considerably faster than GBF, with PMC-OLS being the fastest. In particular, it looks like the runtime of GBF scales very poorly with the sample size.

\FloatBarrier

\subsubsection{Beta: Concentration}\label{app:beta_concentration}

We consider the product of $K=2$ symmetric Beta components, $\mathrm{Beta}\left(a_1,a_1\right)$ and $\mathrm{Beta}\left(a_2,a_2\right)$, either with homogeneous or with heterogeneous concentrations. In the homogeneous case, the components are identical and are defined as follows:
\[
    \mathrm{Beta}\left(a,a\right),
    \qquad
    a\in\left\{1,2,5\right\}.
\]
In the heterogeneous case, the components have different concentrations and are defined as follows:
\[
    \mathrm{Beta}\left(2,2\right),
    \qquad
    \mathrm{Beta}\left(5,5\right).
\]
It is easy to verify that, in both cases, the product of these components is again a Beta distribution:
\[
    \pipro
    =
    \mathrm{Beta}\left(a_1+a_2-1,a_1+a_2-1\right).
\]

The results are shown in Figures~\ref{fig:beta_concentration_convergence} and~\ref{fig:beta_concentration_convergence_square} for the homogeneous case and in Figures~\ref{fig:beta_heterogeneous_concentration_convergence} and~\ref{fig:beta_heterogeneous_concentration_convergence_square} for the heterogeneous case, for $h(u)=u$ and $h(u)=u^2$, respectively.

In both cases, PMC-OLS closely matches the Oracle in terms of MSE for $h(u)=u$, as expected given its correct specification, but fails for $h(u)=u^2$, for which the linear model is misspecified. In both cases, PMC-MLP and PMC-GP outperform GBF for large sample sizes, at least for $h(u)=u$, while for $h(u)=u^2$ they still underperform GBF in the uniform case and overtake it at $S\geq 5\times10^5$ for the other concentrations. In terms of training time, all PMC methods are considerably faster than GBF, with PMC-OLS being the fastest.

\convergencefigure
    {beta/concentration}
    {beta_concentration}
    {beta_concentration_convergence}
    {MSE and runtime versus $S$ for $h(u)=\convergencefunction$ with $K=2$ homogeneous Beta components. Columns vary $a\in\left\{1,2,5\right\}$.}

\convergencefigure
    {beta/heterogeneous_concentration}
    {beta_heterogeneous_concentration}
    {beta_heterogeneous_concentration_convergence}
    {MSE and runtime versus $S$ for $h(u)=\convergencefunction$ with $K=2$ heterogeneous Beta components with $\pi_1=\mathrm{Beta}(2,2)$ and $\pi_2=\mathrm{Beta}(5,5)$.}

\FloatBarrier

\subsubsection{Beta: Skewness}\label{app:beta_skewness}

We consider the product of $K=2$ skewed Beta components, with skewness either in the same direction or in opposite directions. In the matching case, the components are identical and are defined as follows:
\[
    \pi_1=\pi_2=\mathrm{Beta}\left(2,5\right).
\]
It is easy to verify that the product of these two components is again a Beta distribution:
\[
    \pipro=\mathrm{Beta}\left(3,9\right).
\]
In the opposing case, the components are mirror images of each other and are defined as follows:
\[
    \pi_1=\mathrm{Beta}\left(2,5\right),
    \qquad
    \pi_2=\mathrm{Beta}\left(5,2\right).
\]
Similarly, the product of these two components is again a Beta distribution:
\[
    \pipro=\mathrm{Beta}\left(6,6\right).
\]

The results are shown in Figures~\ref{fig:beta_matching_skew_convergence} and~\ref{fig:beta_matching_skew_convergence_square} for the matching case and in Figures~\ref{fig:beta_opposing_skew_convergence} and~\ref{fig:beta_opposing_skew_convergence_square} for the opposing case, for $h(u)=u$ and $h(u)=u^2$, respectively.

In the matching case, PMC-OLS fails to estimate both the first and second moments, as expected given its misspecification. PMC-MLP and PMC-GP underperform GBF, with the gap shrinking as the sample size increases. In the opposing case, PMC-OLS closely matches the Oracle in terms of MSE for $h(u)=u$, as expected given its correct specification, but fails for $h(u)=u^2$, for which the linear model is misspecified. In the opposing case, PMC-MLP and PMC-GP are comparable to GBF at $S=10^5$ and tend to outperform it at larger sample sizes. In terms of training time, all PMC methods are considerably faster than GBF, with PMC-OLS being the fastest.

\convergencefigure
    {beta/matching_skew}
    {beta_matching_skew}
    {beta_matching_skew_convergence}
    {MSE and runtime versus $S$ for $h(u)=\convergencefunction$ with $K=2$ Beta components with $(a_1,b_1)=(2,5)$ and $(a_2,b_2)=(2,5)$.}

\convergencefigure
    {beta/opposing_skew}
    {beta_opposing_skew}
    {beta_opposing_skew_convergence}
    {MSE and runtime versus $S$ for $h(u)=\convergencefunction$ with $K=2$ Beta components with $(a_1,b_1)=(2,5)$ and $(a_2,b_2)=(5,2)$.}

\FloatBarrier

\subsubsection{Beta: Boundary Singularities}\label{app:beta_boundary_singularities}

We consider the product of $K=2$ Beta components whose densities diverge at opposite endpoints of the unit interval. The components are defined as follows:
\[
    \pi_1=\mathrm{Beta}\left(0.8,3\right),
    \qquad
    \pi_2=\mathrm{Beta}\left(3,0.8\right).
\]
It is easy to verify that the product of these two components is again a Beta distribution, which is bounded even though its components are not:
\[
    \pipro
    =
    \mathrm{Beta}\left(2.8,2.8\right).
\]

The results are shown in Figures~\ref{fig:beta_endpoint_singular_stress_convergence} and~\ref{fig:beta_endpoint_singular_stress_convergence_square} for $h(u)=u$ and $h(u)=u^2$, respectively.

\convergencefigure
    {beta/endpoint_singular_stress}
    {beta_endpoint_singular_stress}
    {beta_endpoint_singular_stress_convergence}
    {MSE and runtime versus $S$ for $h(u)=\convergencefunction$ with $K=2$ Beta components with $(a_1,b_1)=(0.8,3)$ and $(a_2,b_2)=(3,0.8)$.}

In this setting, PMC-OLS closely matches the Oracle in terms of MSE for $h(u)=u$, as expected given its correct specification, but fails for $h(u)=u^2$, for which the linear model is misspecified. PMC-MLP and PMC-GP outperform GBF only at the largest sample size ($S=10^6$). In terms of training time, all PMC methods are considerably faster than GBF, with PMC-OLS being the fastest.

\FloatBarrier

\endgroup

\subsection{Additional Real-Data Results}\label{app:additional_real_data_results}

\paragraph{MAGIC.} We complement the large-scale computation application with an experiment on the MAGIC dataset~\citep{bock2004methods}, a scientific dataset with 19,020 samples and 10 features. We fit a logistic regression model to predict the binary label from the features, for a total of 11 coefficients including the intercept, and split the data into $K=10$ batches. Because the sample size is much smaller than in the HIGGS experiment, GBF is feasible here.

Table~\ref{tab:magic-identity-comparison} shows that, among the merging methods, PMC-GP achieves the lowest estimation error, followed closely by PMC-MLP, then GBF and PMC-OLS. The merging cost of GBF is large, about 27 minutes compared with less than one minute of batch MCMC, so its total runtime is about 28 times that of MCMC on the full dataset. In contrast, the merging cost of the PMC methods is negligible, and their total runtimes are comparable to that of MCMC on the full dataset.

\begin{table}[ht]
\centering
\small
\setlength{\tabcolsep}{2.5pt}
\caption{Results on the MAGIC data.}
\label{tab:magic-identity-comparison}
\input{tables/magic_identity_comparison}
\end{table}

\FloatBarrier

\paragraph{Evaluation metric.} For the real-data tables, let $\hat{\mu}_h,\mu_h\in\mathbb{R}^r$ denote the method estimate and the full-data or Oracle MCMC reference, respectively. We report
\[
    \operatorname{Estimation\ error}\left(\hat{\mu}_h\right)
    =
    \left[
        \frac{1}{r}\sum_{j=1}^r
        \left(\hat{\mu}_{h,j}-\mu_{h,j}\right)^2
    \right]^{1/2}.
\]
We use $r=5$, $r=29$, and $r=11$ for respiratory, HIGGS, and MAGIC, respectively, and $r=1$ for HSP90, where the metric reduces to absolute error.

\begingroup

\renewcommand{\convergencefiguretarget}[6]{
    \begin{figure}[htbp]
        \def\convergencefunction{#3}
        \centering
        \includegraphics{#1/h_#2.pdf}
        \caption{#6}
        \label{fig:#5}
    \end{figure}
}

\subsection{Fusion Baseline Comparison}\label{app:fusion_baseline_comparison}

This section compares the fusion methods in terms of estimation accuracy and computational cost, and shows that GBF is at least as accurate and as fast as MCF and BF, with the advantage growing with the dimension. We begin with a complete overview of the configurations considered:
\begin{itemize}
    \item \textbf{Gaussian}:
    \begin{itemize}
        \item \textbf{Univariate}:
        \begin{itemize}
            \renewcommand{\labelitemiii}{$\circ$}
            \item Location and scale.
        \end{itemize}
        \item \textbf{Bivariate}:
        \begin{itemize}
            \renewcommand{\labelitemiii}{$\circ$}
            \item Mean separation.
        \end{itemize}
        \item \textbf{Multivariate}:
        \begin{itemize}
            \renewcommand{\labelitemiii}{$\circ$}
            \item Dimension.
        \end{itemize}
    \end{itemize}
\end{itemize}

\FloatBarrier

\subsubsection{Univariate Gaussian: Location and Scale}\label{app:fusion_univariate_gaussian_location_scale}

We consider the product of $K=2$ univariate Gaussian components with fixed first component and varying second component in terms of location and scale. The components are defined as follows:
\[
    \pi_1=\mathcal{N}(0, 1),
    \qquad
    \pi_2=\mathcal{N}(\mu_2,\sigma_2^2),
    \qquad
    \mu_2,\sigma_2\in\left\{1,2,3\right\}.
\]
It is easy to verify that the product of these two components is again a Gaussian distribution:
\[
    \pipro
    =
    \mathcal{N}\left(\frac{\mu_2}{1+\sigma_2^2}, \frac{\sigma_2^2}{1+\sigma_2^2}\right).
\]

The results are shown in Figures~\ref{fig:fusion_methods_benchmark_partial} and~\ref{fig:fusion_methods_benchmark_partial_square} for $h(u)=u$ and $h(u)=u^2$, respectively.

\begingroup

\renewcommand{\convergencefiguretarget}[6]{
    \begin{figure}[H]
        \def\convergencefunction{#3}
        \centering
        \includegraphics[width=\linewidth]{#1/h_#2.pdf}
        \caption{#6}
        \label{fig:#5}
    \end{figure}
}

\convergencefigure
    {fusion_methods_benchmark/univariate}
    {fusion_methods_benchmark}
    {fusion_methods_benchmark_partial}
    {MSE and runtime versus $S$ for $h(u)=\convergencefunction$ with $K=2$ univariate Gaussian components. Columns and rows vary $\mu_2,\sigma_2\in\left\{1,2,3\right\}$, respectively. The first component is fixed with $\mu_1=0$ and $\sigma_1=1$.}

\endgroup

Across these settings, MCF and GBF are comparable in terms of MSE, while BF is decisively worse. In terms of runtime, GBF and MCF are comparable, while BF is the slowest. Moreover, notice that the runtime of MCF and GBF is almost constant, while the runtime of BF increases with the sample size.

\newpage

\subsubsection{Bivariate Gaussian: Mean Separation}\label{app:fusion_bivariate_gaussian_mean_separation}

We consider the product of $K=2$ bivariate Gaussian components with fixed first component and varying second component in terms of location along the diagonal. The components are defined as follows:
\[
    \pi_1=\mathcal{N}(\mathbf{0}_2, I_2),
    \qquad
    \pi_2=\mathcal{N}(\mu\mathbf{1}_2, I_2),
    \qquad
    \mu\in\left\{1,2,3\right\}.
\]

It is easy to verify that the product of these two components is again a Gaussian distribution:
\[
    \pipro
    =
    \mathcal{N}\left(\frac{\mu}{2}\mathbf{1}_2, \frac{1}{2}I_2\right).
\]

The results are shown in Figures~\ref{fig:bivariate_mean_shift_fusion_benchmark} and~\ref{fig:bivariate_mean_shift_fusion_benchmark_square} for $h(u)=u$ and $h(u)=u\odot u$, respectively.

\convergencefigure[u\odot u]
    {fusion_methods_benchmark/bivariate}
    {bivariate_mean_shift_fusion_benchmark}
    {bivariate_mean_shift_fusion_benchmark}
    {MSE and runtime versus $S$ for $h(u)=\convergencefunction$ with $K=2$ bivariate Gaussian components. Columns vary $\mu\in\left\{1,2,3\right\}$. The first component is fixed with $\mu=0$.}

In these settings, MCF and GBF are comparable in terms of MSE, while BF is decisively worse. In terms of runtime, GBF is the fastest method, followed by MCF and then BF. Moreover, notice that the runtime of MCF and GBF is almost constant, while the runtime of BF increases with the sample size. Furthermore, notice that the runtime of MCF increases with the mean separation, while the runtime of GBF grows much more slowly.

\FloatBarrier

\subsubsection{Multivariate Gaussian: Dimension}\label{app:fusion_multivariate_gaussian_dimension}

We consider the product of $K=2$ multivariate Gaussian components with fixed Euclidean distance between their means, varying the dimension. The components are defined as follows:
\[
    \pi_1=\mathcal{N}(\mathbf{0}_p, I_p),
    \qquad
    \pi_2=\mathcal{N}\left(\sqrt{\frac{8}{p}}\mathbf{1}_p, I_p\right),
    \qquad
    p\in\left\{2,3,4\right\}.
\]
It is easy to verify that the product of these two components is again a Gaussian distribution:
\[
    \pipro
    =
    \mathcal{N}\left(\sqrt{\frac{2}{p}}\mathbf{1}_p, \frac{1}{2}I_p\right).
\]

The results are shown in Figures~\ref{fig:multivariate_dimension_scaling_fusion_benchmark} and~\ref{fig:multivariate_dimension_scaling_fusion_benchmark_square} for $h(u)=u$ and $h(u)=u\odot u$, respectively.

\convergencefigure[u\odot u]
    {fusion_methods_benchmark/multivariate}
    {multivariate_dimension_scaling_fusion_benchmark}
    {multivariate_dimension_scaling_fusion_benchmark}
    {MSE and runtime versus $S$ for $h(u)=\convergencefunction$ with $K=2$ multivariate Gaussian components. Columns vary $p\in\left\{2,3,4\right\}$.}

In these settings, GBF is more accurate than MCF in terms of MSE, while BF is decisively worse. In terms of runtime, GBF is the fastest method, followed by MCF and then BF. Moreover, notice that the runtime of MCF and GBF is almost constant, while the runtime of BF increases with the sample size. Furthermore, notice that the runtime of MCF increases with the dimension, while the runtime of GBF grows much more slowly.

\FloatBarrier

\endgroup

\subsection{Computational Setup}\label{app:computational_setup}

We describe the hardware and software used in our experiments below.

\paragraph{Hardware.} On the shared machine, PMC-MLP and the simulation-study PMC-GP fits were trained on an NVIDIA RTX A4000 GPU, whereas the real-data PMC-GP fits, PMC-OLS, the simulation Oracles, and the fusion methods ran on the CPU. The real-data NUTS fits ran on the GPU. Runtime is the wall-clock time to fit the estimator for PMC and the Oracle, or to generate draws from the product distribution for the fusion methods. Except for HSP90 and the fusion baseline comparison in Appendix~\ref{app:fusion_baseline_comparison}, all timings were measured on the shared machine. Both the HSP90 experiment, including the neural posterior, PMC-OLS, PMC-MLP, PMC-GP, and the reference-posterior calculation, and the fusion baseline comparison were run on the CPU of a MacBook Pro with an Apple M3 Max, 14 cores, and 36 GB of memory; the HSP90 image simulator is only available on that machine.

\begin{table}[ht]
    \centering
    \caption{Hardware of the shared machine used for most experiments.}
    \label{tab:computational_hardware}
    \begin{tabular}{@{}ll@{}}
        \toprule
        Component & Specification \\
        \midrule
        CPU & 2$\times$ AMD EPYC 7763 (128 physical cores, 256 threads, 255 online) \\
        GPU & 4$\times$ NVIDIA RTX A4000 (16 GB memory each) \\
        RAM & $\sim$881 GiB usable, shared across all users \\
        OS & Ubuntu 20.04 LTS \\
        \bottomrule
    \end{tabular}
\end{table}

\paragraph{Software.} PMC and the Oracle are implemented in Python 3.12 with GPyTorch 1.15 and PyTorch 2.13 built for CUDA 13.0, and the fusion methods are run in R 4.4.2 using the R package DCFusion 1.0~\citep{chan2021dcfusion}.

\paragraph{Licenses.} The HIGGS and MAGIC datasets are released under CC BY 4.0, the respiratory data are distributed with the R package geepack under GPL ($\geq$ 3), DCFusion is released under GPL ($\geq$ 2), and the HSP90 simulator is built on cryoJAX, which is released under LGPL-2.1.

\paragraph{Code availability.} Code and outputs for the experiments are available at \url{https://github.com/PietroCarlotti/product-monte-carlo}. The HSP90 image simulator is third-party software that we do not have permission to redistribute; consequently, the complete HSP90 pipeline cannot be reproduced from this repository alone. For HSP90, we provide the saved posterior draws, reported outputs, and downstream analysis code. Regenerating the images or rerunning simulator-dependent stages requires access to the external simulator.

\section{Hyperparameters}\label{app:hyperparameters}

This section lists the hyperparameters used in all experiments. Section~\ref{app:hp_simulation_design} gives the settings shared by the baseline and supplementary benchmark suite in Appendix~\ref{app:additional_simulation_results}, Section~\ref{app:hp_pmc_common} gives the corresponding common PMC settings, and Section~\ref{app:hp_deviations} lists the experiments that deviate from them. Section~\ref{app:hp_main_text_additions} documents the main-text and diagnostic studies. Sections~\ref{app:hp_fusion_oracle}, \ref{app:hp_real_data}, and~\ref{app:hp_sbi} cover the fusion methods and Oracle, the real-data experiments, and the simulation-based inference experiment, and Section~\ref{app:hp_selection} explains their selection.

\subsection{Shared Baseline Simulation Settings}\label{app:hp_simulation_design}

The distributions of each experiment are described in Appendix~\ref{app:additional_simulation_results}. Table~\ref{tab:hp_shared} lists the settings shared by this baseline and supplementary benchmark suite.

\begin{table}[!htb]
\centering
\small
\setlength{\tabcolsep}{2.5pt}
\caption{Settings shared by the baseline and supplementary benchmark suite in Appendix~\ref{app:additional_simulation_results}.}
\label{tab:hp_shared}
\begin{tabular}{p{0.3\linewidth} p{0.64\linewidth}}
\toprule
Setting & Value \\
\midrule
Methods & PMC-OLS, PMC-MLP, PMC-GP, GBF, and the Oracle \\
Targets & $h(u)=u$ and $h(u)=u\odot u$ \\
Sample sizes $S$ & $100$, $500$, $1000$, $5000$, $10^4$, $5\times10^4$, $10^5$, $5\times10^5$, $10^6$, and, only for the univariate Gaussian location and scale experiment, $2\times10^6$ \\
Realized sample size & $\operatorname{Poisson}(S)$ draws per component \\
Replications & $20$ independent replications per experimental configuration \\
Number of components $K$ & $2$, except in the experiments on the number of components \\
Seeds & for base seed $b$, replicate $r$, and zero-based configuration index $j$, the seed is $b+rJ+j$, where $J$ is the number of configurations. Most experiments use $b=20260825$; homogeneous and heterogeneous $K$-scaling use $b=26082600+K$ and $b=26092600+K$, respectively, and dimension scaling uses $b=26090400+p$ \\
Consensus and discrepancy & as in Table~\ref{tab:hp_pmc_common} \\
GBF particles & see Section~\ref{app:hp_fusion_oracle} \\
\bottomrule
\end{tabular}
\end{table}

\subsection{Common PMC Settings}\label{app:hp_pmc_common}

For PMC-MLP, ensemble size $2$ averages predictions from two independently initialized networks.

Exact posterior inference for the GP formulation in Section~\ref{sec:UQ} requires factorizing an $S\times S$ covariance matrix, with $O(S^3)$ time and $O(S^2)$ memory, and integration over the GP and likelihood hyperparameters adds further cost. This is computationally infeasible at the sample sizes considered here, so PMC-GP uses a sparse variational approximation. Let $Z=(z_1,\ldots,z_M)$ denote $M$ inducing inputs and let
\[
    \boldsymbol{f}_Z=
    \left(f_h(z_1),\ldots,f_h(z_M)\right)^\top.
\]
We approximate the posterior of the inducing values by
\[
    q(\boldsymbol{f}_Z)
    =\mathcal{N}(\boldsymbol{a},\boldsymbol{B}),
\]
which induces the approximate process posterior
\[
    q(f_h)
    =\int p(f_h\mid\boldsymbol{f}_Z)
    q(\boldsymbol{f}_Z)\,\mathrm{d}\boldsymbol{f}_Z.
\]
Writing $y_s=h(C^{(s)})$, the variational parameters $\boldsymbol{a}$ and $\boldsymbol{B}$ maximize the evidence lower bound
\[
    \mathcal{L}
    =\sum_{s=1}^{S}
    \E_{q(f_h(D^{(s)}))}
    \left[\log p\left(y_s\mid f_h(D^{(s)})\right)\right]
    -\operatorname{KL}\left(
        q(\boldsymbol{f}_Z)\,\middle\|\,p(\boldsymbol{f}_Z)
    \right).
\]
The PMC estimate and its uncertainty are obtained from the marginal $q(f_h(\mathbf{0}_m))$.

\paragraph{Approximation settings.}
We use $M=64$ inducing inputs, initialized at empirical quantiles of the standardized discrepancies and then fixed. The common suite uses a Mat\'ern-$3/2$ kernel. The GP mean, kernel lengthscale and output scale, and likelihood noise variance correspond to $\beta_0$, $\ell$, $\tau^2$, and $\sigma^2$ in Section~\ref{sec:UQ}.

\paragraph{Optimization.}
Natural gradients optimize $\boldsymbol{a}$ and $\boldsymbol{B}$. The GP and likelihood hyperparameters are fixed, so the listed Adam learning rate is inactive. Each update uses the full sample, making the listed batch size inactive. Training stops after at most $40$ updates, or after at least $20$ epochs when the objective has not improved by a relative tolerance of $10^{-4}$ for $5$ epochs.

\paragraph{Uncertainty.}
The reported $0.95$ interval is the central Gaussian interval under $q(f_h(\mathbf{0}_m))$, conditional on the fixed hyperparameters. Local neighborhoods are disabled, so all observations inform this distribution and the listed neighborhood-size bounds are inactive.

Table~\ref{tab:hp_pmc_common} lists the resulting PMC settings shared by the baseline and supplementary benchmark suite; deviations are listed in Section~\ref{app:hp_deviations}.

\begin{table}[!htb]
\centering
\small
\setlength{\tabcolsep}{2.5pt}
\caption{Common PMC settings used in the baseline and supplementary benchmark suite.}
\label{tab:hp_pmc_common}
\input{tables/hp_pmc_common}
\end{table}

\subsection{Deviations from the Common Settings}\label{app:hp_deviations}

Within the baseline and supplementary benchmark suite, Table~\ref{tab:hp_deviations} lists the experiments whose PMC settings differ from the common ones, and Table~\ref{tab:hp_overrides} lists the settings that were tuned separately for individual cells; all other settings are as in Table~\ref{tab:hp_pmc_common}.

\begin{table}[!htb]
\centering
\small
\setlength{\tabcolsep}{2.5pt}
\caption{Baseline and supplementary experiments whose PMC settings differ from the common settings.}
\label{tab:hp_deviations}
\input{tables/hp_deviations}
\end{table}

\begin{table}[!htb]
\centering
\small
\setlength{\tabcolsep}{2.5pt}
\caption{Per-cell hyperparameter overrides.}
\label{tab:hp_overrides}
\input{tables/hp_overrides}
\end{table}

\subsection{Details of Main-Text Experiments and Diagnostics}\label{app:hp_main_text_additions}

\paragraph{Posterior predictive check.} Figure~\ref{fig:pmc_ppc_misspecification} uses $S=500$ independent paired draws from
\[
    X_1^{(s)}\sim\mathcal{N}\left(0,1\right),
    \qquad
    X_2^{(s)}\sim\mathcal{N}\left(2,4\right),
    \qquad s=1,\ldots,S.
\]
For each pair, we define
\[
    C^{(s)}=\frac{X_1^{(s)}+X_2^{(s)}}{2},
    \qquad
    D^{(s)}=X_2^{(s)}-X_1^{(s)},
    \qquad
    y_s=h\left(C^{(s)}\right),
    \quad h\left(u\right)=u^2.
\]
The true conditional mean is
\[
    f_h\left(d\right)
    =\E\left[h\left(C\right)\mid D=d\right]
    =0.96+0.24d+0.09d^2.
\]
We compare the misspecified linear mean with the well-specified quadratic mean:
\[
    f_{h,1}\left(d\right)=\beta_0+\beta_1d,
    \qquad
    f_{h,2}\left(d\right)=\beta_0+\beta_1d+\beta_2d^2.
\]
For $j\in\left\{1,2\right\}$, the Bayesian polynomial regression uses
\[
    y_s\mid\boldsymbol{\beta}_j,\sigma^2
    \sim\mathcal{N}\left(f_{h,j}\left(D^{(s)}\right),\sigma^2\right),
    \qquad
    p\left(\boldsymbol{\beta}_j,\sigma^2\right)\propto\frac{1}{\sigma^2}.
\]
We use $R=1000$ posterior-predictive draws and $B=10$ discrepancy bins. The data seed is $2026100301$, and the posterior seeds are $2026100302$ and $2026100303$ for the linear and quadratic models, respectively. Let $\hat f$ be the least-squares polynomial fit of the corresponding degree and $r_s$ its residual. After sorting observations by $D$, we divide them into $B$ near-equal bins $I_b$ and compute
\[
    r_s = y_s-\hat{f}\left(D^{(s)}\right),
    \qquad
    \bar r_b = \frac{1}{\left|I_b\right|}\sum_{s\in I_b}r_s,
    \qquad
    T =
    \frac{
        \sqrt{\sum_{b=1}^{B}\frac{\left|I_b\right|}{S}\bar r_b^2}
    }{
        \sqrt{\frac{1}{S}\sum_{s=1}^{S}r_s^2}
    }
    .
\]
For each posterior draw, we generate a replicated response dataset at the observed values $D^{(s)}$, refit the polynomial model of the same degree, and recompute $T$. We define the PPC value as
\[
    \mathrm{PPC}
    =\frac{1+\#\left\{T_{\mathrm{rep}}\geq T_{\mathrm{obs}}\right\}}{R+1}.
\]
The resulting values are
\[
    \mathrm{PPC}_{\mathrm{misspecified}}
    =\frac{1}{1001}\approx0.001,
    \qquad
    \mathrm{PPC}_{\mathrm{well\text{-}specified}}
    =\frac{540}{1001}\approx0.539.
\]

\paragraph{Dimension scaling.} Figure~\ref{fig:normalized_sum_dimension_scaling} uses
\[
    K=2,
    \qquad
    p\in\left\{1,10,100,1000\right\}.
\]
The component distributions are
\[
    \pi_1=\mathcal{N}\left(\mathbf{0}_p,2I_p\right),
    \qquad
    \pi_2=\mathcal{N}\left(2p^{-1/2}\mathbf{1}_p,2I_p\right).
\]
Their normalized product and the target are
\[
    \pipro=\mathcal{N}\left(p^{-1/2}\mathbf{1}_p,I_p\right),
    \qquad
    h(u)=p^{-1/2}\mathbf{1}_p^\top u,
    \qquad
    \mu_h=\E_{\pipro}[h(X)]=1,
\]
so the estimand is fixed as the dimension changes. The nominal per-component sample sizes are
\[
    S\in\left\{10^2,10^3,10^4,10^5,10^6\right\}.
\]
In each replicate, we draw a common realized sample size
\[
    N\sim\operatorname{Poisson}(S)
\]
and generate $N$ i.i.d. draws from each component. The paired samples yield
\[
    C^{(s)}=\frac{X_1^{(s)}+X_2^{(s)}}{2},
    \qquad
    D^{(s)}=X_1^{(s)}-X_2^{(s)},
    \qquad s=1,\ldots,N.
\]
Each $(p,S)$ cell has $20$ independent replicates. Points show arithmetic means across replicates, and error bars show one standard error. In replicate $r$, the squared error is
\[
    \left(\hat\mu_h^{(r)}-\mu_h\right)^2
    =\left(\hat\mu_h^{(r)}-1\right)^2.
\]
Runtime measures estimator fitting or aggregation and excludes data generation.

The available GBF sample sizes are
\[
\begin{aligned}
    p=1 &: \quad S\in\left\{10^2,10^3,10^4,10^5,10^6\right\},\\
    p=10 &: \quad S\in\left\{10^2,10^3,10^4,10^5\right\},\\
    p=100 &: \quad S\in\left\{10^2,10^3,10^4\right\}.
\end{aligned}
\]
GBF results are reported only for configurations completed within a 24-hour wall-time limit. Runs beyond the displayed sample-size ranges, including all $p=1000$ configurations, did not complete within this limit. For the $p=100$, $S=100$ cell only, the empirical full covariance matrices are singular or severely ill-conditioned because the realized sample size is comparable to the dimension, so GBF uses diagonal sample covariance matrices for both the fitted Gaussian components and the covariance preconditioners; all other GBF cells use full sample covariance matrices. The Oracle estimates each component's Gaussian parameters and combines the fitted components analytically.

\begin{table}[!htb]
\centering
\small
\setlength{\tabcolsep}{2.5pt}
\caption{Method settings for the dimension-scaling experiment.}
\label{tab:dimension_scaling_settings}
\begingroup
\renewcommand{\arraystretch}{0.8}
\begin{tabular}{ll}
\toprule
\multicolumn{2}{l}{\textit{PMC-OLS}} \\
Polynomial degree & 1 \\
\midrule
\multicolumn{2}{l}{\textit{PMC-MLP}} \\
Hidden layers & 64, 64, 64 \\
Activation & ReLU \\
Discrepancy centering & mean \\
Optimizer & Adam \\
Learning rate & $10^{-4}$ \\
Learning-rate schedule & cosine \\
Batch size & 1024 \\
Training steps & 4000 \\
Ensemble size & 1 \\
Loss & MSE \\
\midrule
\multicolumn{2}{l}{\textit{PMC-GP}} \\
Kernel & RBF \\
Inducing points & 64 \\
Learn inducing locations & No \\
Learn kernel hyperparameters & Yes \\
Maximum updates & 1600 \\
Minimum epochs & 10 \\
Early-stopping patience (epochs) & 3 \\
Early-stopping relative tolerance & $10^{-4}$ \\
Natural-gradient learning rate & 0.003 \\
Adam learning rate, kernel hyperparameters & 0.003 \\
Batch size & 8192 \\
Floating-point precision & 32-bit \\
\midrule
\multicolumn{2}{l}{\textit{GBF}} \\
Output particles & $S$ \\
Coalescence time & optimal rule \\
Time mesh & optimal rule \\
\bottomrule
\end{tabular}
\endgroup
\end{table}

\paragraph{Calibration.} Figure~\ref{fig:univariate_gaussian_k2_calibration} uses
\[
    K=2,
    \qquad
    \pi_1=\pi_2=\mathcal{N}\left(1,1\right).
\]

The consensus and discrepancy are
\[
    C=\frac{X_1+X_2}{2},
    \qquad
    D=X_1-X_2.
\]
The target is
\[
    h(u)=u,
    \qquad
    \mu_h=\E_{\pipro}[h]=1.
\]
The displayed experiment uses
\[
    S=1000,
    \qquad
    R=500
\]
independent replicates, each with exactly $S$ draws from each component. The base seed is $20261001$. The nominal coverage levels are
\[
    \left\{0.05,0.10,0.20,\ldots,0.90,0.95\right\}.
\]
For each method, nominal level $\alpha$, and replicate, the Gaussian interval is
\[
    I_{\alpha}
    =\left[
        \hat{\mu}_{h}
        -z_{\left(1+\alpha\right)/2}\widehat{\operatorname{se}},
        \hat{\mu}_{h}
        +z_{\left(1+\alpha\right)/2}\widehat{\operatorname{se}}
    \right].
\]
Its empirical coverage is
\[
    \widehat{\operatorname{coverage}}\left(\alpha\right)
    =\frac{\#\left\{\text{replicates for which }\mu_h\in I_{\alpha}\right\}}{R}.
\]
The method settings are listed in Table~\ref{tab:calibration_settings}.

\begin{table}[!htb]
\centering
\small
\setlength{\tabcolsep}{2.5pt}
\caption{Method settings for the calibration experiment.}
\label{tab:calibration_settings}
\begingroup
\renewcommand{\arraystretch}{0.8}
\begin{tabular}{ll}
\toprule
\multicolumn{2}{l}{\textit{PMC-GP}} \\
Kernel & Mat\'ern-3/2 \\
Inducing points & 64 \\
Learn inducing locations & No \\
Learn kernel hyperparameters & Yes \\
Maximum updates & 40 \\
Minimum epochs & 20 \\
Early-stopping patience (epochs) & 5 \\
Early-stopping relative tolerance & $10^{-4}$ \\
Natural-gradient learning rate & 0.03 \\
Adam learning rate, kernel hyperparameters & 0.03 \\
Batching & full batch \\
Floating-point precision & 64-bit \\
\midrule
\multicolumn{2}{l}{\textit{GBF}} \\
Output particles & $1000$, $10{,}000$, or $100{,}000$ \\
Coalescence time & optimal rule \\
Time mesh & optimal rule \\
Initial CESS threshold & 0.2 \\
Preconditioning & variance \\
Diffusion estimator & Poisson \\
ESS threshold & 0.5 \\
Number of cores & 1 \\
\bottomrule
\end{tabular}
\endgroup
\end{table}

For PMC-GP, $\widehat{\operatorname{se}}$ is the latent posterior standard deviation. For GBF, it is
\[
    \widehat{\operatorname{se}}_{\mathrm{GBF}}
    =\frac{\widehat{\operatorname{sd}}\left(\text{final particles}\right)}
    {\sqrt{\mathrm{ESS}_{\mathrm{final}}}},
\]
with the square root of the particle count used when the final ESS is unavailable.

\paragraph{GP posterior variance versus overlap.} Figure~\ref{fig:gp_overlap_posterior_variance} uses
\[
    K=2,
    \qquad
    \pi_1=\mathcal{N}\left(0,1\right),
    \qquad
    \pi_2=\mathcal{N}\left(\mu,1\right),
    \qquad
    \mu\in\left\{0,\ldots,6\right\}.
\]
The target and normalized product are
\[
    h(u)=u,
    \qquad
    \pipro=\mathcal{N}\left(\frac{\mu}{2},\frac{1}{2}\right),
    \qquad
    \mu_h=\E_{\pipro}[h]=\frac{\mu}{2}.
\]
The experiment uses
\[
    S=1000,
    \qquad
    R=50
\]
independent replicates, each with exactly $S$ draws from each component. The base seed is $26100500$. Within a replicate, common standard-normal draws are reused across all values of $\mu$:
\[
    Z_1^{(s)},Z_2^{(s)}\overset{\mathrm{iid}}{\sim}\mathcal{N}\left(0,1\right),
    \qquad
    X_{1,\mu}^{(s)}=Z_1^{(s)},
    \qquad
    X_{2,\mu}^{(s)}=Z_2^{(s)}+\mu,
    \quad s=1,\ldots,S.
\]
The consensus and discrepancy satisfy
\[
    C_\mu^{(s)}=\frac{X_{1,\mu}^{(s)}+X_{2,\mu}^{(s)}}{2},
    \qquad
    D_\mu^{(s)}=X_{1,\mu}^{(s)}-X_{2,\mu}^{(s)},
\]
and hence
\[
    C_\mu^{(s)}=C_0^{(s)}+\frac{\mu}{2},
    \qquad
    D_\mu^{(s)}=D_0^{(s)}-\mu.
\]
Within each replicate, one reference GP is fitted at $\mu=0$, and its learned hyperparameters are held fixed across mean separations. Therefore, physical zero discrepancy at mean separation $\mu$ corresponds to querying the reference GP at discrepancy $\mu$; the response translation changes the posterior mean but not its variance. Table~\ref{tab:gp_overlap_settings} gives the GP settings.

\begin{table}[!htb]
\centering
\small
\setlength{\tabcolsep}{2.5pt}
\caption{PMC-GP settings for the posterior-variance experiment.}
\label{tab:gp_overlap_settings}
\begingroup
\renewcommand{\arraystretch}{0.8}
\begin{tabular}{ll}
\toprule
\multicolumn{2}{l}{\textit{PMC-GP}} \\
Kernel & RBF \\
Inducing points & 64 \\
Learn inducing locations & No \\
Learn kernel hyperparameters & Yes \\
Maximum updates & 1600 \\
Minimum epochs & 10 \\
Early-stopping patience (epochs) & 3 \\
Early-stopping relative tolerance & $10^{-4}$ \\
Natural-gradient learning rate & 0.003 \\
Adam learning rate, kernel hyperparameters & 0.003 \\
Batch size & 8192 \\
Floating-point precision & 32-bit \\
\bottomrule
\end{tabular}
\endgroup
\end{table}

For each $\mu$, the plotted value is the median latent posterior variance across the $R$ replicates. This variance excludes observation noise, and the vertical axis uses a logarithmic scale.

\subsection{Fusion Methods and Oracle}\label{app:hp_fusion_oracle}

In the supplementary benchmark suite, GBF uses $100{,}000$ output particles for the Gaussian and Beta experiments and $S$ particles for the Gaussian mixture experiments. The fusion baseline comparison uses $100{,}000$ output particles and $20$ paired replications, with all methods evaluated on the same simulated component samples. The number of draws from each component is $\operatorname{Poisson}(S)$, where
\[
    S\in\left\{100,1000,10^4\right\},
\]
with $S=10^5$ also included in the univariate experiment. The base seeds are $20260826$, $20260831$, and $20260901$ for the univariate, bivariate, and multivariate experiments, respectively. Table~\ref{tab:hp_fusion} lists the remaining fusion-method settings, and Table~\ref{tab:hp_oracle} describes the Oracle for each component family.

\begin{table}[!htb]
\centering
\small
\setlength{\tabcolsep}{2.5pt}
\caption{Settings of the fusion methods in the simulation experiments.}
\label{tab:hp_fusion}
\begin{tabular}{p{0.34\linewidth} p{0.6\linewidth}}
\toprule
Setting & Value \\
\midrule
\multicolumn{2}{l}{\textit{GBF}} \\
Output particles & $100{,}000$ (Gaussian and Beta components); $S$ (Gaussian mixtures) \\
Coalescence time and time mesh & chosen by the SSH guidance of the DCFusion package, with a minimum conditional effective sample size of $0.2$ at initialization \\
Preconditioning & sample variance (univariate Gaussians) or sample covariance (bivariate and multivariate Gaussians) of each component \\
Resampling threshold & resample when the effective sample size falls below half of the number of particles \\
Resampling method & multinomial \\
Diffusion estimator & Poisson \\
Number of cores & 1 \\
Mixture components & each component is fitted by EM with the true number of modes, a full covariance, $3$ random restarts, at most $100$ iterations, a tolerance of $10^{-3}$, and a variance floor of $10^{-8}$ \\
Beta components & fit each component by fixed-support maximum likelihood, perform fusion in Jacobian-corrected logit coordinates using the auxiliary-logit sample variance for preconditioning, and transform the output back to $(0,1)$ \\
\midrule
\multicolumn{2}{l}{\textit{MCF and BF (fusion baseline comparison)}} \\
MCF coalescence time & $1$, with a single time step \\
BF coalescence time and time mesh & chosen by the same guidance as for GBF \\
BF preconditioning & identity \\
\bottomrule
\end{tabular}
\end{table}

\begin{table}[!htb]
\centering
\small
\setlength{\tabcolsep}{2.5pt}
\caption{Computation of the Oracle in the simulation experiments.}
\label{tab:hp_oracle}
\begin{tabular}{p{0.2\linewidth} p{0.74\linewidth}}
\toprule
Components & Oracle \\
\midrule
Gaussian & fit the mean and covariance of each component from its draws and combine them analytically into the Gaussian product \\
Beta & fit the two shape parameters of each component by maximum likelihood, with location $0$ and scale $1$ fixed, and combine them analytically into the Beta product \\
Gaussian mixture & fit each component by EM with the true number of modes, a full covariance, $10$ random restarts, at most $1000$ iterations, a tolerance of $10^{-8}$, and a variance floor of $10^{-8}$, and compute the product exactly \\
\bottomrule
\end{tabular}
\end{table}

\subsection{Real-Data Experiments}\label{app:hp_real_data}

In all three real-data experiments, we fit the logistic regression model
\[
Y_i\mid\beta\sim\operatorname{Bernoulli}\left(\operatorname{logit}^{-1}(X_i^\top\beta)\right),
\qquad
\beta_j\overset{\mathrm{ind}}{\sim}\mathcal{N}(0,2^2).
\]
The coefficient dimensions are $5$, $29$, and $11$ for respiratory, HIGGS, and MAGIC, respectively. We leave the intercept unchanged and standardize every other predictor using its full-data mean and standard deviation. The respiratory batches are the two observed clinical centers, containing $56$ and $55$ observations. HIGGS is partitioned into $10$ batches of $1.1$ million observations using seed $91011$, and MAGIC into $10$ batches of $1902$ observations using seed $91001$.

The Oracle is NUTS applied to the full data, and each batch is sampled using the same sampler. A batch of size $n_k$ uses prior scale $2/\sqrt{n_k/n}$, where $n$ is the full sample size, so that the product of the fractional priors recovers the full-data prior. Table~\ref{tab:hp_real_data} lists the data, sampler, and GBF settings. GBF was not run for HIGGS because it was computationally infeasible at this scale.

\begin{table}[!htb]
\centering
\footnotesize
\setlength{\tabcolsep}{4pt}
\caption{Settings of the real-data experiments.}
\label{tab:hp_real_data}
\begin{tabular}{lccc}
\toprule
Setting & Respiratory & HIGGS & MAGIC \\
\midrule
\multicolumn{4}{l}{\textit{Data}} \\
Observations $n$ & $111$ & $11{,}000{,}000$ & $19{,}020$ \\
Coefficients & $5$ & $29$ & $11$ \\
Number of batches $K$ & $2$ & $10$ & $10$ \\
\midrule
\multicolumn{4}{l}{\textit{Oracle and batch sampler (NUTS)}} \\
Prior scale & $2$ & $2$ & $2$ \\
Batch prior scale & $2.82$, $2.84$ & $6.32$ & $6.32$ \\
Chains & $4$ & $4$ & $4$ \\
Warmup iterations & $1000$ & $5000$ & $1000$ \\
Post-warmup draws per chain & $10{,}000$ & $10{,}000$ & $10{,}000$ \\
Mass matrix & diagonal & dense & dense \\
Initialization & default & MAP (up to $500$ iterations) & default \\
Precision & 32-bit & 64-bit & 32-bit \\
Seed & $91011$ & $91011$ & $91011$ \\
\bottomrule
\end{tabular}
\end{table}

\begin{table}[!htb]
\ContinuedFloat
\centering
\small
\setlength{\tabcolsep}{4pt}
\caption[]{Settings of the real-data experiments (continued).}
\begin{tabular}{lccc}
\toprule
Setting & Respiratory & HIGGS & MAGIC \\
\midrule
\multicolumn{4}{l}{\textit{GBF}} \\
Run & Yes & No & Yes \\
Particles per level & $1000$ & -- & $1000, 1000$ \\
Merge factors & $2$ & -- & $5, 2$ \\
Time mesh & $0.5, 1$ & -- & $0.5, 1$ \\
Preconditioning & Yes & -- & Yes \\
Resampling & residual & -- & residual \\
Resampling threshold (ESS fraction) & $0.5$ & -- & $0.5$ \\
Diffusion estimator & NB & -- & NB \\
Number of cores & $1$ & -- & $1$ \\
Seed & $93011$ & -- & $93011$ \\
\bottomrule
\end{tabular}
\end{table}

For PMC, we use all $S=40{,}000$ posterior draws from each batch and pair draws independently across batches. The consensus is the equally weighted average, the discrepancy consists of consecutive differences, and the target functions are
\[
h(u)=u,
\qquad
h(u)=u\odot u.
\]
Table~\ref{tab:hp_real_data_pmc} lists the PMC settings shared by the three experiments.

\subsection{Simulation-Based Inference}\label{app:hp_sbi}

Table~\ref{tab:hp_sbi} lists the settings of the HSP90 experiment. The protein binding angle is obtained from the latent parameter through $\theta=19\,\operatorname{logit}^{-1}(z)$. A single amortized NPE-C network is trained once and then conditioned separately on each of the ten observed images. Each image is independently standardized to zero mean and unit variance. The image-level posteriors use the fractional prior $\mathcal{N}(0,10)$ on the logit-scale parameter $z$, so that the product of the exact image-level posteriors would equal the posterior given all images under the prior $\mathcal{N}(0,1)$. PMC estimates the two targets $h(u)=u$ and $h(u)=u^2$. The reference combined posterior is the normalized product of the ten fitted NPE-C subposterior densities; its moments are computed using random-walk Metropolis and verified by numerical quadrature. The settings of PMC-GP are the recommended defaults of the implementation, which differ from the ones used in the simulations.

\subsection{Selection of the Hyperparameters}\label{app:hp_selection}

Experiment-level settings were selected manually rather than through an automated search. Where indicated in the tables, PMC-GP kernel hyperparameters were learned during fitting. We fixed one set of PMC settings for the baseline simulation suite (Section~\ref{app:hp_pmc_common}) and tuned individual cells separately in a few cases, namely the largest sample sizes in settings with 20 components, variance scaling, or variance imbalance, and the largest cell of the location and scale grid; each such override was compared with the common setting in a paired comparison over 20 or 30 replicates, and Table~\ref{tab:hp_overrides} lists them. For PMC-GP in this suite, the kernel hyperparameters are fixed rather than learned. The additional main-text and diagnostic studies use the settings documented in Section~\ref{app:hp_main_text_additions}. The real-data settings were selected manually, whereas HSP90 PMC-GP uses the implementation's recommended defaults.

The adaptive number of training steps of PMC-MLP grows with the square root of the sample size:
\[
    \text{steps} = \min\bigl(\max\bigl(\operatorname{round}\bigl(200\sqrt{S/100}\bigr),\, 200\bigr),\, 8000\bigr),
\]
where $S$ is the realized sample size of the simulated data, drawn from a Poisson distribution with mean equal to the nominal sample size. It equals 200 steps for $S\leq 100$ and reaches the cap of 8,000 steps at $S=160{,}000$.

\begin{table}[!htb]
\centering
\small
\setlength{\tabcolsep}{2.5pt}
\caption{PMC settings in the real-data experiments.}
\label{tab:hp_real_data_pmc}
\begin{tabular}{p{0.34\linewidth} p{0.6\linewidth}}
\toprule
Setting & Value \\
\midrule
\multicolumn{2}{l}{\textit{All PMC methods}} \\
Posterior draws per batch $S$ & $40{,}000$ \\
Sampling & independent pairing across batches \\
Consensus weights & equal ($1/K$) \\
Discrepancy construction & consecutive differences \\
Targets & $h(u)=u$ and $h(u)=u\odot u$ \\
Seeds & $94011$ (mean), $94012$ (square) \\
\midrule
\multicolumn{2}{l}{\textit{PMC-OLS}} \\
Polynomial degree & $1$ \\
Device & CPU \\
\midrule
\multicolumn{2}{l}{\textit{PMC-MLP}} \\
Hidden layers & $32, 32$ \\
Activation & ReLU \\
Optimizer & Adam \\
Learning rate & $10^{-3}$ \\
Learning-rate schedule & cosine \\
Batch size & $128$ \\
Training steps & $1000$ \\
Ensemble size & $5$ \\
Discrepancy centering & target \\
Loss & MSE \\
Device & GPU \\
\midrule
\multicolumn{2}{l}{\textit{PMC-GP}} \\
Kernel & Mat\'ern-3/2 (mean), RBF (square) \\
Inducing points & $8$, with learned locations \\
Kernel hyperparameters & learned \\
Batch size & $8192$ \\
Natural-gradient learning rate & $0.03$ \\
Adam learning rate & $0.03$ \\
Minimum epochs & $10$ \\
Early-stopping patience (epochs) & $3$ \\
Early-stopping relative tolerance & $10^{-4}$ \\
Floating-point precision & 32-bit \\
Device & CPU \\
\bottomrule
\end{tabular}
\end{table}

\begin{table}[!htb]
\centering
\footnotesize
\setlength{\tabcolsep}{2.5pt}
\caption{Settings of the HSP90 experiment.}
\label{tab:hp_sbi}
\begin{tabular}{p{0.34\linewidth} p{0.6\linewidth}}
\toprule
Setting & Value \\
\midrule
\multicolumn{2}{l}{\textit{Experiment}} \\
Number of images $K$ & $10$ \\
Posterior draws per image & $1{,}000$ \\
Parameter & $\theta = 19\,\mathrm{sigmoid}(z)$, with true value $\theta = 12$ \\
Image-level prior on $z$ & $\mathcal{N}(0,10)$ \\
Prior given all images & $\mathcal{N}(0,1)$ \\
\midrule
\multicolumn{2}{l}{\textit{Image simulator}} \\
Pose & $(0,0,0,90,90,0)$ \\
Defocus & $1500$ \\
Noise strength & $1$ \\
Random key and noise & independent for each image \\
\midrule
\multicolumn{2}{l}{\textit{Neural posterior}} \\
Method & NPE-C \\
Embedding network & convolutional \\
Density estimator & neural spline flow \\
Training simulations & $1{,}024$ \\
Training batch size & $64$ \\
Maximum epochs & $200$ (training stopped at epoch $57$) \\
Seeds & training $1000$, observation $20000$, posterior sampling $30000$ \\
\midrule
\multicolumn{2}{l}{\textit{All PMC methods}} \\
Parameter space & $z$ \\
Sampling & i.i.d. \\
Consensus weights & equal ($0.1$) \\
Discrepancy construction & consecutive differences \\
Seeds & $94011$ (mean), $94012$ (square) \\
\bottomrule
\end{tabular}
\end{table}

\begin{table}[!htb]
\ContinuedFloat
\centering
\footnotesize
\setlength{\tabcolsep}{2.5pt}
\caption[]{Settings of the HSP90 experiment (continued).}
\begin{tabular}{p{0.34\linewidth} p{0.6\linewidth}}
\toprule
Setting & Value \\
\midrule
\multicolumn{2}{l}{\textit{PMC-OLS}} \\
Polynomial degree & $1$ \\
\midrule
\multicolumn{2}{l}{\textit{PMC-MLP}} \\
Hidden layers & $32, 32$ \\
Activation & ReLU \\
Optimizer & Adam \\
Learning rate & $10^{-3}$ \\
Learning-rate schedule & cosine \\
Batch size & $128$ \\
Training steps & $1000$ \\
Ensemble size & $5$ \\
Discrepancy centering & target \\
Loss & MSE \\
\midrule
\multicolumn{2}{l}{\textit{PMC-GP}} \\
Kernel & Mat\'ern-3/2 (mean), RBF (square) \\
Inducing points & $8$, with learned locations \\
Kernel hyperparameters & learned \\
Batch size & $8192$ \\
Natural-gradient learning rate & $0.03$ \\
Adam learning rate & $0.03$ \\
Minimum epochs & $10$ \\
Early-stopping patience (epochs) & $3$ \\
Early-stopping relative tolerance & $10^{-4}$ \\
Floating-point precision & 32-bit \\
\midrule
\multicolumn{2}{l}{\textit{Reference combined posterior}} \\
Sampler & random-walk Metropolis on $z$ \\
Chains & $4$ \\
Burn-in & $2000$ iterations, with the step size adapted toward an acceptance rate of $0.44$ \\
Draws per chain & $5000$ \\
Seed & $40000$ \\
Check & quadrature on a grid of $20{,}001$ points on $[-2,3]$ \\
\bottomrule
\end{tabular}
\end{table}

\FloatBarrier

%% file: tables/magic_identity_comparison.tex
\begin{tabular}{lccc}
\toprule
Method & Estimation error & Merge (min) & Total (min) \\
\midrule
Oracle & 0.0000 & -- & 0.9936 \\
Batch & -- & -- & 0.8302 \\
GBF & 0.0475 & 27.2566 & 28.0868 \\
PMC-OLS & 0.0676 & 0.0050 & 0.8352 \\
PMC-MLP & 0.0101 & 0.1152 & 0.9455 \\
PMC-GP & 0.0091 & 0.5056 & 1.3359 \\
\bottomrule
\end{tabular}

%% file: tables/hp_pmc_common.tex
\begin{tabular}{ll}
\toprule
\multicolumn{2}{l}{\textit{PMC-OLS}} \\
Polynomial degree & 1 \\
\midrule
\multicolumn{2}{l}{\textit{PMC-MLP}} \\
Hidden layers & 64, 64, 64 \\
Activation & ReLU \\
Ensemble size & 2 \\
Optimizer & Adam \\
Weight decay & 0 \\
Learning rate, mean target & $10^{-4}$ \\
Learning rate, square target & $3\times10^{-4}$ \\
Learning-rate schedule & cosine \\
Batch size & 1024 \\
Training steps & adaptive (Section~\ref{app:hp_selection}) \\
Loss & MSE \\
\midrule
\multicolumn{2}{l}{\textit{PMC-GP}} \\
Kernel & Mat\'ern-3/2 \\
Inducing points & 64 \\
Maximum updates & 40 \\
Minimum epochs & 20 \\
Early-stopping patience (epochs) & 5 \\
Early-stopping relative tolerance & $10^{-4}$ \\
Natural-gradient learning rate & 0.03 \\
Adam learning rate, kernel hyperparameters & 0.03 \\
Batch size & 8192 \\
Full-batch updates & Yes \\
Learn inducing locations & No \\
Fixed kernel hyperparameters & Yes \\
Fixed mean & $-0.0014$ \\
Fixed lengthscale & 6 \\
Fixed output scale & 0.5 \\
Fixed noise variance & 0.64 \\
Floating-point precision & 64-bit \\
Prediction interval level & 0.95 \\
Minimum local neighborhood size & 256 \\
Maximum local neighborhood size & 131072 \\
Local neighborhood exponent & None (full sample) \\
\midrule
\multicolumn{2}{l}{\textit{All PMC methods}} \\
Sampling & i.i.d. \\
Consensus weights & uniform \\
Discrepancy centering & mean \\
Discrepancy construction & consecutive differences \\
\bottomrule
\end{tabular}

%% file: tables/hp_deviations.tex
\begin{tabular}{p{0.4\linewidth} l l l}
\toprule
Experiments & Setting & Common value & Value used \\
\midrule
\raggedright Multivariate Gaussian, dimension ($p=2,3,5,10,20$) & Kernel & Mat\'ern-3/2 & RBF \\
 & Fixed mean & $-0.0014$ & $0.0000$ \\
 & Fixed lengthscale & 6 & 3 \\
\bottomrule
\end{tabular}

%% file: tables/hp_overrides.tex
\begin{tabular}{p{0.2\linewidth} l p{0.24\linewidth} p{0.3\linewidth}}
\toprule
Experiment & Method & Cell & Changed settings \\
\midrule
\raggedright Univariate Gaussian, location and scale & PMC-MLP & \raggedright $S=2{,}000{,}000$, $\mu_2=4$, $\sigma_2=1$, square & \raggedright learning rate $6\times10^{-4}$, training steps 16,000 \tabularnewline
 & PMC-MLP & \raggedright $S=2{,}000{,}000$, $\mu_2=4$, $\sigma_2=2$, square & \raggedright hidden layers 256, 256, 256, learning rate $2\times10^{-4}$, batch size 4096 \tabularnewline
 & PMC-GP & \raggedright $S=2{,}000{,}000$, $\mu_2=4$, $\sigma_2=1$, identity & \raggedright fixed lengthscale 16 \tabularnewline
 & PMC-GP & \raggedright $S=2{,}000{,}000$, $\mu_2=4$, $\sigma_2=1$, square & \raggedright fixed lengthscale 20 \tabularnewline
\midrule
\raggedright Univariate Gaussian, number of components, homogeneous ($K=20$) & PMC-GP & \raggedright $S=10^{6}$, $K=20$, identity & \raggedright kernel RBF, fixed lengthscale 16 \tabularnewline
 & PMC-GP & \raggedright $S=10^{6}$, $K=20$, square & \raggedright kernel RBF, fixed lengthscale 16 \tabularnewline
\midrule
\raggedright Univariate Gaussian, number of components, heterogeneous ($K=20$) & PMC-MLP & \raggedright $S=10^{6}$, identity & \raggedright hidden layers 256, 256, 256, activation SiLU, learning rate $2\times10^{-4}$, batch size 8192 \tabularnewline
 & PMC-GP & \raggedright $S=10^{6}$, identity & \raggedright kernel RBF, inducing points 128, fixed lengthscale 4, fixed noise variance 1 \tabularnewline
\midrule
\raggedright Bivariate Gaussian, isotropic variance scale & PMC-MLP & \raggedright $S=10^{6}$, $\sigma=3$, square & \raggedright hidden layers 256, 256, 256, activation SiLU, learning rate $2\times10^{-4}$, batch size 8192 \tabularnewline
 & PMC-MLP & \raggedright $S=10^{6}$, $\sigma=4$, square & \raggedright hidden layers 256, 256, 256, activation SiLU, learning rate $3.5\times10^{-4}$, batch size 8192 \tabularnewline
 & PMC-GP & \raggedright $S=10^{6}$, $\sigma=3$, square & \raggedright kernel RBF, fixed mean $0.0000$, fixed lengthscale 2 \tabularnewline
 & PMC-GP & \raggedright $S=10^{6}$, $\sigma=4$, square & \raggedright fixed mean $0.0000$, fixed lengthscale 4, fixed noise variance 1 \tabularnewline
\midrule
\raggedright Bivariate Gaussian, marginal variance imbalance & PMC-GP & \raggedright $S=10^{6}$, $\sigma_2=4$, square & \raggedright kernel RBF, fixed mean $0.0000$, fixed lengthscale 2 \tabularnewline
\bottomrule
\end{tabular}